\documentclass[11pt]{article}
\usepackage[letterpaper, margin=1in]{geometry}
\usepackage{times}
\usepackage{amsmath,amssymb,amsthm}
\usepackage{graphicx}
\graphicspath{{figures/}{./}}
\usepackage{float}
\usepackage{booktabs}
\usepackage{enumitem}
\usepackage{bm}
\usepackage{setspace}
\usepackage[unicode]{hyperref}

\newtheorem{theorem}{Theorem}[section]
\newtheorem{corollary}[theorem]{Corollary}
\newtheorem{lemma}[theorem]{Lemma}
\newtheorem{proposition}[theorem]{Proposition}
\newtheorem{definition}[theorem]{Definition}
\newtheorem{remark}[theorem]{Remark}
\newtheorem{example}[theorem]{Example}

\DeclareMathOperator{\Tr}{Tr}
\DeclareMathOperator{\ad}{ad}
\DeclareMathOperator{\Aut}{Aut}
\DeclareMathOperator{\spn}{span}
\newcommand{\R}{\mathbf{R}}
\newcommand{\Fro}[1]{\lVert #1 \rVert_F}

\title{Matched Generators for the Karhunen--Lo\`eve Transform:\\
A Double-Commutator Eigenvalue Theory}

\author{Mitchell A. Thornton\\[2pt]
\small ORCID 0000-0003-3559-9511\\
\small \texttt{mitchat@sbcglobal.net}}
\date{}

\begin{document}
\maketitle

\begin{abstract}
The Karhunen--Lo\`eve transform (KLT) diagonalizes the covariance operator of a second-order process and is the optimal orthonormal expansion for mean-square truncation. Which classical transform the KLT reduces to is governed by the symmetry commutant of the covariance: when the kernel commutes with a group action, the KLT eigenfunctions are the irreducible representation functions of that group, recovering the Fourier, cosine, Mellin, and spherical-harmonic systems. We study the inverse question. Given a covariance operator $R$ and a finite-dimensional space of candidate generators, the generator nearest to commuting with $R$, the minimizer of the commutativity residual $\delta(A,R)=\Fro{[R,A]}/(\Fro{R}\,\Fro{A})$, is the smallest-eigenvalue solution of a double-commutator eigenvalue problem $\ad_R^2(A^\ast)=\lambda A^\ast$, which reduces to a Hermitian positive-semidefinite generalized eigenvalue problem of dimension equal to the number of generators, independent of the ambient dimension. The framework recovers hidden transforms as well as classical ones: a variational characterization turns the existence of a commuting generator into a spectral condition, and a tridiagonal commutant-uniqueness result yields the prolate spheroidal, cosine, and discrete orthogonal-polynomial bases as exact recoveries, admits matrix-valued extensions and local obstruction results, and produces a continuum of transforms that interpolates between and beyond the classical families. When the symmetry is only approximate, the high-resolution coding penalty of the symmetry-adapted blockwise transform is exactly the multi-information among the symmetry sectors, giving an exact marginal selection threshold between the fixed group transform and the data-driven KLT. We further give an eigenvalue-displacement formula and a graph-automorphism characterization for permutation structure, a sequential deflation that, under a generating-completeness hypothesis, recovers non-Abelian symmetry, and stability bounds for the recovered generator under covariance estimation error. As an application, we synthesize the Karhunen--Lo\`eve transform of a two-paradigm covariance directly from its two known generators, without forming the mixed covariance, reaching the full-data transform's energy compaction from a handful of observations.
\end{abstract}

\section{Introduction}\label{sec:intro}

The Karhunen--Lo\`eve transform occupies a central position in harmonic analysis and second-order stochastic process theory. Given a process with covariance kernel $R$, the KLT expands the process in the eigenfunctions of the associated integral operator, achieving the fastest possible mean-square decay of truncation error among all orthonormal systems. The transform itself is data-dependent: its basis is dictated by $R$, with no closed form in general. Yet for the kernels that arise most often in practice the KLT collapses onto a \emph{classical} transform, the Fourier basis for stationary processes, the cosine basis for even-stationary processes, the Mellin basis for scale-invariant processes, the spherical harmonics for isotropic processes on the sphere, and this collapse is not a coincidence of special functions but a consequence of symmetry. Each of these kernels commutes with a group action, and the KLT eigenfunctions are forced, by Schur's lemma, to be the irreducible representation functions of that group.

This paper takes the symmetry viewpoint as primary and studies the inverse problem. Rather than asking for the eigenbasis of a given kernel, we ask which symmetry a kernel possesses, or most nearly possesses. Concretely, fix a covariance operator $R$ and a finite-dimensional space $\mathcal{B}$ of candidate generators (elements of a Lie algebra, shifted permutations, or graph-structured operators). We seek the generator $A^\ast\in\mathcal{B}$ that minimizes the commutativity residual
\begin{equation}\label{eq:delta-intro}
\delta(A,R)=\frac{\Fro{[R,A]}}{\Fro{R}\,\Fro{A}},\qquad [R,A]=RA-AR,
\end{equation}
which vanishes exactly when $A$ commutes with $R$, that is, when the one-parameter group $\exp(\theta A)$ is an exact symmetry of $R$. The minimizer identifies the algebraic structure of $R$ and hence the classical transform onto which the KLT collapses.

The technical core is a single variational principle. The residual minimization \eqref{eq:delta-intro} is the Rayleigh--Ritz problem for the self-adjoint double-commutator superoperator $\ad_R^2\colon A\mapsto[R,[R,A]]$ restricted to $\mathcal{B}$, and its solution is the smallest-eigenvalue generalized eigenvector of a $d\times d$ Hermitian positive-semidefinite pencil $\mathbf{M}\mathbf{c}=\lambda\mathbf{G}\mathbf{c}$, where $d=\dim\mathcal{B}$ is independent of the ambient dimension $M$ (Section~\ref{sec:gevp}). The pencil eigenvalues are real and nonnegative, the smallest vanishing exactly when $\mathcal{B}$ contains an exact symmetry (Section~\ref{sec:spectrum}); the two notions of optimality at play, the nearest commuting generator and the best energy-compacting transform, are reconciled in Section~\ref{sec:optimality}.

The same principle recovers transforms, not merely symmetries. When the pencil kernel is nontrivial it returns the exact infinitesimal symmetry algebra, including a metaplectic generator on the time--frequency plane (Section~\ref{sec:gevp}). Beyond the classical Fourier and cosine systems, a variational characterization converts the existence of a hidden commuting operator into a spectral condition, and a tridiagonal commutant-uniqueness result makes the recovery exact: the prolate spheroidal basis, a discrete orthogonal-polynomial ensemble, and the discrete cosine transform are recovered as the unique tridiagonal-commuting transforms of their covariances (Sections~\ref{sec:prolate}, \ref{sec:hahn}, \ref{sec:dct}), the Haar wavelet basis appears as a numerical illustration from a hierarchical generator (Section~\ref{sec:wavelet}), and the construction extends toward a matrix-valued setting (Section~\ref{sec:matrix}). Convex combinations of generators produce a continuum of transforms between the classical ones (Section~\ref{sec:interp}), and the same continuum reaches transforms that are certifiably outside the sinusoidal family, witnessed by the triviality of the tridiagonal commutant (Section~\ref{sec:interpjain}).

When the symmetry is only approximate the theory becomes quantitative. The isotypic refinement and the strong-versus-weak symmetry distinction (Section~\ref{sec:isotypic}) yield a blockwise group transform whose high-resolution coding penalty is exactly the multi-information among the symmetry sectors, with an exact threshold selecting between the fixed group transform and the data-driven KLT. The symplectic and reflection-group families place the construction on the time--frequency plane (Section~\ref{sec:symplectic}). For discrete structure, the residual equals the total squared eigenvalue displacement $\sum_k(\lambda_k-\lambda_{\sigma(k)})^2$ for permutation generators, and $\delta=0$ characterizes graph automorphisms for graph-diffusion covariance (Section~\ref{sec:pcf}); a sequential deflation extends the single-generator theorem to non-Abelian symmetry with strict subgroup growth, a logarithmic iteration bound, and, under a generating-completeness hypothesis, recovery of the full automorphism group (Section~\ref{sec:sequential}). Finally, the residual and the recovered generator are stable under covariance estimation error, with a quadratic residual bound at a true symmetry and a resulting sample-complexity estimate (Section~\ref{sec:perturbation}). These results have a direct application in data-scarce settings: when a covariance mixes two paradigms with known generators, the matched transform is synthesized from the two generators alone, without forming the mixed covariance, and attains the full-data Karhunen--Lo\`eve compaction from a handful of observations, on synthetic and public real-world signals (Section~\ref{sec:application}).

Section~\ref{sec:discussion} returns to the finite-dimensional blind group matching problem, of which the continuum theory is the generalization, and to the connection with isospectral and double-bracket flows. None of the spectral results requires positivity of $R$; they hold for any Hermitian (in the continuum, self-adjoint) operator, though the covariance case is the one of interest.

This work is part of a broader program. The algebraic diversity paradigm for signal and data processing, from which the present results emerged, is described in~\cite{thornton2026algebraicdiversityprinciplesgrouptheoretic}. The polynomial-time group-selection formula that the double-commutator pencil generalizes to the continuum is developed in~\cite{thornton2026polynomialtimeoptimalgroupselection}, and the unification of the classical signal transforms under the matched-group principle is treated in~\cite{thornton2026unificationsignaltransformtheory}. The finite-dimensional blind group matching problem and its single-observation variance reduction appear in companion work~\cite{thornton2026algebraicdiversitygrouptheoreticspectral,thornton2026continuousalgebraicdiversityunifying}, discussed further in Section~\ref{sec:discussion}.

\subsection{Scope and contributions}\label{sec:contributions}

A portion of the development below is classical and is included for completeness, for review, and to keep the paper self-contained. The link between the Karhunen--Lo\`eve transform and the symmetry commutant of a covariance, the recovery of the Fourier, cosine, Mellin, and spherical-harmonic systems from group invariance, and the isotypic decomposition of a group-invariant covariance are standard material in representation-theoretic harmonic analysis; they are reviewed in Sections~\ref{sec:ckl} and~\ref{sec:commutant} to fix notation and viewpoint. The contributions of this paper concern the inverse direction, recovering a structured generator, and hence a transform, from a covariance operator. They are the following.

\begin{enumerate}
\item \emph{An inverse formulation and its variational solution.} We pose the recovery of a covariance's symmetry generator as the minimization of the commutativity residual $\delta$ over a finite-dimensional candidate space, and show it is the Rayleigh--Ritz problem for the double-commutator superoperator $\ad_R^2$, solved by the smallest generalized eigenvector of a $d\times d$ Hermitian positive-semidefinite pencil whose size $d$ is independent of the ambient dimension $M$ (Section~\ref{sec:gevp}).

\item \emph{A spectral characterization of exact symmetry.} We show the pencil eigenvalues are real and nonnegative, and that the smallest vanishes exactly when the candidate space contains a generator that commutes with $R$, turning the detection of an exact symmetry into a spectral condition (Section~\ref{sec:spectrum}).

\item \emph{Reconciliation of two notions of optimality.} We reconcile the nearest-commuting-generator criterion with the classical energy-compaction optimality of the Karhunen--Lo\`eve transform, separating the regime in which the matched group transform and the data-driven transform coincide from the regime in which they diverge (Section~\ref{sec:optimality}).

\item \emph{Stability under estimation error.} We show the residual and the recovered generator are stable when the covariance is replaced by a sample estimate, with the residual of a true symmetry second-order in the perturbation while that of a false candidate is first-order, and we give a resulting sample-complexity estimate (Section~\ref{sec:perturbation}).

\item \emph{From symmetry recovery to transform recovery.} We prove that when the covariance is an injective function of an irreducible Jacobi matrix, the symmetric-tridiagonal commutant is exactly $\spn\{\mathbf{I},J\}$, so the pencil returns the orthogonal-polynomial generator and its Karhunen--Lo\`eve transform rather than a mere residual minimizer (Sections~\ref{sec:prolate}, \ref{sec:hahn}). This uniqueness is the mechanism that upgrades symmetry recovery to transform recovery.

\item \emph{Exact recovery of hidden transforms.} We show the discrete prolate spheroidal (Slepian) basis (Section~\ref{sec:prolate}), the Hahn orthogonal-polynomial transform (Section~\ref{sec:hahn}), and the discrete cosine transform (Section~\ref{sec:dct}) are recovered exactly as finite-dimensional instances; the Haar wavelet basis appears from a hierarchical generator as a numerical illustration (Section~\ref{sec:wavelet}); and the construction extends toward a matrix-valued setting (Section~\ref{sec:matrix}).

\item \emph{The isotypic refinement and a blockwise group transform.} Recasting and extending classical material, we factor the transform of a group-invariant covariance into a fixed symmetry-adapted part and a data-dependent eigendecomposition confined to the multiplicity blocks, distinguish strong from weak symmetry, and obtain a blockwise group transform whose high-resolution coding penalty equals the multi-information among the symmetry sectors, with an exact threshold selecting between the fixed group transform and the data-driven transform (Section~\ref{sec:isotypic}).

\item \emph{Symmetry on the time--frequency plane.} We place the construction on the time--frequency plane through the symplectic and reflection-group families, recovering a metaplectic generator and connecting to the fractional Fourier and linear canonical transforms (Section~\ref{sec:symplectic}).

\item \emph{A permutation commutator formula and graph automorphisms.} For permutation generators we show the residual equals the total squared eigenvalue displacement $\sum_k(\lambda_k-\lambda_{\sigma(k)})^2$, and that $\delta=0$ characterizes the graph automorphisms of a graph-diffusion covariance (Section~\ref{sec:pcf}).

\item \emph{Sequential recovery for non-Abelian symmetry.} We extend the single-generator theorem to non-Abelian symmetry by sequential deflation, with strict subgroup growth, a logarithmic iteration bound, and, under a generating-completeness hypothesis, recovery of the full automorphism group (Section~\ref{sec:sequential}).

\item \emph{A continuum of transforms between the classical ones.} We show that convex combinations of generators produce a continuum of transforms interpolating the classical ones, including transforms certifiably outside the sinusoidal family, witnessed by the triviality of the tridiagonal commutant (Sections~\ref{sec:interp}, \ref{sec:interpjain}).

\item \emph{An application to data-scarce transform synthesis.} We apply the construction to the two-paradigm setting, synthesizing the matched transform from two known generators without forming the mixed covariance, and show on synthetic and public real-world signals that it attains the full-data Karhunen--Lo\`eve compaction from a handful of observations (Section~\ref{sec:application}).
\end{enumerate}

\subsection{Related work}\label{sec:related}

Existing transform theories, including algebraic signal processing, graph Fourier analysis, wavelet and time-frequency constructions, and bispectral and orthogonal-polynomial theory, derive transforms from a prescribed algebraic, geometric, or operator model. In contrast, the present work addresses the inverse problem: given a covariance operator, recover structured generators whose eigenspaces explain the observed Karhunen--Lo\`eve basis. The double-commutator pencil provides a variational mechanism for discovering hidden transform families directly from second-order statistics. We position the construction against the most closely related lines below; the connection to isospectral and double-bracket flows and to supervised group-discovery methods is taken up again in Section~\ref{sec:discussion}.

\emph{Algebraic signal processing.} The algebraic signal processing theory of P\"uschel and Moura~\cite{puschel2008,puschelmoura2008space,puschelmoura2008cooleytukey} models a signal as a module over a filter algebra generated by a prescribed shift, and derives the Fourier, cosine, and sine transforms, with their fast algorithms, as the Fourier transforms attached to that algebra. This is the forward, model-first counterpart of the commutant table of Section~\ref{sec:commutant}; the present work runs in the opposite direction, recovering candidate generators, and hence transform families, from a covariance operator, with the residual $\delta$ quantifying the departure from any exact algebraic model.

\emph{Time-band limiting and commuting operators.} The prolate spheroidal theory of Slepian, Landau, and Pollak~\cite{slepianpollak1961,landaupollak1961,slepian1978} proceeds from a prescribed time-and-band-limiting operator to a commuting tridiagonal operator with a well-separated spectrum, whose existence was established by direct construction~\cite{grunbaumlonghiperlstadt1982}. Section~\ref{sec:prolate} reverses the arrow: given only the covariance, the pencil recovers the commuting tridiagonal operator with no prior assumption that one exists, an inverse recovery of a hidden commuting operator rather than its construction from a known concentration problem.

\emph{Bispectrality.} The Jacobi and Hahn results of Sections~\ref{sec:prolate} and~\ref{sec:hahn} sit close to the bispectral program of Duistermaat and Gr\"unbaum~\cite{duistermaat1986} and its development by Gr\"unbaum, Haine, Iliev, and coworkers~\cite{grunbaumhaine1997,haineiliev2000,grunbaumyakimov2002,grunbaumyakimov2004,grunbaumvinetzhedanov2018}, in which an orthogonal-polynomial family is paired with a commuting differential or difference operator and the prolate phenomenon is read as a consequence of bispectrality. That theory begins with the polynomial family or the operator; the present theorem begins with the covariance, recovering $J$ from $R=f(J)$ for injective $f$, an inverse viewpoint. The matrix-valued setting of Section~\ref{sec:matrix} is adjacent to matrix bispectral and time-band-limiting theory~\cite{duran2004,grunbaumpacharonizurrian2020,casperyakimov2022}.

\emph{Wavelets and multiresolution.} Wavelet bases are usually built forward from a prescribed dilation-translation (affine) structure or multiresolution analysis~\cite{haar1910,mallat1989,daubechies1992}. In Section~\ref{sec:wavelet} the Haar basis instead appears as the eigenbasis of a hierarchical generator recovered from the covariance, the affine structure read off the data rather than imposed.

\emph{Time-frequency analysis and the fractional Fourier transform.} The fractional Fourier and linear canonical transforms arise from a prescribed metaplectic action on the time--frequency plane~\cite{namias1980}. Here the metaplectic generator is recovered from the covariance as the matched symmetry on that plane (Section~\ref{sec:symplectic}).

\emph{Reflection groups and Dunkl operators.} A process invariant under a finite Coxeter group is diagonalized by the Dunkl transform attached to that prescribed reflection group~\cite{dunkl1989}; here the reflection-group structure is instead matched to the covariance (Section~\ref{sec:symplectic}).

\emph{Graph signal processing.} Graph Fourier analysis~\cite{shuman2013,sandryhailamoura2013} builds a transform forward from a prescribed graph operator, a Laplacian or adjacency shift, by eigendecomposition. The present construction is again inverse: it asks which structured operator a covariance most nearly commutes with, and recovers that operator rather than assuming it.

\emph{Graph automorphism and spectral graph theory.} Graph isomorphism and automorphism algorithms~\cite{weisfeilerleman1968,babai2016,godsilroyle2001,biggs1993} act on the combinatorial graph. Section~\ref{sec:pcf} recovers the automorphism group instead as the exact-commutativity locus of a graph-diffusion covariance, through the same pencil, with a permutation commutator formula linking the residual to eigenvalue displacement.

\emph{Representation-theoretic harmonic analysis.} The isotypic refinement of Section~\ref{sec:isotypic} is standard material in the harmonic analysis of finite and compact groups~\cite{serre1977,diaconis1988,terras1999,folland1995}, and the commutant and Schur-decomposition tools it uses are also central to the study of symmetry in quantum information~\cite{ragone2023}. We leverage these results rather than extend them; the new content is the variational recovery of the generator.

\emph{Joint diagonalization.} Joint approximate diagonalization seeks a single basis that nearly diagonalizes a given set of matrices~\cite{cardoso1996,pham2001} and underlies common principal component analysis~\cite{flury1984}. The present problem is different: a single covariance $R$ is given and the structured generator $A$ is unknown and to be recovered, an inverse problem $R\mapsto A$ rather than a common eigenbasis of several prescribed operators.

\emph{Isospectral and double-bracket flows.} The double commutator $\ad_R^2$ also drives the double-bracket flow toward diagonal form against a fixed target~\cite{brockett1991,chu2008,chudriessel1990,helmkemoore1994}. Here it is used statically: a Rayleigh--Ritz problem on a fixed candidate subspace that names the matched symmetry of a fixed $R$, rather than a flow that diagonalizes it. The full comparison is given in Section~\ref{sec:discussion}.

\emph{Operator-theoretic data analysis.} Data-driven operator methods such as dynamic mode decomposition and Koopman analysis~\cite{schmid2010,williamskevrekidisrowley2015} also pass from data to an operator to an eigenstructure, but they recover the dynamics that generate a time series; here the recovered operator is a symmetry generator of a covariance, and its eigenbasis is the transform.

\section{The Continuous Karhunen--Lo\`eve Transform}\label{sec:ckl}

Let $f$ be a zero-mean second-order process on a domain $\mathcal{T}\subseteq\mathbb{R}^n$ with continuous covariance kernel $R(s,t)=\mathbb{E}[f(s)\overline{f(t)}]$. The kernel defines a covariance operator $\mathcal{R}\colon L^2(\mathcal{T})\to L^2(\mathcal{T})$,
\begin{equation}\label{eq:covop}
[\mathcal{R}\,\varphi](s)=\int_{\mathcal{T}} R(s,t)\,\varphi(t)\,dt .
\end{equation}
Because $R$ is Hermitian and positive semidefinite, $\mathcal{R}$ is a compact self-adjoint operator; by Mercer's theorem~\cite{mercer1909} and the Hilbert--Schmidt theory~\cite{courant1953} it has a countable nonnegative spectrum $\lambda_1\ge\lambda_2\ge\cdots\ge 0$ with orthonormal eigenfunctions $\{\varphi_k\}\subset L^2(\mathcal{T})$ satisfying the Fredholm equation
\begin{equation}\label{eq:fredholm}
\int_{\mathcal{T}} R(s,t)\,\varphi_k(t)\,dt=\lambda_k\,\varphi_k(s).
\end{equation}

\begin{definition}[Continuous Karhunen--Lo\`eve transform]\label{def:ckl}
The continuous Karhunen--Lo\`eve transform is the unitary map $\mathcal{K}\colon L^2(\mathcal{T})\to \ell^2$ sending $f$ to the sequence of coefficients
\begin{equation}\label{eq:klt-analysis}
c_k \;=\; \langle f,\varphi_k\rangle \;=\; \int_{\mathcal{T}} f(t)\,\overline{\varphi_k(t)}\,dt ,
\end{equation}
where $\{\varphi_k\}$ are the orthonormal eigenfunctions of the covariance operator \eqref{eq:covop}, that is, the solutions of the integral equation \eqref{eq:fredholm}. The synthesis $f=\sum_k c_k\varphi_k$ converges in $L^2(\mathcal{T})$, and the coefficients $c_k$ are uncorrelated with $\mathbb{E}|c_k|^2=\lambda_k$.
\end{definition}

This integral-operator definition of the optimal expansion first appeared in Kosambi~\cite{kosambi1943} and was independently established by Karhunen~\cite{karhunen1947} and Lo\`eve~\cite{loeve1978}; the finite-dimensional principal-axis construction that it generalizes is due to Pearson~\cite{pearson1901} and Hotelling~\cite{hotelling1933}, the latter in the form now called the Hotelling transform or principal component analysis. We take \eqref{eq:fredholm}--\eqref{eq:klt-analysis} as classical and reserve all novelty for the symmetry theory built on top of it.

The optimality that distinguishes \eqref{eq:klt-analysis} among all orthonormal expansions is the following standard fact.

\begin{theorem}[Optimality of the continuous KLT]\label{thm:klt-opt}
Among all orthonormal systems $\{\psi_k\}_{k=1}^{K}$ in $L^2(\mathcal{T})$, the KLT eigenfunctions maximize the captured energy:
\begin{equation}
\sum_{k=1}^{K}\langle\mathcal{R}\,\varphi_k,\varphi_k\rangle=\sum_{k=1}^{K}\lambda_k\;\ge\;\sum_{k=1}^{K}\langle\mathcal{R}\,\psi_k,\psi_k\rangle
\end{equation}
for every $K\ge 1$.
\end{theorem}

\begin{proof}
For an orthonormal system $\{\psi_k\}_{k=1}^{K}$, set $V=\spn\{\psi_1,\ldots,\psi_K\}$ and let $P_V$ be the orthogonal projection onto $V$, so that $\sum_{k=1}^{K}\langle\mathcal{R}\psi_k,\psi_k\rangle=\operatorname{tr}(P_V\mathcal{R})$. The Ky Fan maximum principle for the compact self-adjoint operator $\mathcal{R}$ states that
\[
\sum_{k=1}^{K}\lambda_k=\max_{\dim V=K}\operatorname{tr}(P_V\mathcal{R}),
\]
with the maximum attained at $V=\spn\{\varphi_1,\ldots,\varphi_K\}$~\cite{fan1949}. The stated inequality is exactly this principle, with equality on the KLT eigenfunctions.
\end{proof}

Closed-form eigenfunctions of \eqref{eq:fredholm} are known for several canonical kernels: the Wiener kernel $R(s,t)=\min(s,t)$ yields sines, the Ornstein--Uhlenbeck (first-order Markov) kernel yields the transcendental sinusoids whose discrete counterpart is the cosine basis~\cite{rao1974,jain1979}, and the Brownian bridge yields a shifted sine system. For a general kernel no closed form exists, and the symmetry theory of the next sections provides the structural information that the eigenfunctions alone do not.

\section{The Symmetry Commutant of the Covariance Operator}\label{sec:commutant}

The classical transforms arise from the KLT through symmetry. Let $G$ be a group acting on $L^2(\mathcal{T})$ by a unitary representation $\pi$. We say $R$ is $G$-invariant when $\mathcal{R}$ lies in the commutant of the representation,
\begin{equation}\label{eq:commute}
[\pi(g),\mathcal{R}]=0\qquad\text{for all } g\in G .
\end{equation}

\begin{proposition}[Commutant forces the eigenbasis]\label{prop:schur}
Let $R$ be $G$-invariant, $[\pi(g),\mathcal{R}]=0$ for all $g\in G$, and suppose $\pi$ is \emph{multiplicity-free}, so that each irreducible of $G$ occurs at most once. By Schur's lemma $\mathcal{R}$ is then scalar on each isotypic component of $\pi$, so the KLT eigenspaces are unions of isotypic components and any representation-adapted basis is a KLT eigenbasis. How completely the symmetry fixes that basis depends on the irreducibles:
\begin{enumerate}[leftmargin=2.4em]
\item[(i)] \emph{(Abelian.)} If $G$ is Abelian, all irreducibles are one-dimensional, $\mathcal{R}$ is diagonal in the character basis, and the KLT eigenfunctions are the characters of $G$, unique up to phase.
\item[(ii)] \emph{(Non-Abelian.)} For an irreducible $\rho$ of dimension $n_\rho>1$, $\mathcal{R}=\lambda_\rho I_{n_\rho}$ on its isotypic component; the representation basis is a KLT eigenbasis, but $\lambda_\rho$ is $n_\rho$-fold degenerate and the eigenbasis within the component is fixed only up to a rotation in $U(n_\rho)$.
\end{enumerate}
\end{proposition}

\begin{proof}
Multiplicity-freeness makes each isotypic component a single copy of its irreducible, so the restriction of $\mathcal{R}$ to the $\rho$-component commutes with an irreducible action and is therefore scalar by Schur's lemma, $\lambda_\rho I$; the eigenspace for $\lambda_\rho$ is exactly that component. If $n_\rho=1$ the component is one-dimensional and the eigenfunction is the character; if $n_\rho>1$ the eigenvalue is $n_\rho$-fold degenerate and every orthonormal basis of the component diagonalizes $\mathcal{R}$, so the representation basis is one admissible choice in $U(n_\rho)$.
\end{proof}

The standard correspondences are instances of Proposition~\ref{prop:schur}.

\begin{center}
\small
\begin{tabular}{lll}
\toprule
\textbf{Covariance structure} & \textbf{Commuting group} & \textbf{KLT eigenfunctions (transform)} \\
\midrule
\multicolumn{3}{l}{\emph{Abelian: eigenbasis unique up to phase}}\\
stationary, $R(s-t)$ & translation $(\mathbb{R},+)$ & $e^{i\omega t}$ (Fourier) \\
periodic, $R(s-t \bmod N)$ & cyclic $\mathbb{Z}_N$ & DFT vectors \\
even stationary, $R(|s-t|)$ & reflection $\mathbb{Z}_2$ & $\cos\omega t$ (cosine, DCT) \\
scale invariant, $R(\lambda s,\lambda t){=}R(s,t)$ & dilation $(\mathbb{R}^+,\times)$ & $t^{i\sigma}$ (Mellin) \\
stationary on the circle & rotation $SO(2)$ & $e^{in\theta}$ (Fourier series) \\
\addlinespace
\multicolumn{3}{l}{\emph{Non-Abelian: eigenbasis unique up to rotation within each irreducible}}\\
isotropic on the sphere, $R(\hat n_1\!\cdot\!\hat n_2)$ & rotation $SO(3)$ & $Y_\ell^m$ (spherical harmonics) \\
circularly symmetric on the disk & rotation $SO(2)\ltimes$ radial & $e^{in\theta}J_n(kr)$ (Fourier--Bessel) \\
exchangeable (compound symmetry) & symmetric $S_M$ & $\mathbf{1}\oplus$ contrasts (Helmert) \\
\bottomrule
\end{tabular}
\end{center}

For stationary processes the cyclic (Fourier) case is exact only asymptotically: the discrete Fourier transform diagonalizes the circulant matrix, which is asymptotically equivalent to the Toeplitz covariance as the dimension grows~\cite{gray2006}. When $R$ commutes \emph{exactly} with a group, the KLT and the group transform coincide; the substantive case, and the subject of the remainder of the paper, is when $R$ commutes with no standard group, and one seeks the symmetry that $R$ most nearly possesses.

The two horizontal bands of the table are the two regimes of Proposition~\ref{prop:schur}: for an Abelian group the characters are forced, whereas for a non-Abelian group the symmetry fixes the eigenspaces, the isotypic components, but leaves a unitary rotation free within each irreducible, because $\mathcal{R}$ is scalar there. Two further structures extend the table beyond the exact commutant. A self-similar process with stationary increments, such as fractional Brownian motion, is invariant under the affine ``$ax+b$'' group, and the continuous wavelet transform approximately decorrelates it, the wavelet analogue of the Fourier diagonalization of a stationary process~\cite{wornell1993}; the decorrelation is asymptotic, so this row lies outside the exact table. A process invariant under a finite Coxeter (reflection) group is diagonalized by the Dunkl transform~\cite{dunkl1989}, which generalizes the cosine row (the rank-one $\mathbb{Z}_2$ case) to higher-rank reflection groups.

The rows above use groups acting on the domain $\mathcal{T}$. A further family acts on the time--frequency plane and is developed in Section~\ref{sec:symplectic}: the symplectic group, whose metaplectic representation gives the fractional Fourier and linear canonical transforms, and the finite reflection groups, whose Dunkl transform generalizes the cosine row.

The multiplicity-free hypothesis of Proposition~\ref{prop:schur} is lifted in Section~\ref{sec:isotypic}, where the general (multiplicity) case gives the isotypic refinement of the KLT and connects to the weak-symmetry condition.

\begin{example}[Mixed-structure kernel]\label{ex:mixed}
On $\mathcal{T}=[-1,1]$ consider
\begin{equation}\label{eq:mixed}
R(s,t)=e^{-\alpha|s-t|}\bigl(1+\beta\,st\bigr),\qquad \alpha>0,\ \beta\ge 0 .
\end{equation}
For $\beta=0$ this is the Ornstein--Uhlenbeck kernel with asymptotically cosine eigenfunctions. For $\beta>0$ the factor $1+\beta st$ depends on absolute position, breaking stationarity, and no standard transform is exactly optimal. The kernel is, however, persymmetric ($R(-s,-t)=R(s,t)$), so it commutes with the reflection $s\mapsto-s$; the symmetry theory below recovers this reflection as the exact matched generator from $R$ alone. Since the reflection squares to the identity, it yields an even/odd parity split rather than a full sinusoidal transform; recovering the cosine basis itself would require the stronger even-stationary (Toeplitz-plus-Hankel) structure of Section~\ref{sec:dct}, which the $\beta>0$ term breaks. The parity split is not evident from the eigenfunctions, which have no closed form for $\beta>0$.
\end{example}

Figure~\ref{fig:ckl} shows the continuous KLT of this kernel for $\alpha=2$, $\beta=0.5$: the covariance, its leading eigenfunctions and eigenvalue decay, and the residuals $\delta$ of the candidate translation, reflection, and dilation generators, with reflection alone attaining $\delta=0$.

\begin{figure}[H]
\centering
\includegraphics[width=0.96\textwidth]{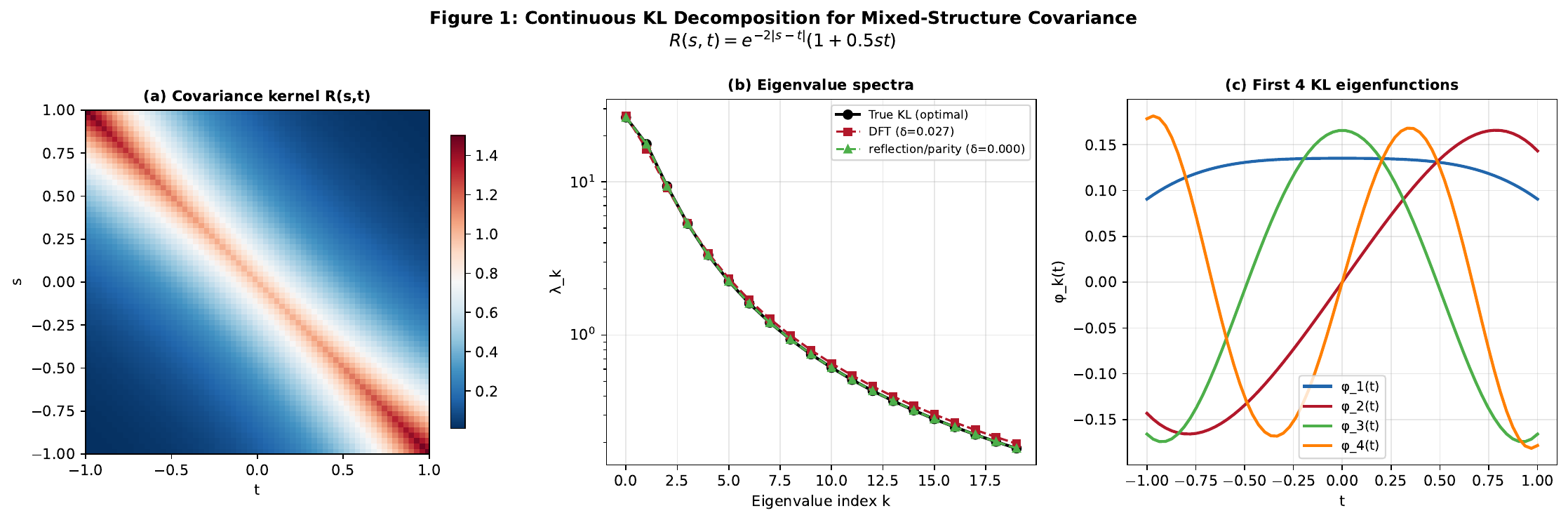}
\caption{Continuous KLT of the mixed kernel \eqref{eq:mixed} with $\alpha=2$, $\beta=0.5$: the covariance surface, the leading eigenfunctions of the integral operator \eqref{eq:fredholm}, the eigenvalue decay, and the generator residuals $\delta(A,R)$ for the translation, reflection, and dilation generators. The reflection residual is zero because the kernel is persymmetric.}
\label{fig:ckl}
\end{figure}

\section{Multiplicity, the Isotypic Refinement, and Weak Symmetry}\label{sec:isotypic}

Proposition~\ref{prop:schur} assumed each irreducible occurs without multiplicity. The general case is the complete statement, and it is the one met in practice: it is where the matched transform ceases to be fully determined by the group and acquires a small data-dependent part.

\subsection{The isotypic refinement}

Let $\pi$ decompose the space into the isotypic sum $\mathcal{H}=\bigoplus_\rho(\mathbb{C}^{m_\rho}\otimes V_\rho)$, where $V_\rho$ is the irreducible of dimension $n_\rho$ and $m_\rho$ its multiplicity.

\begin{theorem}[Isotypic structure of an invariant covariance]\label{thm:isotypic}
If $R$ is $G$-invariant then, in a symmetry-adapted basis,
\begin{equation}\label{eq:isotypic}
R=\bigoplus_\rho\bigl(M_\rho\otimes I_{n_\rho}\bigr),
\end{equation}
with $M_\rho$ a Hermitian $m_\rho\times m_\rho$ matrix, the \emph{reduced multiplicity covariance} of the $\rho$-isotypic component. The KLT eigenvalues are those of the $M_\rho$, each repeated $n_\rho$ times, and the KLT eigenvectors are $\{v\otimes e:\ v\ \text{an eigenvector of }M_\rho,\ e\in V_\rho\}$. The matched KLT is therefore obtained by (i) projecting $R$ onto each isotypic component and (ii) diagonalizing the $m_\rho\times m_\rho$ matrices $M_\rho$; the eigenbasis is unique up to the eigenbasis of each $M_\rho$ (unique when $M_\rho$ has simple spectrum) and an arbitrary unitary basis of each $V_\rho$.
\end{theorem}

\begin{proof}
By Schur's lemma the space of $G$-equivariant endomorphisms of $\mathbb{C}^{m_\rho}\otimes V_\rho$ is $\mathrm{End}(\mathbb{C}^{m_\rho})\otimes I_{n_\rho}$, and a $G$-invariant operator carries nothing between distinct isotypic summands; hence $R$ has the form~\eqref{eq:isotypic}. The spectrum of $M_\rho\otimes I_{n_\rho}$ is that of $M_\rho$ with each eigenvalue repeated $n_\rho$ times, with eigenvectors $v\otimes e$. The multiplicity-free case $m_\rho\equiv1$ returns Proposition~\ref{prop:schur}, with $M_\rho$ a scalar.
\end{proof}

\begin{remark}[Reynolds projection and the isotypic pipeline]\label{rem:reynolds}
The projection in step~(i) is the Reynolds (group-averaging) operator onto the commutant, applied componentwise; in the algebraic-diversity setting it is the isotypic (Schur) separation that splits an observation into its symmetry sectors, after which each reduced covariance $M_\rho$ is diagonalized separately. Theorem~\ref{thm:isotypic} states that this two-stage procedure, group projection followed by a small eigendecomposition, is exactly the KLT when $R$ is invariant; it is the multiplicity-correct version of ``apply the group transform.''
\end{remark}

\begin{proposition}[Factorization of the matched KLT]\label{prop:factor}
For a $G$-invariant $R$ the KLT factors as $U=F_G\,D$, where $F_G$ is the symmetry-adapted transform that block-diagonalizes the $G$-action into isotypic components and depends only on $G$, and $D=\bigoplus_\rho(W_\rho\otimes I_{n_\rho})$ with $W_\rho$ the eigenbasis of $M_\rho$. The data dependence of the KLT is confined to the multiplicity blocks $W_\rho$; the remainder is the fixed group transform. The factor $F_G$ is the simultaneous block-diagonalization of the representation, the isotypic (Schur) decomposition, and $D$ is the eigendecomposition of $R$ within each block; the departure of $R$ from commuting with the group, which no symmetry-adapted basis can remove, is the commutativity residual $\delta$ of Section~\ref{sec:gevp}. For an Abelian group the generators commute and $F_G$ reduces to an ordinary joint diagonalization of $R$ with them~\cite{cardoso1996}.
\end{proposition}

\begin{proof}
$F_G$ block-diagonalizes $R$ to $\bigoplus_\rho(M_\rho\otimes I_{n_\rho})$ by Theorem~\ref{thm:isotypic}, and diagonalizing each $M_\rho$ by $W_\rho$ supplies the residual rotation $D$. The basis $F_G$ simultaneously block-diagonalizes the generators $\pi(g)$ into their isotypic blocks, which is all a non-commuting family admits, and $D$ then diagonalizes $R$ within each block, so $F_GD$ diagonalizes $R$ while leaving the generators block-diagonal. The mass measuring the failure of $R$ to commute with the representation, and hence to share the group's block structure, is $\sum_g\|[\pi(g),R]\|_F^2$ up to normalization, the residual of Section~\ref{sec:gevp}.
\end{proof}

\subsection{Weak symmetry and partial commutation}

The invariance $[\pi(g),R]=0$ is, in the terminology of Doat and Vanrietvelde~\cite{doat2026can}, the \emph{weak} symmetry condition: $R$ lies in the commutant, equivalently $R$ is block-diagonal in the charge (isotypic) sectors, equivalently $R$ is fixed by the weak twirl $\mathcal{T}_W(R)=\tfrac{1}{|G|}\sum_g\pi(g)R\pi(g)^\dagger$. The \emph{strong} condition restricts $R$ to the trivial (zero-charge) sector alone. Theorem~\ref{thm:isotypic} is the harmonic-analysis content of the weak condition: weak symmetry is exactly the block structure~\eqref{eq:isotypic} that makes the KLT group-aligned, and the weak twirl is the Reynolds projection of Remark~\ref{rem:reynolds}.

For an Abelian symmetry with charge generator $C$ and sectors of charge $\lambda_\rho$ (so $\pi(g)=e^{igC}$ and $C=\bigoplus_\rho\lambda_\rho P_\rho$), the residual decomposes by sector.

\begin{proposition}[Sectorwise residual]\label{prop:sector}
With $C=\bigoplus_\rho\lambda_\rho P_\rho$,
\begin{equation}\label{eq:sector}
\|[C,R]\|_F^2=\sum_{\rho\neq\rho'}(\lambda_\rho-\lambda_{\rho'})^2\,\|P_\rho R P_{\rho'}\|_F^2 .
\end{equation}
The residual is the off-sector mass weighted by squared charge differences; it vanishes iff $R$ is block-diagonal, i.e.\ weakly symmetric.
\end{proposition}

\begin{proof}
The $(\rho,\rho')$ block of $[C,R]$ is $(\lambda_\rho-\lambda_{\rho'})P_\rho R P_{\rho'}$; squaring and summing the Frobenius norms gives~\eqref{eq:sector}. This is the eigenvalue-difference representation of Section~\ref{sec:pcf} applied to the diagonal generator $C$.
\end{proof}

Equation~\eqref{eq:sector} makes precise a graded, partial notion of symmetry. Say $R$ is \emph{partially $G$-invariant on a subspace} $\mathcal{V}$ if $\mathcal{V}$ is $G$-invariant and $P_\mathcal{V}RP_\mathcal{V}$ commutes with the restricted action, while $R$ itself need not. On such a $\mathcal{V}$ the KLT eigenfunctions are the group transform's, and the residual~\eqref{eq:sector} measures the leakage off $\mathcal{V}$, sector by sector. The strong-symmetry projection onto a single sector is the limiting case in which all but one block is discarded.

Proposition~\ref{prop:schur} states that an exact commutant forces the group eigenbasis. The same conclusion holds approximately: when a covariance only nearly commutes with the group action, its KLT eigenspaces stay close to the symmetry-adapted eigenspaces forced by $G$, with a deviation controlled by the commutativity defect $\|R-\mathcal{T}_G(R)\|$ and the spectral gap.

\begin{theorem}[Approximate symmetry implies approximate KLT reduction]\label{thm:approx-klt}
Let $R$ be a compact self-adjoint covariance operator on a Hilbert space $\mathcal{H}$, and let $G$ be a compact group with unitary representation $\pi$. Let
\[
\mathcal{T}_G(R)=\int_G \pi(g)R\pi(g)^\ast\,dg
\]
be the Reynolds projection of $R$ onto the commutant of $G$, and set $R_0=\mathcal{T}_G(R)$, $E=R-R_0$. Assume $R_0$ has an isolated spectral cluster $\Lambda$ separated from the rest of $\sigma(R_0)$ by a gap $\eta=\operatorname{dist}(\Lambda,\sigma(R_0)\setminus\Lambda)>0$. Let $P_0$ be the spectral projector of $R_0$ for $\Lambda$ and $P$ the corresponding spectral projector of $R$. If $\|E\|<\eta/2$, then
\begin{equation}\label{eq:approx-klt}
\|P-P_0\|\le\frac{2\|R-\mathcal{T}_G(R)\|}{\eta}.
\end{equation}
\end{theorem}

\begin{proof}
The operator $R_0=\mathcal{T}_G(R)$ commutes with $\pi(g)$ for every $g\in G$, so its spectral subspaces are symmetry-adapted. Writing $R=R_0+E$ with $E=R-\mathcal{T}_G(R)$ and applying the Davis--Kahan $\sin\theta$ theorem to the isolated cluster of $R_0$ separated by $\eta$, with $\|E\|<\eta/2$, gives $\|P-P_0\|\le 2\|E\|/\eta$. Substituting $E=R-\mathcal{T}_G(R)$ gives~\eqref{eq:approx-klt}.
\end{proof}

The bound is global, controlling the whole spectral cluster at once. The sectorwise refinement, which selects \emph{which} blocks to treat by the group transform, is the content of the blockwise result below.

\subsection{The blockwise group transform: eigenspaces and coding gain}

Write the covariance as $R=R_0+E$ with block-diagonal part $R_0=\sum_\rho P_\rho R P_\rho$ and off-sector part $E=\sum_{\rho\ne\rho'}P_\rho R P_{\rho'}$. The \emph{blockwise group transform} $U_G$ uses, on each sector, the eigenvectors of $A_\rho=P_\rho R P_\rho$; it is the KLT of $R_0$ and is forced by Schur's lemma to be symmetry-adapted. Two questions arise. First, how close are its sector eigenspaces to those of the true KLT of $R$? Second, what is the exact coding penalty for using it in place of the optimal KLT? The first has a small-leakage answer through Davis--Kahan; the second has an exact answer, valid at any leakage, through a determinant identity.

\begin{theorem}[Blockwise eigenspace closeness and truncation-energy discrepancy]\label{thm:blockwise-optimal}
Let $\mathcal H=\bigoplus_{\rho\in\widehat G}\mathcal H_\rho$ be the isotypic decomposition with projections $P_\rho$, let $R$ be a compact positive self-adjoint covariance, and write $R=R_0+E$ as above. Fix a sector $\rho$, set $A_\rho=P_\rho R P_\rho$, let $\Lambda_\rho=\{\lambda_1(A_\rho),\dots,\lambda_{k_\rho}(A_\rho)\}$ collect its top $k_\rho$ eigenvalues, and let $Q_\rho$ be the corresponding rank-$k_\rho$ spectral projector of $A_\rho$ (extended by zero outside $\mathcal H_\rho$). Let
\begin{equation}\label{eq:blockwise-gap}
\eta_\rho=\operatorname{dist}\bigl(\Lambda_\rho,\,\operatorname{spec}(R_0)\setminus\Lambda_\rho\bigr)>0
\end{equation}
be the separation of this cluster from the rest of the block-diagonal spectrum $\operatorname{spec}(R_0)=\bigcup_{\rho'}\operatorname{spec}(A_{\rho'})$, taken across all sectors and not only within $\rho$, and let $\ell_\rho=\|P_\rho R(I-P_\rho)\|$ be the off-sector leakage. If $\|E\|<\eta_\rho/2$, the eigenvalue cluster of $R$ perturbed from $\Lambda_\rho$ is well-defined, and, writing $\widetilde Q_\rho$ for its rank-$k_\rho$ spectral projector,
\begin{equation}\label{eq:blockwise-proj}
\|\widetilde Q_\rho-Q_\rho\|\le\frac{2\ell_\rho}{\eta_\rho},
\end{equation}
and the rank-$k_\rho$ truncation-energy discrepancy of the group-adapted projector obeys the second-order bound
\begin{equation}\label{eq:blockwise-energy}
\bigl|\operatorname{tr}(R\widetilde Q_\rho)-\operatorname{tr}(RQ_\rho)\bigr|\le\bigl(\lambda_{\max}(R)-\lambda_{\min}(R)\bigr)\sum_{j=1}^{k_\rho}\sin^2\theta_j\le 4k_\rho\bigl(\lambda_{\max}(R)-\lambda_{\min}(R)\bigr)\frac{\ell_\rho^2}{\eta_\rho^2},
\end{equation}
where $\theta_1,\dots,\theta_{k_\rho}$ are the principal angles between $\operatorname{range}(Q_\rho)$ and $\operatorname{range}(\widetilde Q_\rho)$; the discrepancy is quadratic in the relative leakage $\ell_\rho/\eta_\rho$. When $\Lambda_\rho$ comprises the $k_\rho$ largest eigenvalues of $R_0$, so that $\widetilde Q_\rho$ is the dominant rank-$k_\rho$ eigenspace of $R$, the difference is moreover nonnegative and the group-adapted truncation can only lose energy.
\end{theorem}

\begin{proof}
Track the cluster $\Lambda_\rho$ of the block-diagonal $R_0=\bigoplus_{\rho'}A_{\rho'}$ under the off-block perturbation $E=R-R_0$. By hypothesis $\Lambda_\rho$ is separated from $\operatorname{spec}(R_0)\setminus\Lambda_\rho$ by $\eta_\rho$ and $\|E\|<\eta_\rho/2$, so Weyl's inequality keeps the $k_\rho$ eigenvalues of $R$ perturbed from $\Lambda_\rho$ isolated from the rest of the spectrum and $\widetilde Q_\rho$ is well-defined. The Davis--Kahan $\sin\Theta$ theorem in residual form gives $\|\widetilde Q_\rho-Q_\rho\|\le\|EQ_\rho\|/(\eta_\rho-\|E\|)$; since $Q_\rho=P_\rho Q_\rho$ and $EP_\rho=(I-P_\rho)RP_\rho$, we have $\|EQ_\rho\|\le\|EP_\rho\|=\|(I-P_\rho)RP_\rho\|=\ell_\rho$, while $\eta_\rho-\|E\|>\eta_\rho/2$, and~\eqref{eq:blockwise-proj} follows. For the energy bound, work in an eigenbasis of $R$ and let $S$ index the cluster, so $\widetilde Q_\rho$ projects onto $\{e_i:i\in S\}$ and $\operatorname{tr}(R\widetilde Q_\rho)=\sum_{i\in S}\lambda_i$. Writing $w_i=(Q_\rho)_{ii}\in[0,1]$ with $\sum_i w_i=k_\rho$,
\[
\operatorname{tr}(R\widetilde Q_\rho)-\operatorname{tr}(RQ_\rho)=\sum_{i\in S}\lambda_i(1-w_i)-\sum_{i\notin S}\lambda_i w_i,
\]
and, each $\lambda_i$ lying in $[\lambda_{\min}(R),\lambda_{\max}(R)]$ and each weight in $[0,1]$, the right side has absolute value at most $(\lambda_{\max}(R)-\lambda_{\min}(R))\,\sigma$, where $\sigma=\sum_{i\notin S}w_i=\sum_{i\in S}(1-w_i)=\sum_j\sin^2\theta_j$ is the mass of $\operatorname{range}(Q_\rho)$ outside the cluster. If $\Lambda_\rho$ is the dominant rank-$k_\rho$ cluster of $R_0$, then $S$ indexes the top $k_\rho$ eigenvalues of $R$ and $\widetilde Q_\rho$ maximizes $\operatorname{tr}(R\,\cdot)$ among rank-$k_\rho$ projectors, so the difference is nonnegative. Finally, the principal-angle sum is the squared Frobenius distance of the two projectors, $\sum_j\sin^2\theta_j=\tfrac12\|\widetilde Q_\rho-Q_\rho\|_F^2$, whereas~\eqref{eq:blockwise-proj} controls the largest angle, $\|\widetilde Q_\rho-Q_\rho\|=\max_j\sin\theta_j$. Because $\widetilde Q_\rho-Q_\rho$ has rank at most $2k_\rho$, the conversion $\|\cdot\|_F^2\le(\operatorname{rank})\,\|\cdot\|^2$ supplies the rank factor explicitly,
\[
\sum_j\sin^2\theta_j=\tfrac12\|\widetilde Q_\rho-Q_\rho\|_F^2\le k_\rho\,\|\widetilde Q_\rho-Q_\rho\|^2\le\frac{4k_\rho\,\ell_\rho^2}{\eta_\rho^2},
\]
which is~\eqref{eq:blockwise-energy}. The constant $k_\rho$ is the truncation rank, the cost of passing from the single largest principal angle that the operator-norm Davis--Kahan bound controls to the sum of all $k_\rho$ of them in the trace functional; it is intrinsic to measuring leakage in operator norm rather than swept into the inequality. Measured in Frobenius norm the prefactor disappears: with the Frobenius leakage $\ell_\rho^{(F)}=\|(I-P_\rho)RP_\rho\|_F$, the Frobenius $\sin\Theta$ theorem gives $\sum_j\sin^2\theta_j\le(\ell_\rho^{(F)})^2/\eta_\rho^2$ up to the same Davis--Kahan constant, whence $|\operatorname{tr}(R\widetilde Q_\rho)-\operatorname{tr}(RQ_\rho)|\le(\lambda_{\max}(R)-\lambda_{\min}(R))\,(\ell_\rho^{(F)})^2/\eta_\rho^2$, the rank then residing in the Frobenius leakage.
\end{proof}

\begin{corollary}[Linear energy bound]\label{cor:blockwise-linear}
Under the hypotheses of Theorem~\ref{thm:blockwise-optimal}, the second-order bound~\eqref{eq:blockwise-energy} implies the linear bound
\begin{equation}\label{eq:blockwise-linear}
\bigl|\operatorname{tr}(R\widetilde Q_\rho)-\operatorname{tr}(RQ_\rho)\bigr|\le\frac{4k_\rho\|R\|\,\ell_\rho}{\eta_\rho}.
\end{equation}
\end{corollary}

\begin{proof}
By~\eqref{eq:blockwise-energy}, using $\lambda_{\max}(R)-\lambda_{\min}(R)\le\|R\|$ for positive $R$ and $\ell_\rho/\eta_\rho<1$,
\[
\bigl|\operatorname{tr}(R\widetilde Q_\rho)-\operatorname{tr}(RQ_\rho)\bigr|\le 4k_\rho\bigl(\lambda_{\max}(R)-\lambda_{\min}(R)\bigr)\frac{\ell_\rho^2}{\eta_\rho^2}\le 4k_\rho\|R\|\frac{\ell_\rho^2}{\eta_\rho^2}\le\frac{4k_\rho\|R\|\,\ell_\rho}{\eta_\rho}.\qedhere
\]
\end{proof}

\begin{corollary}[Energy-domain block selection rule]\label{cor:block-selection}
Suppose a transform-coding or estimation objective charges a complexity cost $c_\rho>0$ for handling sector $\rho$ with a data-driven KLT block, traded against distortion through the Lagrangian $\mathcal J=\text{distortion}+\lambda\,\text{complexity}$. By~\eqref{eq:blockwise-linear} the distortion penalty $\Delta_\rho$ of the fixed group-adapted block satisfies $\Delta_\rho\le 4k_\rho\|R\|\ell_\rho/\eta_\rho$, so the group-adapted block lowers $\mathcal J$ whenever
\begin{equation}\label{eq:block-selection}
\frac{4k_\rho\|R\|\,\ell_\rho}{\eta_\rho}\le\lambda c_\rho,
\qquad\text{equivalently}\qquad
\ell_\rho\le\frac{\lambda c_\rho\,\eta_\rho}{4k_\rho\|R\|};
\end{equation}
otherwise the data-driven KLT block is justified. This is a sufficient energy-domain criterion and requires the cluster separation $\eta_\rho$ of Theorem~\ref{thm:blockwise-optimal}; the exact criterion, valid at any leakage and with no gap hypothesis, is the multi-information threshold of Theorem~\ref{thm:rd-threshold}.
\end{corollary}

Equation~\eqref{eq:blockwise-energy} controls the eigenspaces. The coding penalty of the whole transform is governed by a determinant, and there the answer is exact.

\begin{theorem}[Exact rate-distortion gain of the blockwise group transform]\label{thm:rd-gain}
Assume high-resolution scalar transform coding with optimal bit allocation and marginal entropy coding. Let $R\succ0$ on $\mathcal H=\bigoplus_\rho\mathcal H_\rho$, let $R_0=\sum_\rho P_\rho R P_\rho$ be its block-diagonal part, and let $U_G$ be the blockwise group transform. Under this model the distortion of an orthonormal transform $U$ at fixed rate is proportional to the geometric mean of its coordinate variances $\bigl(\prod_i(URU^\ast)_{ii}\bigr)^{1/M}$, which the KLT of $R$ minimizes~\cite{huang1963,goyal2001}. The multiplicative distortion penalty of $U_G$ relative to the optimal KLT is exactly
\begin{equation}\label{eq:rd-gain}
G=\Bigl(\frac{\det R_0}{\det R}\Bigr)^{1/M}=\exp\!\Bigl(\frac{2C}{M}\Bigr)\ge 1,\qquad C=\tfrac12\Bigl(\sum_\rho\log\det A_\rho-\log\det R\Bigr),
\end{equation}
where $C$ is the Gaussian multi-information (total correlation) among the isotypic sectors, with equality if and only if $R$ is weakly symmetric ($E=0$). Equivalently, at fixed distortion the excess coding rate of $U_G$ over the rate-distortion bound is exactly $C$ nats, independent of the rate.
\end{theorem}

\begin{proof}
Because $U_G$ is block-diagonal across the sectors and $E$ is off-block, $U_G E U_G^\ast$ is off-block and has zero diagonal, so the transform-domain variances are $(U_G R U_G^\ast)_{ii}=(U_G R_0 U_G^\ast)_{ii}=\lambda_i(R_0)$ and $\prod_i(U_G R U_G^\ast)_{ii}=\det R_0$. By Hadamard's inequality the KLT attains the minimum $\prod_i(URU^\ast)_{ii}=\det R$ over orthonormal $U$, so the ratio of geometric means is $(\det R_0/\det R)^{1/M}$. Fischer's inequality gives $\det R_0=\prod_\rho\det A_\rho\ge\det R$, with equality iff $E=0$, so $G\ge1$. For jointly Gaussian sectors of covariance $R$ the multi-information is $C=\sum_\rho h(X_\rho)-h(X)=\tfrac12\log(\det R_0/\det R)$ with $h$ the differential entropy~\cite{covertthomas2006}, whence $G=\exp(2C/M)$ and the excess rate is $C$.
\end{proof}

The penalty resolves per sector into canonical correlations, which is where the off-sector mass of Proposition~\ref{prop:sector} reappears information-theoretically.

\begin{proposition}[Per-sector gain via canonical correlations]\label{prop:canon}
For a single sector $\rho$ against its complement, write $R=\left(\begin{smallmatrix}A_\rho & B_\rho\\ B_\rho^\ast & C_\rho\end{smallmatrix}\right)$ and let $\varrho_1,\dots,\varrho_m$ be the canonical correlations between the sector and its complement, the singular values of $A_\rho^{-1/2}B_\rho C_\rho^{-1/2}$. Then
\begin{equation}\label{eq:canon}
\frac{\det A_\rho\det C_\rho}{\det R}=\prod_j\frac{1}{1-\varrho_j^2},\qquad I_\rho=-\tfrac12\sum_j\log(1-\varrho_j^2),
\end{equation}
where $I_\rho$ is the mutual information between the sector and its complement. As the leakage tends to zero, $\varrho_j\to0$ and $I_\rho=\tfrac12\sum_j\varrho_j^2+O(\varrho^4)=O(\ell_\rho^2)$.
\end{proposition}

\begin{proof}
The Schur complement gives $\det R=\det A_\rho\,\det(C_\rho-B_\rho^\ast A_\rho^{-1}B_\rho)$, and $\det(C_\rho-B_\rho^\ast A_\rho^{-1}B_\rho)=\det C_\rho\prod_j(1-\varrho_j^2)$ since the $\varrho_j^2$ are the eigenvalues of $C_\rho^{-1}B_\rho^\ast A_\rho^{-1}B_\rho$. The Gaussian mutual information of the two blocks is $I_\rho=\tfrac12\log(\det A_\rho\det C_\rho/\det R)$~\cite{covertthomas2006}; the small-leakage expansion follows from $-\log(1-x)=x+O(x^2)$.
\end{proof}

The rate accounting is now exact and carries no small-leakage hypothesis: it trades the coding rate of grouping a sector against the complexity saved.

\begin{theorem}[Exact blockwise rate identity and marginal selection threshold]\label{thm:rd-threshold}
Let a transform-coding objective charge a complexity cost $c_\rho>0$ for each sector handled by a data-driven KLT block, traded against rate through the Lagrangian $\mathcal J(S)=\Delta R(S)+\lambda\sum_{\rho\notin S}c_\rho$, where $S$ is the set of sectors handled by the fixed group transform and $\Delta R(S)$ is the excess coding rate it incurs in the model of Theorem~\ref{thm:rd-gain}. Then
\begin{equation}\label{eq:rd-threshold}
\Delta R(S)=\tfrac12\Bigl(\sum_{\rho\in S}\log\det A_\rho+\log\det C_{S^c}-\log\det R\Bigr),
\end{equation}
the multi-information of the partition that separates the sectors in $S$ and keeps the complement $S^c$ together, where $C_{S^c}=P_{S^c}RP_{S^c}$. The marginal cost of moving sector $\rho$ from the data-driven complement into $S$ is exactly its mutual information $I(X_\rho;X_{S^c\setminus\rho})$ with the remaining data-driven sectors. Hence, for a current set $S$, moving sector $\rho$ into the fixed group-transform set lowers $\mathcal J$ if and only if
\begin{equation}\label{eq:threshold}
I(X_\rho;X_{S^c\setminus\rho})<\lambda c_\rho ,
\end{equation}
with equality leaving $\mathcal J$ unchanged. This is a one-step optimality condition relative to the current $S$: because the conditioning set $S^c\setminus\rho$ varies as sectors are reassigned, it certifies a local improvement rather than the global minimizer of $\mathcal J$, which would require an additional submodularity or greedy-optimality argument.
\end{theorem}

\begin{proof}
The transform realizing $S$ block-diagonalizes each sector in $S$ and applies the KLT of the joint complement block $C_{S^c}$; as in Theorem~\ref{thm:rd-gain} its transform-domain variance product is $\prod_{\rho\in S}\det A_\rho\cdot\det C_{S^c}$, which gives~\eqref{eq:rd-threshold}. The marginal increment is
\[
\Delta R(S\cup\{\rho\})-\Delta R(S)=\tfrac12\log\frac{\det A_\rho\,\det C_{S^c\setminus\rho}}{\det C_{S^c}}=I(X_\rho;X_{S^c\setminus\rho}).
\]
Including $\rho$ in $S$ changes $\mathcal J$ by $I(X_\rho;X_{S^c\setminus\rho})-\lambda c_\rho$, so the move lowers $\mathcal J$ exactly when this marginal is below the released complexity term $\lambda c_\rho$. The criterion is relative to the current $S$, since the conditioning set $S^c\setminus\rho$ varies with $S$.
\end{proof}

\begin{remark}[Small-leakage limit and energy threshold]\label{rem:smallleak}
By Proposition~\ref{prop:canon}, $I(X_\rho;X_{\rho^c})=\tfrac12\sum_j\varrho_j^2+O(\varrho^4)$. In the small-leakage regime the exact threshold~\eqref{eq:threshold} reduces to the quadratic condition $\sum_j\varrho_j^2\lesssim 2\lambda c_\rho$ on the canonical correlations, which is the information-theoretic counterpart of the energy bound~\eqref{eq:blockwise-energy}. Both are second order in the leakage, and both vanish exactly at weak symmetry. The determinant identity~\eqref{eq:rd-gain} is the fully general statement: the coding penalty of the blockwise group transform is the sector multi-information at any leakage level, with no spectral-gap hypothesis.
\end{remark}

\section{Symmetry on the Time--Frequency Plane}\label{sec:symplectic}

The commutant table of Section~\ref{sec:commutant} uses groups acting on the domain $\mathcal{T}$. A larger class of symmetries acts on the time--frequency plane and supplies the two transform families that complete the classical picture: the symplectic family, containing the fractional Fourier and linear canonical transforms, and the reflection family, containing the Dunkl transform.

\subsection{The symplectic family}

The symplectic group $Sp(2,\mathbb{R})$ consists of the linear maps of the time--frequency plane that preserve the symplectic form; its unitary action on $L^2(\mathbb{R})$ is the metaplectic (Weil) representation $\mu$, and the operators $\mu(g)$ are exactly the linear canonical transforms (LCTs), which contain the Fourier, fractional Fourier, Fresnel, and scaling transforms as special cases~\cite{namias1980}. A one-parameter subgroup $g_\theta$ has $\mu(g_\theta)=e^{-i\theta H}$ for a self-adjoint quadratic Hamiltonian $H$, the Weyl quantization of the corresponding quadratic form on phase space. The three conjugacy classes of one-parameter subgroup give the generators
\begin{equation}
H_{\mathrm{ell}}=\tfrac12\bigl(-\partial_t^2+t^2\bigr),\qquad
H_{\mathrm{hyp}}=\tfrac{i}{2}\bigl(t\partial_t+\partial_t\,t\bigr),\qquad
H_{\mathrm{par}}=-\partial_t^2 .
\end{equation}

\begin{theorem}[Symplectic commutant: the elliptic class]\label{thm:symplectic}
Let $R$ be a covariance operator on $L^2(\mathbb{R})$ commuting with the metaplectic image $\mu(g_\theta)$ of the elliptic one-parameter subgroup, generated by $H_{\mathrm{ell}}=\tfrac12(-\partial_t^2+t^2)$. This generator is the harmonic-oscillator Hamiltonian, with discrete simple spectrum $\{n+\tfrac12:n\ge0\}$ and the Hermite--Gauss functions $h_n(t)=H_n(t)e^{-t^2/2}$ as an orthonormal eigenbasis of $L^2(\mathbb{R})$. Then $R$ is a function of $H_{\mathrm{ell}}$, and the KLT of $R$ is the Hermite--Gauss basis, the eigenfunctions of the fractional Fourier transform (eigenvalues $e^{-i(n+1/2)\theta}$). The same conclusion holds for any LCT conjugate to the elliptic class, whose eigenfunctions are the correspondingly scaled and chirped Hermite--Gauss functions.
\end{theorem}

\begin{proof}
The relation $[\mu(g_\theta),R]=0$ for all $\theta$ is equivalent to $[H_{\mathrm{ell}},R]=0$. Since $H_{\mathrm{ell}}$ is self-adjoint with discrete simple spectrum and complete $L^2$ eigenbasis $\{h_n\}$, the bounded operators commuting with it are exactly the functions of $H_{\mathrm{ell}}$, the maximal Abelian algebra it generates, so $R$ shares the eigenbasis $\{h_n\}$. These satisfy $H_{\mathrm{ell}}h_n=(n+\tfrac12)h_n$ and are the eigenfunctions of the fractional Fourier transform with eigenvalues $e^{-i(n+1/2)\theta}$~\cite{namias1980}.
\end{proof}

Theorem~\ref{thm:symplectic} adds a symplectic band to the commutant table: a process whose covariance is invariant under phase-space rotation has the Hermite--Gauss (fractional Fourier) KLT, exactly as a translation-invariant process has the Fourier KLT. The role played by the differential operator $d/dt$ for translations is played here by the quadratic operator $H$, and the matched transform is read off from which conjugacy class $H$ belongs to.

\begin{remark}[Hyperbolic and parabolic classes: continuous spectrum]\label{rem:symplectic-cont}
The other two conjugacy classes of $Sp(2,\mathbb{R})$ have generators with purely continuous spectrum, so the matched diagonalization is a continuous (direct-integral) spectral decomposition rather than an orthonormal $L^2$ eigenbasis, and the argument above, which relies on a discrete eigenbasis, does not transfer verbatim. The hyperbolic generator $H_{\mathrm{hyp}}=\tfrac{i}{2}(t\partial_t+\partial_t\,t)$ is the dilation operator, diagonalized by the Mellin transform with generalized eigenfunctions $t^{i\sigma}$ ($\sigma\in\mathbb{R}$); a dilation-invariant $R$ is a function of $H_{\mathrm{hyp}}$ and is diagonal in the Mellin variable, the scale-invariant row of Section~\ref{sec:commutant}. The parabolic generator $H_{\mathrm{par}}=-\partial_t^2$ is diagonalized by the Fourier transform, with generalized eigenfunctions the plain exponentials $e^{i\omega t}$; the chirp character of the parabolic class resides in its group elements $\mu(g_\theta)$, which act by chirp multiplication $e^{i\tau t^2/2}$, not in the eigenfunctions of the generator. Placing either case on the same footing as the elliptic one requires the spectral theory of operators with continuous spectrum (rigged Hilbert spaces or direct integrals), which we do not develop here.
\end{remark}

\subsection{The reflection-group family}

A second extension uses finite reflection groups. Let $W$ be a finite Coxeter group acting on $\mathbb{R}^n$ with multiplicity function $k$, and let $\{T_i\}$ be the associated Dunkl operators, a commuting family of $W$-equivariant differential-difference operators that deform the partial derivatives by reflection terms~\cite{dunkl1989}. The Dunkl Laplacian $\Delta_k=\sum_i T_i^2$ is diagonalized by the Dunkl transform $\mathcal{D}_k$, an isometry of $L^2(\mathbb{R}^n,w_k\,dx)$ for the weight $w_k(x)=\prod_{\alpha}|\langle\alpha,x\rangle|^{2k_\alpha}$, which reduces to the Fourier transform at $k=0$~\cite{dejeu1993}.

\begin{proposition}[Reflection-group commutant]\label{prop:dunkl}
Let $R$ act on $L^2(\mathbb{R}^n,w_k\,dx)$ and commute with the $W$-action and with the Dunkl Laplacian $\Delta_k$, which has simple spectrum on each $W$-isotypic component. Then $R$ is a function of $\Delta_k$ and is diagonalized by the Dunkl transform $\mathcal{D}_k$. As with the ordinary Laplacian it generalizes ($k=0$), the spectrum of $\Delta_k$ is purely continuous, so this diagonalization is the generalized (direct-integral) decomposition of Remark~\ref{rem:symplectic-cont}, with the Dunkl kernel as generalized eigenfunctions rather than an orthonormal $L^2$ eigenbasis. At $k=0$ the kernel is the Fourier exponential, and for $W=\mathbb{Z}_2$ acting on $\mathbb{R}$ the rank-one Dunkl (generalized cosine) kernel, the continuous analogue of the cosine row of Section~\ref{sec:commutant}.
\end{proposition}

\begin{proof}
The Dunkl transform $\mathcal{D}_k$ realizes the spectral decomposition of the self-adjoint $\Delta_k$~\cite{dejeu1993}. With $\Delta_k$ of simple spectrum on each isotypic component and $[\Delta_k,R]=[W,R]=0$, functional calculus makes $R$ a function of $\Delta_k$ on each isotypic component, so $\mathcal{D}_k$ diagonalizes $R$. Since $\Delta_k$ has continuous spectrum, this is the generalized diagonalization of Remark~\ref{rem:symplectic-cont}, not an $L^2$ eigenbasis. The $k=0$ and rank-one specializations follow from the corresponding reductions of the Dunkl kernel.
\end{proof}

The reflection-group row generalizes the cosine entry the way $SO(3)$ generalizes $SO(2)$: the rank-one $\mathbb{Z}_2$ reflection gives the cosine transform, and higher-rank Coxeter groups give the Dunkl family.

\section{Two Notions of Optimality}\label{sec:optimality}

Theorem~\ref{thm:klt-opt} is often quoted without its criterion. The criterion is energy compaction: among orthonormal expansions the KLT minimizes the mean-square error of rank-$K$ truncation. This is a \emph{second-order} statement, a function of the covariance $R$ alone, and the entire apparatus of this paper, the commutant, the generator eigenvalue problem, the matched transform, likewise sees only $R$. A different and equally natural criterion is the \emph{coding rate}: the number of bits needed to represent the process to a fixed fidelity. The two criteria do not coincide in general, and separating them resolves a persistent confusion about whether the KLT can be outperformed.

\begin{proposition}[The KLT is a second-order construct]\label{prop:secondorder}
The KLT depends on the process only through its covariance operator $R$. Two processes with the same $R$ have the same KLT, the same compaction performance, and the same matched symmetry, irrespective of their higher-order statistics.
\end{proposition}

This is immediate from Definition~\ref{def:ckl}, and it has a sharp consequence. The compaction criterion is determined by $R$; the coding rate is not. At high resolution the rate to encode a coefficient $c_k$ to fixed distortion is governed by its differential entropy $h(c_k)$, and the total rate by $\sum_k h(c_k)$ up to transform-independent terms~\cite{goyal2001}. For a Gaussian process each $h(c_k)$ is a function of the variance $\sigma_k^2$ alone, so minimizing $\sum_k h(c_k)$ coincides with compaction and the KLT is rate-optimal; this is the classical result of Huang and Schultheiss~\cite{huang1963}.

\begin{theorem}[Criterion gap]\label{thm:criterion}
For a Gaussian process the compaction-optimal transform (the KLT) is also optimal for high-resolution transform coding under mean-square distortion. For non-Gaussian processes the two criteria diverge: there exist processes for which an orthonormal transform other than the KLT achieves strictly smaller coding rate at the same distortion. Equivalently, the compaction-optimal ``KLT genie'' can be strictly outperformed on the rate criterion.
\end{theorem}

\begin{proof}
We first record a high-resolution rate identity that drives both parts. Let $U$ be an orthonormal transform, $c=Ux$ the coefficient vector, $\sigma_k^2=(URU^\ast)_{kk}$ the coefficient variances, and $h(c_k)$ the marginal differential entropies in bits. Encode each coefficient with an entropy-constrained uniform scalar quantizer of step $\Delta_k$. In the high-resolution limit $\Delta_k\to0$ the per-coefficient rate and mean-square distortion obey $D_k=\tfrac1{12}\,2^{2h(c_k)}2^{-2R_k}(1+o(1))$, equivalently $R_k=h(c_k)-\tfrac12\log_2(12D_k)+o(1)$~\cite{goyal2001}. Minimizing the total rate $\tfrac1N\sum_kR_k$ at a fixed average distortion $\tfrac1N\sum_kD_k=D$ holds $\sum_kh(c_k)$ fixed and leaves $\max\sum_k\log_2D_k$ subject to $\sum_kD_k=ND$, which by the arithmetic-geometric mean inequality is attained at $D_k=D$ for every $k$. Hence, in the high-resolution regime,
\begin{equation}\label{eq:hr-rate}
R_{\mathrm{total}}(U)=\frac1N\sum_{k}h(c_k)-\tfrac12\log_2(12D)+o(1),
\end{equation}
so at a fixed distortion $D$ minimizing the rate over orthonormal $U$ is equivalent to minimizing the summed marginal entropy $\sum_kh(c_k)$.

\emph{Gaussian case.} If $x$ is Gaussian then each $c_k$ is Gaussian with variance $\sigma_k^2$, so $h(c_k)=\tfrac12\log_2(2\pi e\,\sigma_k^2)$ and $\sum_kh(c_k)=\tfrac N2\log_2(2\pi e)+\tfrac12\log_2\prod_k\sigma_k^2$. The matrix $URU^\ast$ is Hermitian positive-semidefinite with diagonal entries $\sigma_k^2$ and, since $U$ is unitary, determinant $\det R$. Hadamard's inequality gives $\prod_k\sigma_k^2\ge\det R$, with equality if and only if $URU^\ast$ is diagonal, that is, if and only if $U$ is an eigenbasis of $R$. Thus the KLT minimizes $\sum_kh(c_k)$, and by \eqref{eq:hr-rate} it minimizes the high-resolution rate at every $D$; when $R$ has simple spectrum the minimizer is unique up to coordinate sign and order. This is the theorem of Huang and Schultheiss~\cite{huang1963}.

\emph{Non-Gaussian case.} By~\eqref{eq:hr-rate} the rate-optimal transform minimizes $\sum_kh(c_k)$, which for a non-Gaussian source is not a function of the coordinate variances $\sigma_k^2$ and so need not coincide with the variance-product-minimizing KLT. That the two are in fact strictly distinct for some processes is the theorem of Effros, Feng, and Zeger~\cite{effros2004}: for Gaussian-mixture sources the KLT is in general strictly suboptimal for the transform-coding rate, in both the fixed-rate and the high-resolution variable-rate settings. On such a source~\eqref{eq:hr-rate} makes an orthonormal transform other than the KLT attain a strictly smaller rate at equal distortion, which is the second claim.
\end{proof}

\begin{remark}[An explicit instance]\label{rem:criterion-witness}
An explicit construction makes the separation concrete. Let $x$ be the equiprobable mixture of two zero-mean bivariate Gaussians with covariances
\[
\Sigma_{\pm}=Q(\pm\theta)\,\mathrm{diag}(s,1/s)\,Q(\pm\theta)^{\!\top},\qquad s>1,\ 0<\theta<\tfrac\pi4,
\]
where $Q(\phi)$ is the planar rotation. The off-diagonal entries of $\Sigma_+$ and $\Sigma_-$ cancel under $\theta\mapsto-\theta$, so the mixture covariance $R=\tfrac12(\Sigma_++\Sigma_-)=\mathrm{diag}(s\cos^2\theta+\tfrac1s\sin^2\theta,\ s\sin^2\theta+\tfrac1s\cos^2\theta)$ is diagonal with distinct entries and the KLT is the coordinate basis. There each coefficient marginal is an equal mixture of two Gaussians of \emph{equal} variance, since $\Sigma_+$ and $\Sigma_-$ share each diagonal entry, hence is itself Gaussian and at the maximum-entropy point for its variance. The rotation $U=Q(\theta)^{\!\top}$ diagonalizes $\Sigma_+$ to $\mathrm{diag}(s,1/s)$ and sends $\Sigma_-$ to $Q(2\theta)\,\mathrm{diag}(s,1/s)\,Q(2\theta)^{\!\top}$, so each marginal becomes a scale mixture of two Gaussians of \emph{unequal} variance, whose differential entropy is strictly below that of the Gaussian of the same variance. Whether the summed marginal entropy falls is a finite-amplitude question, not a local one: by Hadamard every non-KLT orthonormal transform strictly raises the variance product $\prod_k\sigma_k^2$, so the per-coordinate entropy deficit must overcome that penalty, and near the KLT the deficit is $O(\theta^4)$ against an $O(\theta^2)$ penalty, leaving the KLT a strict local minimum of summed marginal entropy; the gain appears only past a threshold rotation. At $s=8$ and $\theta=28^\circ$, evaluating the marginal entropies by quadrature gives $5.8657$ bits for the KLT and $5.6703$ bits for $Q(\theta)^{\!\top}$, a strict reduction of $0.1954$ bits per vector at the common variance sum $\operatorname{tr}R=8.125$; the evaluation script is included with the manuscript. This realizes the separation of Theorem~\ref{thm:criterion} for a concrete source.
\end{remark}

The mechanism is structural. The KLT removes second-order dependence but not higher-order dependence; for a non-Gaussian source the coefficients can be uncorrelated yet statistically dependent, and a transform that reduces the residual dependence yields a lower-entropy, more compressible representation. The two criteria coincide for Gaussian sources, and more broadly in settings where second-order statistics determine the relevant marginal entropy criterion.

\begin{example}[A signal sparse off the KLT axes]\label{ex:chirp}
Let the process be a linear chirp with random rate, $f(t)=e^{i(\omega t+\frac12\gamma t^2)}$ with $\gamma$ drawn from a distribution. Averaged over $\gamma$, the covariance is invariant under phase-space rotation, so by Theorem~\ref{thm:symplectic} the ensemble KLT is the Hermite--Gauss basis and the energy of any single realization is spread across many Hermite--Gauss coefficients. Each realization is, however, a single coefficient in the fractional Fourier basis matched to its rate: it is one-sparse off the KLT axes. A coder that estimates $\gamma$ and encodes in the matched fractional Fourier domain spends one coefficient plus the side information for $\gamma$, beating the ensemble KLT genie whenever that side information is cheaper than the spread. This is the symplectic instance of Theorem~\ref{thm:criterion}: the rate-optimal transform is the matched LCT, not the KLT.
\end{example}

In the language of this paper, the KLT and its matched symmetry are optimal for the second-order criterion they are built from; outperforming the KLT requires a criterion, coding rate, sparsity, classification error, that depends on structure beyond the covariance, together with a transform matched to that structure. The symplectic example shows that this is not exotic: a process matched to a nontrivial symplectic subgroup is sparse in the corresponding linear canonical domain and compresses better there than in the KLT domain.

\section{The Double-Commutator Eigenvalue Theorem}\label{sec:gevp}

We now formalize the inverse problem. Fix a Hermitian operator $R$ (the covariance operator restricted to a finite-dimensional invariant subspace, or a discretization $\R\in\mathbb{C}^{M\times M}$) and a finite-dimensional space of candidate generators
\begin{equation}
\mathcal{B}=\spn\{B_1,\ldots,B_d\}\subseteq\mathbb{C}^{M\times M},\qquad \{B_k\}\ \text{linearly independent}.
\end{equation}
No structural restriction is placed on the $B_k$. The problem is to find $A^\ast=\sum_k c_k B_k\in\mathcal{B}$ minimizing the commutativity residual $\delta(A,\R)=\Fro{[\R,A]}/(\Fro{\R}\,\Fro{A})$.

\begin{definition}[Double commutator]
For $A\in\mathbb{C}^{M\times M}$, the double-commutator superoperator is
\begin{equation}
\ad_{\R}^2(A)=[\R,[\R,A]]=\R^2A-2\R A\R+A\R^2 .
\end{equation}
\end{definition}

The construction is the Rayleigh--Ritz / Galerkin reduction of the operator $\ad_{\R}\colon A\mapsto[\R,A]$ under the Frobenius inner product $\langle X,Y\rangle_F=\Tr(X^HY)$. For Hermitian $\R$ this operator is self-adjoint (Lemma~\ref{lem:quad}), so its normal operator $\ad_{\R}^\ast\ad_{\R}$ is $\ad_{\R}^2$ itself, and the double commutator appears intrinsically rather than as a derived quantity.

\begin{lemma}[Self-adjointness and squared residual]\label{lem:quad}
Let $\R$ be Hermitian and equip $\mathbb{C}^{M\times M}$ with the Frobenius (Hilbert--Schmidt) inner product $\langle X,Y\rangle_F=\Tr(X^HY)$. Then $\ad_{\R}\colon A\mapsto[\R,A]$ is self-adjoint,
\begin{equation}\label{eq:adR-selfadjoint}
\langle X,[\R,Y]\rangle_F=\langle[\R,X],Y\rangle_F ,
\end{equation}
so $\ad_{\R}^2$ is self-adjoint and positive semidefinite, with
\begin{equation}\label{eq:LR-bilinear}
\langle A,\ad_{\R}^2(B)\rangle_F=\langle[\R,A],[\R,B]\rangle_F=\langle\ad_{\R}^2(A),B\rangle_F .
\end{equation}
In particular, taking $B=A$,
\begin{equation}\label{eq:quadform}
\langle A,\ad_{\R}^2(A)\rangle_F=\Tr\!\bigl(A^H\ad_{\R}^2(A)\bigr)=\Fro{[\R,A]}^2\ge 0 ,
\end{equation}
with equality if and only if $[\R,A]=0$.
\end{lemma}

\begin{proof}
Since $\R$ is Hermitian, $[\R,X]^H=X^H\R-\R X^H$, so cyclicity of the trace gives
\[
\begin{aligned}
\langle X,[\R,Y]\rangle_F&=\Tr(X^H\R Y)-\Tr(X^HY\R)=\Tr(X^H\R Y)-\Tr(\R X^HY)\\
&=\Tr\!\bigl((X^H\R-\R X^H)Y\bigr)=\langle[\R,X],Y\rangle_F,
\end{aligned}
\]
which is~\eqref{eq:adR-selfadjoint}. Applying it twice, $\langle A,\ad_{\R}^2(B)\rangle_F=\langle[\R,A],[\R,B]\rangle_F=\langle\ad_{\R}^2(A),B\rangle_F$, so $\ad_{\R}^2$ is self-adjoint; the middle term is a positive-semidefinite form, giving~\eqref{eq:LR-bilinear}. Setting $B=A$ yields $\langle A,\ad_{\R}^2(A)\rangle_F=\Fro{[\R,A]}^2$, which vanishes exactly when $[\R,A]=0$. Expanding $\ad_{\R}^2(A)=\R^2A-2\R A\R+A\R^2$ and using cyclicity confirms $\Tr(A^H\ad_{\R}^2(A))=\Fro{[\R,A]}^2$, which is~\eqref{eq:quadform}.
\end{proof}

\begin{theorem}[Double-commutator eigenvalue problem]\label{thm:main}
Let $\R\in\mathbb{C}^{M\times M}$ be Hermitian, and let $\mathcal{B}=\spn\{B_1,\ldots,B_d\}$ with the $B_k$ linearly independent. Define the Hermitian matrices
\begin{equation}\label{eq:MG}
M_{ij}=\Tr\!\bigl(B_i^H\,\ad_{\R}^2(B_j)\bigr),\qquad G_{ij}=\Tr\!\bigl(B_i^HB_j\bigr).
\end{equation}
Then the minimizer of $\Fro{[\R,A]}^2/\Fro{A}^2$ over $A\in\mathcal{B}\setminus\{0\}$ is $A^\ast=\sum_j c_j^\ast B_j$, where $\mathbf{c}^\ast$ is a generalized eigenvector for the smallest generalized eigenvalue $\lambda_{\min}$ of
\begin{equation}\label{eq:gevp}
\mathbf{M}\mathbf{c}=\lambda\,\mathbf{G}\mathbf{c}.
\end{equation}
The cost is $O(d^2M^3+d^3)$, independent of $M$ in the dimension $d$ of the pencil.
\end{theorem}

\begin{proof}
Writing $A=\sum_j c_jB_j$ and using Lemma~\ref{lem:quad},
\[
\Fro{[\R,A]}^2=\sum_{i,j}\overline{c_i}c_j\,\Tr(B_i^H\ad_{\R}^2(B_j))=\mathbf{c}^H\mathbf{M}\mathbf{c},
\qquad
\Fro{A}^2=\sum_{i,j}\overline{c_i}c_j\,\Tr(B_i^HB_j)=\mathbf{c}^H\mathbf{G}\mathbf{c}.
\]
The problem is therefore the minimization of the generalized Rayleigh quotient $\mathbf{c}^H\mathbf{M}\mathbf{c}/\mathbf{c}^H\mathbf{G}\mathbf{c}$. Linear independence of $\{B_k\}$ makes $\mathbf{G}$ Hermitian positive definite, so by Courant--Fischer the minimum is attained at the smallest generalized eigenvector of \eqref{eq:gevp}.
\end{proof}

Figure~\ref{fig:gevp} applies the pencil to the mixed kernel of Example~\ref{ex:mixed}: the individual generator residuals, a restricted translation--modulation family that excludes reflection, and the pencil solution, which returns pure reflection at $\lambda_{\min}=0$.

\begin{figure}[H]
\centering
\includegraphics[width=0.96\textwidth]{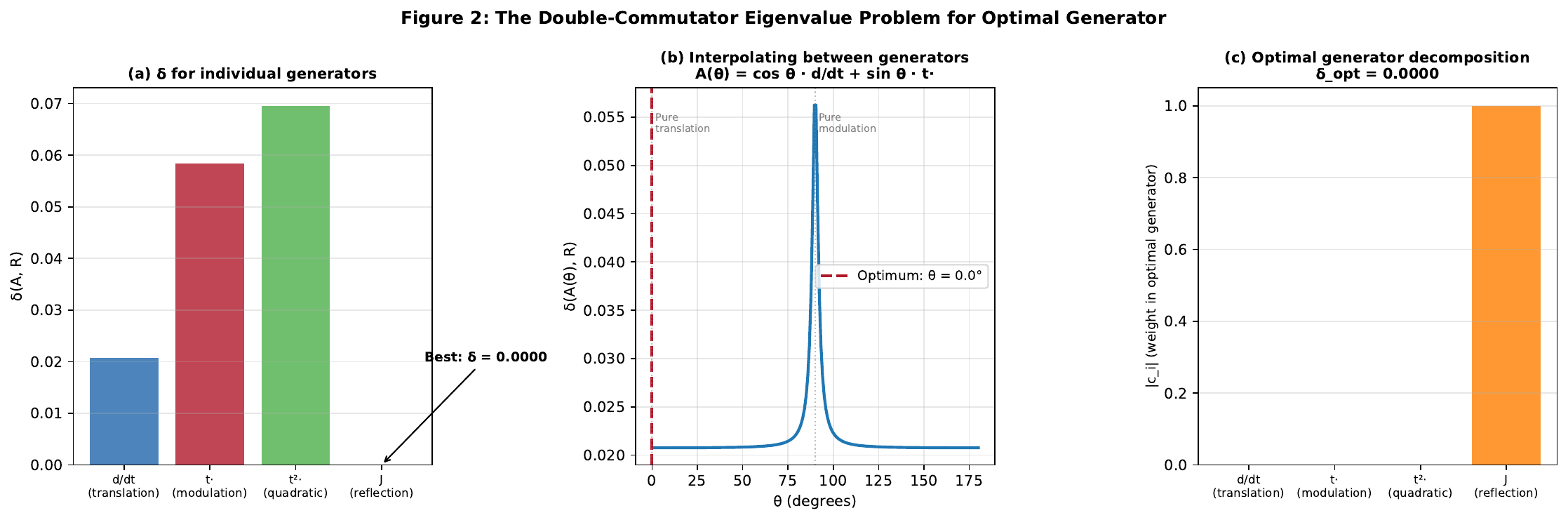}
\caption{The double-commutator eigenvalue problem for the mixed kernel. (a) Individual generator residuals: the reflection generator $J$ attains $\delta=0$ (persymmetry), with translation next at $\delta\approx0.02$. (b) The restricted translation--modulation family $A(\theta)=\cos\theta\,(d/dt)+\sin\theta\,(t\,\cdot)$, which excludes reflection, has its minimum at pure translation. (c) Solving the pencil $\ad_R^2(A^\ast)=\lambda A^\ast$ of Theorem~\ref{thm:main} returns the optimal generator as pure reflection with $\lambda_{\min}=0$, consistent with panel (a) and the persymmetry of Example~\ref{ex:mixed}.}
\label{fig:gevp}
\end{figure}

\begin{corollary}[Zero eigenvalue]\label{cor:zero}
$\lambda_{\min}=0$ if and only if some nonzero $A\in\mathcal{B}$ commutes with $\R$.
\end{corollary}

\begin{proof}
By Theorem~\ref{thm:main}, $\lambda_{\min}=\min_{A\in\mathcal{B}\setminus\{0\}}\Fro{[\R,A]}^2/\Fro{A}^2$, which is zero iff the numerator vanishes for some nonzero $A$, i.e.\ $[A,\R]=0$.
\end{proof}

\begin{remark}[Topology, finite dimension, and attainment]\label{rem:wellposed}
The candidate space $\mathcal{B}$ is finite-dimensional throughout. The residual $\delta(A,\R)$ is then continuous on the unit sphere $\{A\in\mathcal{B}:\Fro{A}=1\}$, which is compact, so the minimum in Theorem~\ref{thm:main} is attained rather than merely an infimum, and all norms on $\mathcal{B}$ are equivalent. When $\R$ acts on an infinite-dimensional Hilbert space the operative requirement is that $[\R,A]$ be Hilbert--Schmidt for every $A\in\mathcal{B}$, so that the Frobenius norm $\Fro{[\R,A]}$ and the quadratic form of Lemma~\ref{lem:quad} are finite; the finite generator families used in the continuum (Sections~\ref{sec:symplectic} and~\ref{sec:prolate}, and Theorem~\ref{thm:lie-recovery}) carry this hypothesis explicitly. No claim is made for an infinite-dimensional generator space, where the unit sphere is noncompact and a minimizing sequence need not converge within $\mathcal{B}$.
\end{remark}

\begin{proposition}[Structure of $\mathbf{M}$ and $\mathbf{G}$]\label{prop:MGstructure}
With $\R$ Hermitian and $\{B_k\}$ linearly independent:
\begin{enumerate}[leftmargin=2.4em]
\item[(i)] $\mathbf{M}$ is Hermitian positive semidefinite and $\mathbf{G}$ is Hermitian positive definite, so the eigenvalues of \eqref{eq:gevp} are real and nonnegative.
\item[(ii)] If each $B_k$ is skew-Hermitian (i.e.\ $\mathcal{B}\subseteq\mathfrak{u}(M)$, the Lie algebra of the unitary group), then $\mathbf{M}$ and $\mathbf{G}$ are real symmetric and \eqref{eq:gevp} is a real symmetric positive-semidefinite pencil over the real vector space $\mathfrak{u}(M)$, taken with the real Hilbert--Schmidt inner product $\langle X,Y\rangle_{\mathrm{Re}}=\operatorname{Re}\Tr(X^HY)$.
\end{enumerate}
\end{proposition}

\begin{proof}
(i) For any $\mathbf{c}$, $\mathbf{c}^H\mathbf{M}\mathbf{c}=\Fro{[\R,A]}^2\ge0$ with $A=\sum_kc_kB_k$ by Lemma~\ref{lem:quad}, so $\mathbf{M}\succeq0$. Self-adjointness of $\ad_\R^2$ under the Frobenius inner product (Lemma~\ref{lem:quad}) gives
\[
M_{ij}=\langle B_i,\ad_\R^2 B_j\rangle_F=\langle\ad_\R^2 B_i,B_j\rangle_F=\overline{M_{ji}},
\]
so $\mathbf{M}$ is Hermitian. $\mathbf{G}$ is Hermitian with $\mathbf{c}^H\mathbf{G}\mathbf{c}=\Fro{\sum_kc_kB_k}^2>0$ for $\mathbf{c}\neq0$.
(ii) Suppose $B_k^H=-B_k$, so $\mathcal{B}\subseteq\mathfrak{u}(M)$, which is a real vector space (it is not closed under multiplication by $i$) carrying the real Hilbert--Schmidt inner product $\langle X,Y\rangle_{\mathrm{Re}}=\operatorname{Re}\Tr(X^HY)$. For skew-Hermitian $X,Y$ the Frobenius pairing is already real, since $\overline{\Tr(X^HY)}=\Tr(Y^HX)=\Tr(X^HY)$ by cyclicity, so $\Tr(X^HY)=\langle X,Y\rangle_{\mathrm{Re}}$ on $\mathfrak{u}(M)$. Hence $G_{ij}=\Tr(B_i^HB_j)=-\Tr(B_iB_j)$ is real, and symmetric because $\Tr(B_iB_j)=\Tr(B_jB_i)$ gives $G_{ij}=G_{ji}$. The adjoint action shifts parity: for Hermitian $\R$, $[\R,B_k]^H=B_k^H\R-\R B_k^H=\R B_k-B_k\R=[\R,B_k]$, so $[\R,B_k]$ is Hermitian, and a second bracket returns $\ad_\R^2(B_k)=[\R,[\R,B_k]]$ to skew-Hermitian; thus $\ad_\R^2$ maps $\mathfrak{u}(M)$ into itself and, by~\eqref{eq:LR-bilinear}, is self-adjoint for the real pairing,
\[
\langle B_i,\ad_\R^2B_j\rangle_{\mathrm{Re}}=\operatorname{Re}\langle[\R,B_i],[\R,B_j]\rangle_F=\langle\ad_\R^2B_i,B_j\rangle_{\mathrm{Re}}.
\]
Therefore $M_{ij}=\Tr(B_i^H\ad_\R^2B_j)=\langle B_i,\ad_\R^2B_j\rangle_{\mathrm{Re}}$ is real and symmetric. With the Hermiticity and positivity from~(i), $\mathbf{M}$ and $\mathbf{G}$ are real symmetric, and \eqref{eq:gevp} is a real symmetric positive-semidefinite pencil over $\mathfrak{u}(M)$.
\end{proof}

Remark~\ref{rem:lie} records why the skew-Hermitian case is the one of interest for continuous symmetry.

\begin{remark}[Lie-algebra specialization]\label{rem:lie}
When the candidate symmetries are continuous group elements $U=\exp(\theta A)$, the natural setting is $\mathcal{B}\subseteq\mathfrak{u}(M)$, where the real symmetric structure of Proposition~\ref{prop:MGstructure}(ii) holds and $\ad_R^2(A^\ast)=\lambda A^\ast$ is a real symmetric eigenvalue problem on $\mathcal{B}$. Discrete candidates such as permutation matrices, which are unitary but neither Hermitian nor skew-Hermitian, are admitted by the general Theorem~\ref{thm:main}; this is the case used in Sections~\ref{sec:pcf}--\ref{sec:sequential}.
\end{remark}

\begin{proposition}[Gradient-flow realization of the matched-generator search]\label{prop:flow}
Let $\rho(\mathbf{c})=\mathbf{c}^H\mathbf{M}\mathbf{c}/\mathbf{c}^H\mathbf{G}\mathbf{c}$, so that $\rho(\mathbf{c})=\Fro{[\R,A]}^2/\Fro{A}^2$ for $A=\sum_k c_k B_k$. Along the flow
\begin{equation}\label{eq:flow}
\dot{\mathbf{c}}=-\bigl(\mathbf{M}-\rho(\mathbf{c})\mathbf{G}\bigr)\mathbf{c},
\end{equation}
$\rho(\mathbf{c}(t))$ is nonincreasing, and the equilibria with $\mathbf{c}\neq0$ are exactly the generalized eigenvectors of~\eqref{eq:gevp}. On the unit $\mathbf{G}$-sphere the norm-preserving representative $\dot{\mathbf{c}}=-(\mathbf{G}^{-1}\mathbf{M}-\rho(\mathbf{c})\mathbf{I})\mathbf{c}$ is the Riemannian gradient flow of $\rho$ and, when $\lambda_1$ is simple, converges to the minimal generalized eigenvector from generic initial data. Writing $A(t)=\sum_k c_k(t)B_k$, the flow~\eqref{eq:flow} is the $\mathcal{B}$-projection of the double-commutator flow $\dot A=-\bigl([\R,[\R,A]]-\rho(A)A\bigr)$.
\end{proposition}

\begin{proof}
The Euclidean gradient of $\rho$ is $\nabla\rho(\mathbf{c})=\tfrac{2}{\mathbf{c}^H\mathbf{G}\mathbf{c}}(\mathbf{M}-\rho(\mathbf{c})\mathbf{G})\mathbf{c}$, so the right side of~\eqref{eq:flow} equals $-\tfrac12(\mathbf{c}^H\mathbf{G}\mathbf{c})\nabla\rho$ and
\[
\dot\rho=\langle\nabla\rho,\dot{\mathbf{c}}\rangle=-\tfrac12(\mathbf{c}^H\mathbf{G}\mathbf{c})\Fro{\nabla\rho}^2\le0,
\]
with equality iff $(\mathbf{M}-\rho\mathbf{G})\mathbf{c}=0$, i.e.\ $\mathbf{c}$ is a generalized eigenvector. The substitution $\mathbf{d}=\mathbf{G}^{1/2}\mathbf{c}$ turns~\eqref{eq:gevp} into the ordinary eigenproblem for $\widetilde{\mathbf{M}}=\mathbf{G}^{-1/2}\mathbf{M}\mathbf{G}^{-1/2}$ and the norm-preserving flow into the standard Rayleigh-quotient flow $\dot{\mathbf{d}}=-(\widetilde{\mathbf{M}}-\rho(\mathbf{d})\mathbf{I})\mathbf{d}$, which preserves $\Fro{\mathbf{d}}$ and descends to the bottom eigenvector for generic data when $\lambda_1$ is simple~\cite{helmkemoore1994}. Finally $\mathbf{M}\mathbf{c}$ and $\mathbf{G}\mathbf{c}$ are the $\mathcal{B}$-coordinates of $\ad_\R^2(A)$ and $A$ by Lemma~\ref{lem:quad}, so~\eqref{eq:flow} is the stated projected double-commutator flow.
\end{proof}

This is the static-search counterpart of the double-bracket flows of Brockett and Chu discussed in Section~\ref{sec:discussion}: here $\R$ is fixed and $\ad_\R^2$ drives the generator, rather than a matrix evolving on its isospectral manifold. It gives a continuous-time alternative to the direct pencil solve of Theorem~\ref{thm:main}.

\subsection{Exact recovery of the symmetry generator}

The zero set of the pencil is more than the yes/no certificate of Corollary~\ref{cor:zero}: restricted to a Lie algebra of candidates, it returns the entire infinitesimal symmetry algebra of $\R$.

\begin{theorem}[Lie algebra of the covariance stabilizer]\label{thm:lie-recovery}
Let $G$ be a Lie group with unitary representation $\pi$ on $\mathcal{H}$ and finite-dimensional infinitesimal representation $\mathfrak{g}$, let $R$ be self-adjoint with stabilizer $G_R=\{g\in G:\pi(g)R\pi(g)^\ast=R\}$, and suppose $[R,A]$ is Hilbert--Schmidt for every $A\in\mathfrak{g}$, so that $\ad_R^2|_{\mathfrak{g}}$ and the quadratic-form identity of Lemma~\ref{lem:quad} are defined on $\mathfrak{g}$. Then the Lie algebra of $G_R$ is
\begin{equation}\label{eq:lie-recovery}
\mathfrak{g}_R=\{A\in\mathfrak{g}:[A,R]=0\}=\ker\bigl(\ad_R^2|_{\mathfrak{g}}\bigr).
\end{equation}
Equivalently, the zero eigenspace of the double-commutator operator restricted to $\mathfrak{g}$ is exactly the infinitesimal symmetry algebra of $R$.
\end{theorem}

\begin{proof}
If $A\in\mathfrak{g}_R$ then $e^{tA}Re^{-tA}=R$ for all small $t$, so differentiating at $t=0$ gives $[A,R]=0$. Conversely $[A,R]=0$ implies $e^{tA}Re^{-tA}=R$ for all $t$, so $A$ lies in the Lie algebra of $G_R$. By Lemma~\ref{lem:quad}, $\langle A,\ad_R^2 A\rangle_F=\Fro{[R,A]}^2$, so $\ad_R^2 A=0$ if and only if $[R,A]=0$; hence the kernel of $\ad_R^2|_{\mathfrak{g}}$ coincides with $\{A\in\mathfrak{g}:[A,R]=0\}$.
\end{proof}

The continuum counterpart is the exact recovery of a metaplectic generator, which makes Theorem~\ref{thm:lie-recovery} concrete for the time--frequency families of Section~\ref{sec:symplectic}.

\begin{theorem}[Exact recovery of a metaplectic generator]\label{thm:metaplectic}
Let $\mathcal{H}=L^2(\mathbb{R})$ and let $H_S$ be a self-adjoint quadratic Hamiltonian generating a one-parameter subgroup of the metaplectic representation of $Sp(2,\mathbb{R})$ (Section~\ref{sec:symplectic}), associated with a nonzero $S\in\mathfrak{sp}(2,\mathbb{R})$. Suppose $R=f(H_S)$ for a real Borel function $f$ injective on the spectrum of $H_S$, and let $\mathcal{B}\subset\mathfrak{sp}(2,\mathbb{R})$ be a finite-dimensional candidate space of generators $A$ for which $[R,A]$ is Hilbert--Schmidt, so that $\delta(A,R)$ and the identity of Lemma~\ref{lem:quad} are defined on $\mathcal{B}$, and whose commutant with $R$ is one-dimensional, $\{A\in\mathcal{B}:[A,R]=0\}=\spn\{H_S\}$. Then the minimizer of $\delta(A,R)$ over nonzero $A\in\mathcal{B}$ is $H_S$, up to scalar multiple, and the KLT of $R$ is the spectral decomposition of the recovered metaplectic generator: the Hermite--Gauss basis of Theorem~\ref{thm:symplectic} when $H_S$ is elliptic, and the continuous decomposition of Remark~\ref{rem:symplectic-cont} for the hyperbolic and parabolic classes.
\end{theorem}

\begin{proof}
Functional calculus gives $[R,H_S]=[f(H_S),H_S]=0$, so $H_S$ is a zero-residual candidate and the minimum value of $\delta$ is zero. By Lemma~\ref{lem:quad} the set of minimizers is $\ker(\ad_R^2|_{\mathcal{B}})=\{A\in\mathcal{B}:[A,R]=0\}=\spn\{H_S\}$, so every minimizer is a scalar multiple of $H_S$. Since $R=f(H_S)$ with $f$ injective on $\sigma(H_S)$, the spectral theorem gives $R$ and $H_S$ the same spectral resolution, so the KLT of $R$ is the spectral decomposition of $H_S$: an $L^2$ eigenbasis in the elliptic case (Theorem~\ref{thm:symplectic}) and a continuous decomposition in the hyperbolic and parabolic cases (Remark~\ref{rem:symplectic-cont}).
\end{proof}

Theorem~\ref{thm:metaplectic} is the continuous analogue of the graph-automorphism certificate (Theorem~\ref{thm:aut}): in both cases an injective spectral map of the structured operator leaves the commutant, and hence the recovered symmetry, unchanged.

\section{Properties of the Spectrum}\label{sec:spectrum}

The pencil eigenvalues collect the elementary invariants of the inverse problem: nonnegativity, the zero-residual certificate, a norm bound, and the degeneracy of white covariance.

\begin{theorem}[Spectral characterization]\label{thm:spectrum}
Let $\lambda_1\le\cdots\le\lambda_d$ be the eigenvalues of \eqref{eq:gevp}. Then
\begin{enumerate}[leftmargin=2.4em]
\item[(i)] $\lambda_k\ge0$ for all $k$;
\item[(ii)] $\lambda_1=0$ iff $\R$ commutes with some $A\in\mathcal{B}$;
\item[(iii)] $\lambda_d\le 4\lVert\R\rVert_2^2$;
\item[(iv)] if $\R=\sigma^2\mathbf{I}$ then $\lambda_k=0$ for all $k$.
\end{enumerate}
\end{theorem}

\begin{proof}
(i) is positive semidefiniteness of $\mathbf{M}$ (Proposition~\ref{prop:MGstructure}(i)); (ii) is Corollary~\ref{cor:zero}; (iii) follows from $\Fro{[\R,A]}\le2\lVert\R\rVert_2\Fro{A}$ on squaring; (iv) is immediate since $[\sigma^2\mathbf{I},A]=0$ for all $A$.
\end{proof}

Part (iv) is the harmonic-analysis statement that white noise has no preferred symmetry: every generator commutes with a scalar covariance, the spectrum is identically zero, and no transform is distinguished, consistent with the KLT being arbitrary for white processes.

\begin{proposition}[Total residual budget]\label{prop:sumrule}
The pencil eigenvalues satisfy the trace identity
\begin{equation}\label{eq:sumrule}
\sum_{k=1}^d\lambda_k=\operatorname{tr}\bigl(\mathbf{G}^{-1}\mathbf{M}\bigr).
\end{equation}
For a basis orthonormal in the Frobenius inner product ($\mathbf{G}=\mathbf{I}$) this reads $\sum_k\lambda_k=\sum_k\Fro{[\R,B_k]}^2$: the total commutator energy of the candidate basis is a fixed budget the eigenvalues partition, so a near-symmetry, with $\lambda_1$ close to zero, forces the remaining eigenvalues to carry that budget.
\end{proposition}

\begin{proof}
The eigenvalues of~\eqref{eq:gevp} are those of $\mathbf{G}^{-1}\mathbf{M}$, so $\sum_k\lambda_k=\operatorname{tr}(\mathbf{G}^{-1}\mathbf{M})$, the trace being basis-independent. For $\mathbf{G}=\mathbf{I}$,
\[
\sum_k\lambda_k=\operatorname{tr}\mathbf{M}=\sum_k M_{kk}=\sum_k\Fro{[\R,B_k]}^2
\]
by Lemma~\ref{lem:quad}.
\end{proof}

\section{Stability Under Noise}\label{sec:perturbation}

In practice $R$ is estimated, not known, and an exact symmetry is never observed; the operative question is how the eigenvalue problem of Section~\ref{sec:gevp} behaves when $R$ is replaced by $\widehat R=R+E$. We give two equivalent statements: a coordinate-free bound on the recovered generator (Theorem~\ref{thm:invstab}), and its computable pencil-coordinate refinement on the residual $\lambda_{\min}$. Write $\kappa_G=\|\mathbf{G}^{-1}\|_2$ and $C_{\mathcal{B}}=\sum_i\|B_i\|_F^2$.

\paragraph{The compressed double-commutator operator.} Although $\ad_R^2$ is self-adjoint on the full operator space, $\ad_R^2(\mathcal{B})$ need not lie in $\mathcal{B}$, so the object whose spectrum the pencil of Theorem~\ref{thm:main} computes is the compressed (Galerkin) operator
\begin{equation}\label{eq:compressed}
L_R^{\mathcal{B}}=P_{\mathcal{B}}\,\ad_R^2|_{\mathcal{B}}\colon\mathcal{B}\to\mathcal{B},
\end{equation}
with $P_{\mathcal{B}}$ the Frobenius-orthogonal projection onto $\mathcal{B}$. For $A,A'\in\mathcal{B}$ one has $\langle A,L_R^{\mathcal{B}}A'\rangle_F=\langle A,\ad_R^2 A'\rangle_F$, so $L_R^{\mathcal{B}}$ is self-adjoint and positive semidefinite on $\mathcal{B}$ with quadratic form $\langle A,L_R^{\mathcal{B}}A\rangle_F=\Fro{[R,A]}^2$ by Lemma~\ref{lem:quad}; its matrix in a basis $\{B_i\}$ of $\mathcal{B}$ is the pencil $(\mathbf{M},\mathbf{G})$, and its eigenvalues are the pencil eigenvalues. The stability statements below concern this compressed operator.

\begin{theorem}[Stability of the recovered generator]\label{thm:invstab}
Let $R$ be self-adjoint and let $L_R^{\mathcal{B}}=P_{\mathcal{B}}\,\ad_R^2|_{\mathcal{B}}$ be the compressed double-commutator operator on $\mathcal{B}$ of~\eqref{eq:compressed}, so that $\langle A,L_R^{\mathcal{B}} A\rangle_F=\Fro{[R,A]}^2$ by Lemma~\ref{lem:quad}. Suppose $L_R^{\mathcal{B}}$ has a one-dimensional kernel $\spn\{A_0\}$, $\Fro{A_0}=1$, separated from the rest of its spectrum by a gap $\gamma>0$. Let $\widehat R=R+E$ and let $\widehat A$ be a unit-norm minimizer of $\Fro{[\widehat R,A]}^2$ over $A\in\mathcal{B}$. If $8\|R\|_2\|E\|_2+4\|E\|_2^2<\gamma/2$, then
\begin{equation}\label{eq:invstab}
\sin\angle(\widehat A,A_0)\le\frac{8\|R\|_2\|E\|_2+4\|E\|_2^2}{\gamma-8\|R\|_2\|E\|_2-4\|E\|_2^2}.
\end{equation}
\end{theorem}

\begin{proof}
Since $\ad_R$ is self-adjoint under the Frobenius inner product (Lemma~\ref{lem:quad}), $L_R=\ad_R^2$ and $L_{\widehat R}=(\ad_R+\ad_E)^2$, so
\begin{equation}\label{eq:LR-decomp}
L_{\widehat R}-L_R=\ad_R\ad_E+\ad_E\ad_R+\ad_E^2 .
\end{equation}
The commutator superoperator satisfies $\|\ad_X\|\le 2\|X\|_2$ for any bounded $X$, since $\|\ad_X Y\|_F=\|XY-YX\|_F\le\|X\|_2\|Y\|_F+\|Y\|_F\|X\|_2=2\|X\|_2\|Y\|_F$. The data perturbation is $\|R-\widehat R\|_2=\|E\|_2$, and the induced perturbation of the superoperator follows from~\eqref{eq:LR-decomp} term by term, by submultiplicativity of the induced norm and the triangle inequality:
\begin{equation}\label{eq:LR-perturb}
\|L_{\widehat R}-L_R\|\le\|\ad_R\|\,\|\ad_E\|+\|\ad_E\|\,\|\ad_R\|+\|\ad_E\|^2\le 8\|R\|_2\|E\|_2+4\|E\|_2^2 .
\end{equation}
This is the explicit operator-norm link between the covariance error $\|E\|_2$ and the superoperator error $\|L_{\widehat R}-L_R\|$: the latter is first order in $\|E\|_2$ with coefficient $8\|R\|_2$, plus a second-order $4\|E\|_2^2$ term. Compression to $\mathcal{B}$ does not increase it, since $L_X^{\mathcal{B}}=P_{\mathcal{B}}L_X|_{\mathcal{B}}$ with $P_{\mathcal{B}}$ a contraction gives $\|L_{\widehat R}^{\mathcal{B}}-L_R^{\mathcal{B}}\|\le\|L_{\widehat R}-L_R\|$. The kernel $\spn\{A_0\}$ sits at the edge of $\sigma(L_R^{\mathcal{B}})$ (the eigenvalue $0$), separated from the remaining spectrum by $\gamma$, and $\widehat A$ is the bottom eigenvector of $L_{\widehat R}^{\mathcal{B}}$. By Weyl's inequality the perturbed remaining spectrum stays at distance at least $\gamma-\|L_{\widehat R}^{\mathcal{B}}-L_R^{\mathcal{B}}\|$ from $0$, so the Davis--Kahan $\sin\theta$ theorem, with $\|L_{\widehat R}^{\mathcal{B}}-L_R^{\mathcal{B}}\|<\gamma/2$, yields $\sin\angle(\widehat A,A_0)\le\|L_{\widehat R}^{\mathcal{B}}-L_R^{\mathcal{B}}\|/(\gamma-\|L_{\widehat R}^{\mathcal{B}}-L_R^{\mathcal{B}}\|)$. Since $t\mapsto t/(\gamma-t)$ is increasing on $[0,\gamma)$ and $\|L_{\widehat R}^{\mathcal{B}}-L_R^{\mathcal{B}}\|\le 8\|R\|_2\|E\|_2+4\|E\|_2^2$ by~\eqref{eq:LR-perturb}, substituting gives~\eqref{eq:invstab}.
\end{proof}

The bound \eqref{eq:invstab} is coordinate-free on $\mathcal{B}$. Its computable, pencil-coordinate refinement begins from the observation that $\mathbf{M}$ is the Gram matrix of the commutators.

\begin{lemma}\label{lem:gram}
$M_{ij}=\langle[R,B_i],[R,B_j]\rangle_F$, so $\Delta\mathbf{M}=\mathbf{M}(\widehat R)-\mathbf{M}(R)$ has entries
\[
\Delta M_{ij}=\langle[E,B_i],[R,B_j]\rangle_F+\langle[R,B_i],[E,B_j]\rangle_F+\langle[E,B_i],[E,B_j]\rangle_F .
\]
\end{lemma}

\begin{proof}
By self-adjointness of $\ad_R$ under the Frobenius inner product (Lemma~\ref{lem:quad}),
\[
M_{ij}=\langle B_i,\ad_R^2 B_j\rangle_F=\langle\ad_R B_i,\ad_R B_j\rangle_F=\langle[R,B_i],[R,B_j]\rangle_F.
\]
Substituting $[\widehat R,B_k]=[R,B_k]+[E,B_k]$ and expanding gives $\Delta\mathbf{M}$.
\end{proof}

\begin{theorem}[Eigenvalue stability]\label{thm:eigstab}
$\bigl|\lambda_{\min}(\widehat R)-\lambda_{\min}(R)\bigr|\le\kappa_G\,\|\Delta\mathbf{M}\|_2\le 4\kappa_G C_{\mathcal{B}}\,\|E\|_2\bigl(2\|R\|_2+\|E\|_2\bigr).$
\end{theorem}

\begin{proof}
The pencil eigenvalues are those of $\mathbf{G}^{-1/2}\mathbf{M}\mathbf{G}^{-1/2}$, perturbed by $\mathbf{G}^{-1/2}\Delta\mathbf{M}\,\mathbf{G}^{-1/2}$, a perturbation of spectral norm at most $\kappa_G\|\Delta\mathbf{M}\|_2$; Weyl's inequality gives the first bound. For the second, define $\Phi_X\colon\mathbf{c}\mapsto\sum_k c_k[X,B_k]$, so that $\mathbf{M}=\Phi_R^\ast\Phi_R$ and $\mathbf{M}(\widehat R)=\Phi_{\widehat R}^\ast\Phi_{\widehat R}$ with $\Phi_{\widehat R}=\Phi_R+\Phi_E$. By Cauchy--Schwarz and $\|[X,B_k]\|_F\le 2\|X\|_2\|B_k\|_F$,
\[
\|\Phi_X\mathbf{c}\|_F\le\sum_k|c_k|\,\|[X,B_k]\|_F\le 2\|X\|_2\Bigl(\textstyle\sum_k\|B_k\|_F^2\Bigr)^{1/2}\|\mathbf{c}\|_2,
\]
so $\|\Phi_X\|_2\le 2\|X\|_2\sqrt{C_{\mathcal{B}}}$. Then $\Delta\mathbf{M}=\Phi_R^\ast\Phi_E+\Phi_E^\ast\Phi_R+\Phi_E^\ast\Phi_E$ gives
\[
\|\Delta\mathbf{M}\|_2\le 2\|\Phi_R\|_2\|\Phi_E\|_2+\|\Phi_E\|_2^2\le 4C_{\mathcal{B}}\|E\|_2\bigl(2\|R\|_2+\|E\|_2\bigr).
\]
\end{proof}

\begin{theorem}[Quadratic residual at an exact symmetry]\label{thm:quad}
If $A^\ast\in\mathcal{B}$ satisfies $[R,A^\ast]=0$ (so $\lambda_{\min}(R)=0$), then
\begin{equation}
\lambda_{\min}(\widehat R)\;\le\;\frac{\|[E,A^\ast]\|_F^2}{\|A^\ast\|_F^2}\;\le\;4\,\|E\|_2^2 .
\end{equation}
The residual at a true symmetry is second order in the perturbation.
\end{theorem}

\begin{proof}
By Theorem~\ref{thm:main}, $\lambda_{\min}(\widehat R)=\min_{A\neq0}\|[\widehat R,A]\|_F^2/\|A\|_F^2\le\|[\widehat R,A^\ast]\|_F^2/\|A^\ast\|_F^2$. Since $[R,A^\ast]=0$, $[\widehat R,A^\ast]=[E,A^\ast]$, and $\|[E,A^\ast]\|_F\le 2\|E\|_2\|A^\ast\|_F$ gives the second inequality.
\end{proof}

\begin{theorem}[Generator stability]\label{thm:genstab}
Suppose the pencil at $R$ has a spectral gap $\gamma=\lambda_2-\lambda_{\min}>0$, and let $\eta=\kappa_G\|\Delta\mathbf{M}\|_2<\gamma$. Then the optimal coefficient vectors satisfy the Davis--Kahan bound
\begin{equation}
\sin\angle\bigl(\mathbf{c}^\ast(\widehat R),\mathbf{c}^\ast(R)\bigr)\le\frac{\eta}{\gamma-\eta},
\end{equation}
and the recovered generator $A^\ast=\sum_k c_k^\ast B_k$ inherits this bound up to the basis conditioning $\kappa_G$, in agreement with the coordinate-free Theorem~\ref{thm:invstab}. The recovered symmetry is stable when the matched group is spectrally separated from its competitors and ill-conditioned when several symmetries nearly coincide.
\end{theorem}

\begin{proof}
Apply the Davis--Kahan $\sin\theta$ theorem to the symmetric matrix $\mathbf{G}^{-1/2}\mathbf{M}\mathbf{G}^{-1/2}$ and its perturbation of norm at most $\eta$, with eigenvalue gap $\gamma$ at the smallest eigenvalue; transferring from $\mathbf{G}^{-1/2}\mathbf{c}$ back to $\mathbf{c}$ and to $A^\ast$ costs the conditioning factor.
\end{proof}

\begin{corollary}[Sample complexity]\label{cor:sample}
If $\widehat R$ satisfies $\|\widehat R-R\|_2\le\varepsilon$, then $|\lambda_{\min}(\widehat R)-\lambda_{\min}(R)|=O(\|R\|_2\,\varepsilon)$ in general, $\lambda_{\min}(\widehat R)=O(\varepsilon^2)$ at an exact symmetry, and the recovered generator moves by $O(\|R\|_2\,\varepsilon/\gamma)$. For a sub-Gaussian sample covariance from $N$ observations in dimension $p$, $\varepsilon\asymp\|R\|_2\sqrt{p/N}$ with high probability~\cite{vershynin2018}, so the acceptance threshold $\tau$ should be set at order $\|R\|_2\varepsilon$: above the second-order floor of a true symmetry and below the first-order separation of a false one.
\end{corollary}

\begin{remark}[Why closest-match selection works]\label{rem:closest}
Theorem~\ref{thm:quad} is the reason matching an empirical covariance to its nearest group is reliable: a genuine symmetry produces a residual that shrinks \emph{quadratically} with the estimation error, so it sits well below the residual of any non-symmetry, which stays within $O(\varepsilon)$ of its positive exact value. Theorem~\ref{thm:genstab} is the caveat: when two candidate symmetries are nearly tied (small gap $\gamma$), the recovered generator is ill-conditioned and the selection should report the tie rather than commit. This near-degenerate case, two structures fitting almost equally well, is exactly the regime met when no exact match exists and the closest is chosen.
\end{remark}

\begin{proposition}[The identifiability gap]\label{prop:idgap}
Define the identifiability gap of the pencil~\eqref{eq:gevp} as $g=\lambda_2-\lambda_1\ge0$. Then:
\begin{enumerate}[leftmargin=2.4em]
\item[(i)] $g>0$ iff the residual-minimizing generator is unique up to scale; $g=0$ iff the minimal residual is attained on a subspace of dimension at least two, in which case the matched symmetry is determined only up to that subspace.
\item[(ii)] Under a covariance perturbation with $\eta=\kappa_G\|\Delta\mathbf{M}\|_2<g$, the recovered generator obeys
\[
\sin\angle\bigl(\mathbf{c}^\ast(\widehat R),\mathbf{c}^\ast(R)\bigr)\le\frac{\eta}{g-\eta},
\]
which is Theorem~\ref{thm:genstab} with $\gamma=g$; the bound degrades as $g\to0$.
\end{enumerate}
Thus $g$ is the margin by which the matched symmetry beats its nearest competitor: the acceptance threshold $\tau$ of Corollary~\ref{cor:sample} certifies a unique match when it lies below $g$, and the tie-reporting of Remark~\ref{rem:closest} is the regime $g\lesssim\tau$.
\end{proposition}

\begin{proof}
(i) The minimizer of $\rho$ is unique up to scale iff the smallest eigenvalue of $\mathbf{G}^{-1/2}\mathbf{M}\mathbf{G}^{-1/2}$ is simple, i.e.\ $\lambda_1<\lambda_2$; otherwise the minimal eigenspace has dimension equal to the multiplicity of $\lambda_1$. (ii) is Theorem~\ref{thm:genstab} with $\gamma=\lambda_2-\lambda_1=g$.
\end{proof}

\section{The Permutation Commutator Formula and Graph Automorphisms}\label{sec:pcf}

For permutation generators the residual has a transparent spectral form. The results of this section refine the framework when $\mathcal{B}$ is spanned by permutation matrices; they were first announced, in a signal-processing context, in the companion work~\cite{thornton2026algebraicdiversitygrouptheoreticspectral}, and are recorded here self-contained from the eigenvalue setup of Section~\ref{sec:gevp}.

\begin{proposition}[Eigenvalue-difference representation]\label{prop:eigdiff}
If $\R=\mathrm{diag}(\lambda_1,\ldots,\lambda_M)$ then for any $A\in\mathbb{C}^{M\times M}$,
\begin{equation}
\Fro{[\R,A]}^2=\sum_{j\neq k}(\lambda_j-\lambda_k)^2\,|A_{jk}|^2 .
\end{equation}
\end{proposition}

\begin{proof}
$[\R,A]_{jk}=(\lambda_k-\lambda_j)A_{jk}$; squaring and summing gives the result.
\end{proof}

\begin{corollary}[Permutation commutator formula]\label{cor:pcf}
For a permutation matrix $\mathbf{P}_\sigma$ expressed in the eigenbasis of $\R$,
\begin{equation}\label{eq:pcf}
\Fro{[\mathbf{P}_\sigma,\R]}^2=\sum_{k=1}^{M}(\lambda_k-\lambda_{\sigma(k)})^2 .
\end{equation}
The residual equals the total squared eigenvalue displacement along the orbits of $\sigma$.
\end{corollary}

The identity incurs zero residual; a transposition of $i,j$ incurs $2(\lambda_i-\lambda_j)^2$; a cyclic shift incurs $\sum_k(\lambda_k-\lambda_{k+1})^2$. Minimizing over permutations is thus an eigenvalue-assignment problem. The clean boundary case $\delta=0$ for graph-structured covariance is the following.

\begin{theorem}[Automorphism characterization]\label{thm:aut}
Let $\mathcal{G}$ be a graph with Laplacian $\mathbf{L}$, and let $\R=f(\mathbf{L})$ for any $f$ injective on the spectrum of $\mathbf{L}$ (for instance the diffusion covariance $(\mathbf{I}+\mathbf{L})^{-1}$ or the heat kernel $\exp(-t\mathbf{L})$). Then for every $\sigma\in S_M$,
\begin{equation}
\delta(\mathbf{P}_\sigma,\R)=0\quad\Longleftrightarrow\quad \sigma\in\Aut(\mathcal{G}).
\end{equation}
\end{theorem}

\begin{proof}
($\Leftarrow$) $\sigma\in\Aut(\mathcal{G})$ gives $\mathbf{P}_\sigma\mathbf{L}\mathbf{P}_\sigma^T=\mathbf{L}$, hence $[\mathbf{P}_\sigma,f(\mathbf{L})]=0$. ($\Rightarrow$) Injectivity of $f$ on the spectrum makes the commutant of $f(\mathbf{L})$ equal to that of $\mathbf{L}$, so $[\mathbf{P}_\sigma,\mathbf{L}]=0$ and $\sigma\in\Aut(\mathcal{G})$.
\end{proof}

\begin{remark}[Scope of the equivalence]\label{rem:commutant-aut}
The equivalence holds for permutation generators and a spectrally injective covariance map, and both hypotheses carry weight. The inclusion $\Aut(\mathcal{G})\subseteq\{A:[A,\R]=0\}$ is unconditional, since every automorphism commutes with $\mathbf{L}$ and hence with $f(\mathbf{L})$. It is the reverse inclusion that requires the hypotheses. The commutant $\{A:[A,\R]=0\}$ is the full algebra of operators preserving the eigenspaces of $\R$, which for a graph with repeated Laplacian eigenvalues is strictly larger than $\Aut(\mathcal{G})$ and contains non-permutation, and even continuous, symmetries; intersecting it with the permutation matrices is what cuts it down to $\Aut(\mathcal{G})$. Likewise, if $f$ is not injective on the spectrum of $\mathbf{L}$, commuting with $f(\mathbf{L})$ no longer forces commuting with $\mathbf{L}$. The pencil on a permutation basis with a spectrally injective $\R$ enforces both restrictions, and only under them does $\delta=0$ certify an automorphism rather than a more general commuting operator.
\end{remark}

That automorphisms commute with the Laplacian is classical algebraic graph theory~\cite{biggs1993,godsilroyle2001}; Theorem~\ref{thm:aut} is the statement that the commutant survives any spectrally injective covariance map, so the eigenvalue problem \eqref{eq:gevp} on a permutation basis has $\lambda_{\min}=0$ exactly when the basis contains an automorphism. The complementary problem of deciding isomorphism between two graphs, addressed by colour refinement~\cite{weisfeilerleman1968} and Babai's quasipolynomial algorithm~\cite{babai2016}, is not directly comparable; here the graph is fixed and its automorphisms are read off from $\delta$.

\paragraph{Numerical illustration.}
Figure~\ref{fig:recovery} illustrates the boundary case $\delta=0$ of Corollary~\ref{cor:pcf} and Theorem~\ref{thm:aut} across four covariance types, each selecting a different generator: a directional process the cyclic shift, a persymmetric AR(1) the reflection, a heteroscedastic diagonal the modulation, and a permutation-invariant covariance that permutation.

\begin{figure}[H]
\centering
\includegraphics[width=0.96\textwidth]{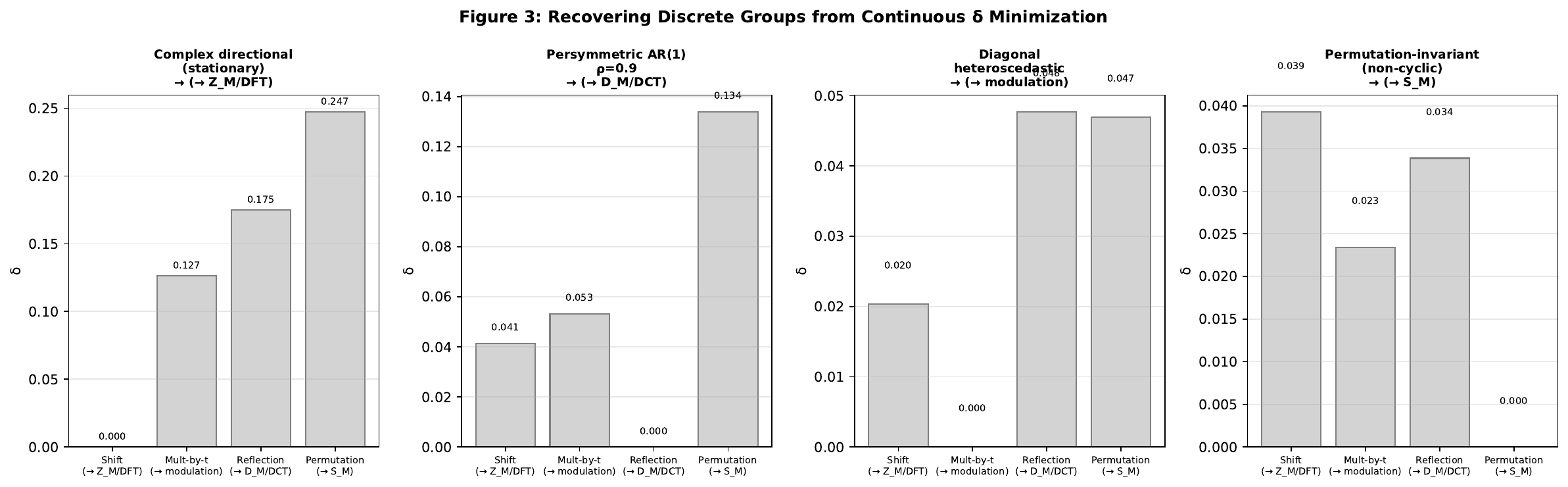}
\caption{Recovering discrete structure from the residual $\delta$ (Corollary~\ref{cor:pcf}, Theorem~\ref{thm:aut}). For four covariance types each generator's residual is computed and exactly one attains $\delta=0$: a complex directional process selects the cyclic shift (Fourier); a real persymmetric AR(1) covariance selects reflection (the even/odd parity split), not the shift, because a symmetric Toeplitz matrix is persymmetric; a diagonal heteroscedastic covariance selects modulation; and a covariance invariant under a fixed non-cyclic permutation selects that permutation.}
\label{fig:recovery}
\end{figure}

Table~\ref{tab:graphaut} applies Theorem~\ref{thm:aut} to six graphs with diffusion covariance $\R=(\mathbf{I}+\mathbf{L})^{-1}$ and a fixed candidate set of five permutations (cyclic shift, reflection, a transposition, a block swap, a $3$-cycle). In every case the residual-minimizing generator attains $\delta=0$ to machine precision and is a verified automorphism, and the recovered generator reflects the structure of $\Aut(\mathcal{G})$.

\begin{table}[H]
\centering
\small
\begin{tabular}{lcccc}
\toprule
\textbf{Graph} & $|\Aut(\mathcal{G})|$ & \textbf{Best generator} & $\delta_{\min}$ & \textbf{Automorphism?} \\
\midrule
cycle $C_6$ & $12$ & reflection & $<10^{-15}$ & yes \\
complete $K_4$ & $24$ & transposition $(1\,2)$ & $<10^{-15}$ & yes \\
path $P_6$ & $2$ & reflection & $<10^{-15}$ & yes \\
triangular prism & $12$ & block swap & $<10^{-15}$ & yes \\
complete bipartite $K_{3,3}$ & $72$ & transposition $(1\,2)$ & $0$ & yes \\
complete $K_3$ & $6$ & cyclic shift & $0$ & yes \\
\bottomrule
\end{tabular}
\caption{Graph-automorphism detection by the residual $\delta$ (Theorem~\ref{thm:aut}), diffusion covariance $\R=(\mathbf{I}+\mathbf{L})^{-1}$. In all six cases the best generator achieves $\delta=0$ and is an automorphism. The reproducibility script \texttt{graph\_automorphism\_table.py} is included with the manuscript.}
\label{tab:graphaut}
\end{table}

\section{Sequential Generator Recovery for Non-Abelian Symmetry}\label{sec:sequential}

Theorem~\ref{thm:main} returns a single optimal generator. For Abelian symmetry this suffices: $A^\ast$ generates the group by powers (finite order) or exponential lift ($A^\ast\in\mathfrak{u}(M)$). A non-Abelian symmetry group needs several generators, which a single eigenvalue solve cannot supply. We formalize a sequential procedure that deflates the search space against the subgroup found so far. Write
\begin{equation}\label{eq:autR}
\Aut(\R)=\{\sigma\in S_M:\mathbf{P}_\sigma\R=\R\mathbf{P}_\sigma\},
\end{equation}
identified with $\Aut(\mathcal{G})$ when $\R=f(\mathbf{L})$ by Theorem~\ref{thm:aut}.

\paragraph{Procedure (sequential eigenvalue problem with deflation).}
Initialize $G_0=\{e\}$, $k=0$, with acceptance threshold $\tau\ge0$ and iteration cap $K_{\max}$.
\begin{enumerate}[leftmargin=2.6em,label=(S\arabic*)]
\item Form the deflated basis $\mathcal{B}^{\perp G_k}$, the Frobenius-orthogonal complement of $\spn\{\mathbf{P}_g:g\in G_k\}$ within $\spn(\mathcal{B})$.
\item Solve \eqref{eq:gevp} on $\mathcal{B}^{\perp G_k}$ for the optimal $A^\ast_{k+1}$ with $\Fro{A^\ast_{k+1}}=1$.
\item Round to the nearest permutation, $\sigma^\ast_{k+1}=\arg\max_{\sigma}\langle A^\ast_{k+1},\mathbf{P}_\sigma\rangle_F$, by the Hungarian algorithm~\cite{kuhn1955} in $O(M^3)$.
\item If $\delta(\mathbf{P}_{\sigma^\ast_{k+1}},\R)>\tau$, terminate and return $G_k$.
\item Otherwise set $G_{k+1}=\langle G_k,\sigma^\ast_{k+1}\rangle$, increment $k$, and repeat unless $k\ge K_{\max}$.
\end{enumerate}
Per-iteration cost is $O(d^2M^3+d^3+|G_k|^2M^2+M^3)$, polynomial in $d$, $M$, and $|G_k|$.

\begin{lemma}[Deflation orthogonality]\label{lem:deforth}
For every $k$, every $A\in\spn(\mathcal{B}^{\perp G_k})$ satisfies $\langle A,\mathbf{P}_g\rangle_F=0$ for all $g\in G_k$; in particular $\langle A^\ast_{k+1},\mathbf{P}_g\rangle_F=0$.
\end{lemma}

\begin{proof}
$\mathcal{B}^{\perp G_k}$ is by definition orthogonal to every $\mathbf{P}_g$, $g\in G_k$, and the property passes to linear combinations, including $A^\ast_{k+1}$.
\end{proof}

\begin{lemma}[Forward progress]\label{lem:forward}
If $\max_\sigma\langle A^\ast_{k+1},\mathbf{P}_\sigma\rangle_F>0$ and the procedure accepts $\sigma^\ast_{k+1}$, then $\sigma^\ast_{k+1}\notin G_k$.
\end{lemma}

\begin{proof}
By Lemma~\ref{lem:deforth}, $\langle A^\ast_{k+1},\mathbf{P}_g\rangle_F=0$ for $g\in G_k$, whereas $\langle A^\ast_{k+1},\mathbf{P}_{\sigma^\ast_{k+1}}\rangle_F=\max_\sigma\langle A^\ast_{k+1},\mathbf{P}_\sigma\rangle_F>0$. Hence $\mathbf{P}_{\sigma^\ast_{k+1}}\neq\mathbf{P}_g$, and since $\sigma\mapsto\mathbf{P}_\sigma$ is injective, $\sigma^\ast_{k+1}\notin G_k$.
\end{proof}

The positive-overlap hypothesis holds whenever the deflated subspace contains a matrix with a positive entry, since $\langle A,\mathbf{P}_\sigma\rangle_F=\sum_iA_{i,\sigma(i)}$; it is routine to verify for bases of permutation matrices and of permutation-difference matrices $\mathbf{P}_\sigma-\mathbf{I}$.

\begin{theorem}[Strict subgroup growth]\label{thm:growth}
Under the hypothesis of Lemma~\ref{lem:forward}, at every accepted iteration $|G_{k+1}|>|G_k|$.
\end{theorem}

\begin{proof}
$\sigma^\ast_{k+1}\notin G_k$ by Lemma~\ref{lem:forward}, and $G_{k+1}=\langle G_k,\sigma^\ast_{k+1}\rangle$ contains $G_k$ and $\sigma^\ast_{k+1}$, so $G_k\subsetneq G_{k+1}$.
\end{proof}

\begin{corollary}[Iteration bound]\label{cor:iter}
The procedure performs at most $K\le\lceil\log_2|G_K|\rceil$ accepted iterations; in particular $K\le\lceil\log_2 M!\rceil=O(M\log M)$.
\end{corollary}

\begin{proof}
By Lagrange's theorem the proper inclusion $G_k\subsetneq G_{k+1}$ gives index at least $2$, so $|G_{k+1}|\ge2|G_k|$; iterating, $|G_K|\ge2^K$. The bound $|G_K|\le M!$ and Stirling's estimate give $O(M\log M)$.
\end{proof}

\begin{theorem}[Convergence at $\tau=0$]\label{thm:conv}
Run with $\tau=0$, every accepted $\sigma^\ast_{k+1}\in\Aut(\R)$, hence $G_K\subseteq\Aut(\R)$.
\end{theorem}

\begin{proof}
Acceptance at $\tau=0$ means $\delta(\mathbf{P}_{\sigma^\ast_{k+1}},\R)=0$, i.e.\ $\mathbf{P}_{\sigma^\ast_{k+1}}\R=\R\mathbf{P}_{\sigma^\ast_{k+1}}$, the defining condition $\sigma^\ast_{k+1}\in\Aut(\R)$ of \eqref{eq:autR}. As $\Aut(\R)$ is a group, it contains $G_K=\langle\sigma^\ast_1,\ldots,\sigma^\ast_K\rangle$.
\end{proof}

\begin{theorem}[Full recovery under generating completeness]\label{thm:full-recovery}
Let $\R=f(\mathbf{L})$ with $f$ injective on $\operatorname{spec}(\mathbf{L})$, and suppose the positive-overlap hypothesis of Lemma~\ref{lem:forward} holds at every stage. Call the candidate basis $\mathcal{B}$ \emph{generating-complete} for $\R$ if, at every stage with $G_k\subsetneq\Aut(\R)$, the rounding step~(S3) returns a permutation $\sigma^\ast_{k+1}\in\Aut(\R)\setminus G_k$. If $\mathcal{B}$ is generating-complete, then the sequential procedure at $\tau=0$ halts with $G_K=\Aut(\R)$, the full automorphism group.
\end{theorem}

\begin{proof}
By Theorem~\ref{thm:conv} every accepted generator lies in $\Aut(\R)$, so $G_k\subseteq\Aut(\R)$ at every stage. While $G_k\subsetneq\Aut(\R)$, generating completeness returns $\sigma^\ast_{k+1}\in\Aut(\R)\setminus G_k$, so $\delta(\mathbf{P}_{\sigma^\ast_{k+1}},\R)=0=\tau$ by Theorem~\ref{thm:aut}; the acceptance test~(S4) passes and $G_{k+1}=\langle G_k,\sigma^\ast_{k+1}\rangle$ strictly contains $G_k$ by Theorem~\ref{thm:growth}, while remaining inside $\Aut(\R)$. The chain $G_0\subsetneq G_1\subsetneq\cdots$ cannot grow without bound inside the finite group $\Aut(\R)$, so it reaches $G_k=\Aut(\R)$ after finitely many accepted steps. At that stage Lemma~\ref{lem:deforth} makes the deflated minimizer Frobenius-orthogonal to every $\mathbf{P}_g$, $g\in\Aut(\R)$, so under the positive-overlap hypothesis its nearest permutation has positive overlap and therefore lies outside $\Aut(\R)$; by Theorem~\ref{thm:aut} that permutation has $\delta>0=\tau$, and~(S4) terminates, returning $G_K=\Aut(\R)$.
\end{proof}

\begin{remark}[Soundness versus completeness]\label{rem:completeness}
Generating completeness is precisely the gap between the two properties. Theorem~\ref{thm:conv} gives the inclusion $G_K\subseteq\Aut(\R)$ unconditionally, from the soundness of each accepted generator; generating completeness promotes it to the equality $G_K=\Aut(\R)$. The hypothesis is a condition on the interaction of the candidate basis, the deflation, and the rounding, not a consequence of them, and it can fail. The six-cycle of Appendix~\ref{app:sequential} uses a basis for which the rounded second residual misses the reflection $\rho\in D_6$, so the procedure halts at the proper subgroup $\langle\tau\rangle\subsetneq D_6$. Certifying generating completeness for a given observation geometry is the algorithmic question pursued in the companion development.
\end{remark}

\section{Recovery of Hidden Transforms}\label{sec:prolate}

The transforms recovered so far follow from symmetries that are, in hindsight, visible in the covariance: a cyclic shift, a reflection, a phase-space rotation. We close the development with several cases in which the matched generator is a genuine hidden symmetry. It is not a geometric symmetry of the covariance, and historically it was found only by an ad hoc calculation. The double-commutator eigenvalue problem recovers it as the solution of a variational problem, with no prior knowledge that it exists.

Consider the discrete time-and-band-limiting operator on $\mathbb{C}^N$,
\begin{equation}\label{eq:prolateC}
\R_{mn}=\frac{\sin\bigl(2\pi W(m-n)\bigr)}{\pi(m-n)}\quad(m\neq n),\qquad \R_{nn}=2W,
\end{equation}
a real symmetric Toeplitz matrix whose eigenvectors are the discrete prolate spheroidal sequences (DPSS), the length-$N$ sequences that maximize energy concentration in the band $[-W,W]$~\cite{slepian1978}. The DPSS are the optimal tapers of multitaper spectral estimation and a standard tool in band-limited signal analysis. The spectrum of $\R$ clusters near $0$ and $1$, so its eigenvectors are numerically ill-determined; in practice the DPSS are computed instead from a symmetric tridiagonal matrix
\begin{equation}\label{eq:slepianT}
T_{nn}=\Bigl(\tfrac{N-1-2n}{2}\Bigr)^2\cos(2\pi W),\qquad T_{n,n+1}=T_{n+1,n}=\tfrac12(n+1)(N-1-n),
\end{equation}
which commutes with $\R$ and has a simple, well-separated spectrum~\cite{slepian1978,percival1993}. The existence of this commuting operator is the prolate spheroidal property, and it is not predicted by any symmetry of $\R$. Across time-band limiting, random matrix theory, and integrable systems such commuting pairs were, before the bispectral theory that now organizes them, constructed by ad hoc methods that worked because a commuting operator of low order could be found by direct calculation~\cite{duistermaat1986,casper2023}. We show that the matched-generator pencil of Section~\ref{sec:gevp} recovers $T$ directly.

\subsection{The variational characterization of a hidden symmetry}

We first record when the pencil detects a commuting operator beyond the trivial one. The identity always commutes with $\R$ and contributes a trivial kernel direction; a hidden symmetry is a second, independent kernel generator.

\begin{proposition}[Variational characterization of the prolate property]\label{prop:prolate}
Let $\R$ be self-adjoint on $\mathbb{C}^M$ and let $\mathcal{B}$ be a candidate generator space with $\mathbf{I}\in\mathcal{B}$. Order the pencil eigenvalues~\eqref{eq:gevp} as $0=\lambda_1\le\lambda_2\le\cdots\le\lambda_d$. Then:
\begin{enumerate}[leftmargin=2.4em]
\item[(i)] $\lambda_1=0$, with $\mathbf{I}$ spanning a kernel direction.
\item[(ii)] $\lambda_2=0$ if and only if $\mathcal{B}$ contains a generator $A$ linearly independent of $\mathbf{I}$ with $[\R,A]=0$. If $\mathcal{B}$ is a structurally restricted class, such as the banded matrices of fixed bandwidth, and such an $A$ is not realized by a low-degree polynomial in $\R$, then $\R$ has the prolate spheroidal property relative to $\mathcal{B}$: a structurally distinct operator sharing the eigenbasis of $\R$.
\item[(iii)] \textup{(Resolution.)} If a kernel generator $A$ has simple spectrum, its eigenbasis is the unique common eigenbasis of $\{\R,A\}$, hence a KLT of $\R$, even on eigenspaces where $\R$ is degenerate. Thus $A$ resolves the KLT wherever the spectrum of $\R$ leaves it undetermined.
\end{enumerate}
\end{proposition}

\begin{proof}
(i) $[\R,\mathbf{I}]=0$, so $\mathbf{I}\in\ker(\ad_\R^2|_{\mathcal{B}})$ and $\lambda_1=0$ by Theorem~\ref{thm:spectrum}(ii). (ii) By Lemma~\ref{lem:quad} the kernel of $\ad_\R^2|_{\mathcal{B}}$ is $\{A\in\mathcal{B}:[\R,A]=0\}$, so $\lambda_2=0$ iff that kernel has dimension at least two, i.e.\ contains a generator independent of $\mathbf{I}$. When $\R$ has simple spectrum every commuting generator is a polynomial in $\R$; the prolate property is then the case in which a structurally simple member of $\mathcal{B}$, such as a banded matrix, is realized only by a high-degree polynomial in $\R$, so its distinctness is one of structure and conditioning rather than of algebra. (iii) If $[\R,A]=0$ and $A$ has simple spectrum, its eigenvectors are determined up to phase and, commuting with $\R$, are eigenvectors of $\R$; on a degenerate eigenspace of $\R$ the distinct eigenvalues of $A$ select a unique orthonormal basis.
\end{proof}

Part (iii) is the operative point. A classical symmetry of $\R$ that squares to the identity, such as a reflection, has only the eigenvalues $\pm1$ and resolves $\R$ no further than a two-block split; a kernel generator with simple spectrum resolves the entire basis. The prolate property, in this language, is the condition $\lambda_2=0$ together with a simple-spectrum kernel generator.

\begin{theorem}[Uniqueness of the tridiagonal commuting generator]\label{thm:tridiag-unique}
Let $\R$ be self-adjoint on $\mathbb{C}^M$ with simple spectrum, and suppose some irreducible symmetric tridiagonal matrix $T$, with nonzero off-diagonal entries, satisfies $[\R,T]=0$. Then the commutant of $\R$ within the symmetric tridiagonal matrices $\mathcal{T}$ is exactly
\begin{equation}\label{eq:tridiag-commutant}
\{A\in\mathcal{T}:[\R,A]=0\}=\spn\{\mathbf{I},T\}.
\end{equation}
In particular every element of this commutant independent of $\mathbf{I}$ equals $\alpha T+\beta\mathbf{I}$ with $\alpha\neq0$; it has the eigenbasis of $T$, and that spectrum being simple, recovers the eigenbasis uniquely as a KLT of $\R$.
\end{theorem}

\begin{proof}
Because $\R$ has simple spectrum, every $A$ with $[\R,A]=0$ is a polynomial in $\R$, hence diagonal in the common eigenbasis $\{v_k\}_{k=0}^{M-1}$ of $\R$ and $T$; write $A=\sum_k a_k\,v_kv_k^{*}$. As $T$ is an irreducible Jacobi matrix its eigenvalues $\mu_0,\ldots,\mu_{M-1}$ are simple and its eigenvector entries are $(v_k)_i=\sqrt{w_k}\,p_i(\mu_k)$, where $p_0,\ldots,p_{M-1}$ are the orthonormal polynomials of the discrete measure $\nu=\sum_k w_k\delta_{\mu_k}$ with weights $w_k=(v_k)_0^2$, normalized by the Gauss quadrature identity $\sum_k w_k\,p_i(\mu_k)p_j(\mu_k)=\delta_{ij}$. Let $f$ be the polynomial of degree at most $M-1$ with $f(\mu_k)=a_k$. Then
\begin{equation}\label{eq:Aij-quad}
A_{ij}=\sum_k a_k(v_k)_i(v_k)_j=\sum_k f(\mu_k)\,w_k\,p_i(\mu_k)p_j(\mu_k)=\langle f p_i,\,p_j\rangle_\nu.
\end{equation}
If $\deg f=d$, then by orthogonality $\langle f p_i,p_j\rangle=\langle p_i,f p_j\rangle=0$ whenever $|i-j|>d$, so $A$ is $(2d+1)$-banded; and for $i+d\le M-1$ the entry $A_{i,i+d}=\langle f p_i,p_{i+d}\rangle=\kappa_f\,k_i/k_{i+d}$ is nonzero, where $\kappa_f$ is the leading coefficient of $f$ and $k_i$ that of $p_i$, so the bandwidth is exactly $2d+1$. Hence $A$ is tridiagonal if and only if $d\le1$, that is $a_k=\alpha\mu_k+\beta$ and $A=\alpha T+\beta\mathbf{I}$. This proves~\eqref{eq:tridiag-commutant}; the final statement follows from the simplicity of the spectrum of $T$ with Proposition~\ref{prop:prolate}(iii).
\end{proof}

The theorem is what turns minimization into recovery. On the tridiagonal search space the pencil kernel is forced to equal $\spn\{\mathbf{I},T\}$, so its second null vector is $T$ up to an affine normalization, not merely a residual minimizer, and its eigenbasis is the transform. The bandwidth of the search space sets the admissible polynomial degree through~\eqref{eq:Aij-quad}: over the pentadiagonal matrices the commutant gains the direction $T^2$ and becomes three-dimensional, which is the precise sense in which the recovered generator depends on the candidate class. Simple spectrum of $\R$ is the load-bearing hypothesis; where $\R$ has a repeated eigenvalue its commutant is only block diagonal, the tridiagonal commutant can be larger, and one falls back on the degeneracy resolution of Proposition~\ref{prop:prolate}(iii), as in the hierarchical example of Section~\ref{sec:wavelet}.

\begin{theorem}[Jacobi commutant recovery]\label{thm:jacobi-recovery}
Let $J$ be an irreducible symmetric tridiagonal (Jacobi) matrix on $\mathbb{C}^M$, with nonzero off-diagonal entries and hence simple spectrum $\mu_0,\ldots,\mu_{M-1}$, and let $\R=f(J)$ for a real-valued $f$ injective on $\{\mu_k\}$. Then $\R$ is self-adjoint with simple spectrum and $[\R,J]=0$, and the double-commutator pencil $\ad_\R^2(A)=\lambda A$ over the symmetric tridiagonal matrices $\mathcal{T}$ has kernel exactly $\spn\{\mathbf{I},J\}$, so its eigenvalues obey $0=\lambda_1=\lambda_2<\lambda_3$. The nontrivial null vector is $J$ up to affine normalization, and its eigenbasis is the orthogonal-polynomial transform generated by $J$, with columns $(v_k)_i=\sqrt{w_k}\,p_i(\mu_k)$ given by the orthonormal polynomials $p_i$ of the spectral measure $\nu=\sum_k w_k\delta_{\mu_k}$; this is the unique KLT of $\R$.
\end{theorem}

\begin{proof}
Irreducibility makes the eigenvalues $\mu_k$ of $J$ simple, and injectivity of $f$ on $\{\mu_k\}$ makes the $f(\mu_k)$ distinct, so $\R=f(J)$ is self-adjoint with simple spectrum; functional calculus gives $[\R,J]=[f(J),J]=0$. The hypotheses of Theorem~\ref{thm:tridiag-unique} thus hold with $T=J$, so the tridiagonal commutant $\{A\in\mathcal{T}:[\R,A]=0\}$ equals $\spn\{\mathbf{I},J\}$. By Lemma~\ref{lem:quad} this commutant is the kernel of $\ad_\R^2|_{\mathcal{T}}$, which is therefore two-dimensional, giving $0=\lambda_1=\lambda_2<\lambda_3$. Every kernel element independent of $\mathbf{I}$ is $\alpha J+\beta\mathbf{I}$ with $\alpha\ne0$; the eigenvector formula and the identification of the eigenbasis as the unique KLT are those of Theorem~\ref{thm:tridiag-unique} and Proposition~\ref{prop:prolate}(iii), using the simplicity of the spectrum of $J$.
\end{proof}

The three exact recoveries below are its instances. The band-limiting covariance~\eqref{eq:prolateC} shares the eigenbasis of Slepian's irreducible tridiagonal $T$~\eqref{eq:slepianT} and has simple spectrum, so it is $f(T)$ for an injective $f$; the Hahn ensemble is $\R=e^{-tL}$ with the Jacobi matrix $L$ and $f(\theta)=e^{-t\theta}$ injective; and the discrete cosine transform (Section~\ref{sec:dct}) is $\R=e^{-tC}$ with the cosine generator $C$ and the same injective $f$. In each, Theorem~\ref{thm:jacobi-recovery} returns the generator and its transform.

\subsection{The prolate spheroidal (Slepian) basis}

The band-limiting covariance~\eqref{eq:prolateC} does have a classical symmetry. It is centrosymmetric, $\mathbf{J}\R\mathbf{J}=\R$ with $\mathbf{J}$ the exchange matrix $\mathbf{J}_{ij}=\delta_{i,N-1-j}$, so $\delta(\mathbf{J},\R)=0$ and $\mathbf{J}$ is a matched generator that a permutation search of the kind in Section~\ref{sec:pcf} would find. But $\mathbf{J}^2=\mathbf{I}$ has spectrum $\{\pm1\}$, so commuting with $\mathbf{J}$ only block-diagonalizes $\R$ into its even and odd subspaces. The reflection is exact but coarse: a two-block reduction, not the transform.

Enlarging the candidate space to the symmetric tridiagonal matrices and solving the pencil $\ad_\R^2(A)=\lambda A$ over that space changes the picture. With $N=48$ and $W=0.1$, the zero-residual subspace is two-dimensional, $\lambda_1=\lambda_2=0$ (both below $10^{-15}$) with a gap of ten orders of magnitude to $\lambda_3\approx5\times10^{-6}$: the prolate property in the form $\lambda_2=0$ of Proposition~\ref{prop:prolate}. The two-dimensionality is not a numerical accident. The band-limiting covariance has simple spectrum~\cite{slepian1978} and $T$ is an irreducible Jacobi matrix commuting with $\R$, so Theorem~\ref{thm:tridiag-unique} forces the tridiagonal commutant to equal $\spn\{\mathbf{I},T\}$ exactly. The pencil therefore recovers $T$ rather than merely minimizing the residual: projecting Slepian's matrix~\eqref{eq:slepianT} onto the computed kernel leaves a relative residual of $2\times10^{-11}$, confirming the nontrivial null vector is $T$ up to scale. Its spectrum is simple, so by Proposition~\ref{prop:prolate}(iii) its eigenvectors resolve the full DPSS basis: forming $V_T^\top\R V_T$ leaves off-diagonal mass below $10^{-14}$, confirming that the eigenvectors of the recovered generator diagonalize $\R$. The matched generator $T$ has $\delta(T,\R)\approx2\times10^{-17}$, machine zero. Figure~\ref{fig:slepian} shows the band-limiting covariance, the recovered tridiagonal generator $T$, the leading discrete prolate spheroidal sequences, and the pencil residual spectrum.

\begin{figure}[H]
\centering
\includegraphics[width=0.96\textwidth]{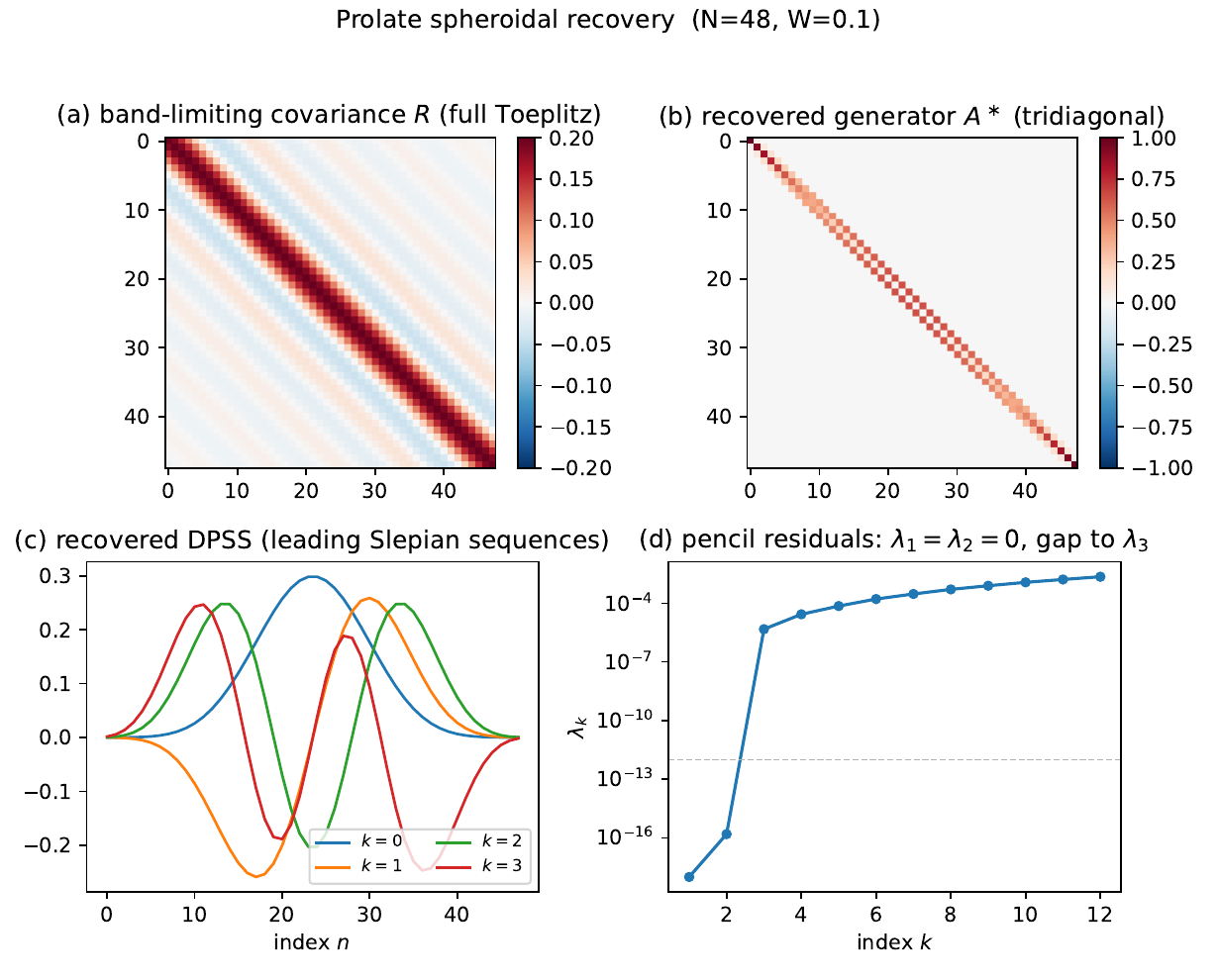}
\caption{Recovery of the prolate spheroidal (Slepian) transform by the double-commutator pencil ($N=48$, $W=0.1$). (a) The band-limiting covariance $\R$, a full Toeplitz matrix. (b) The nontrivial matched generator returned by the pencil over symmetric tridiagonal candidates: Slepian's tridiagonal $T$. (c) The leading discrete prolate spheroidal sequences, the eigenvectors of the recovered generator. (d) The pencil residual spectrum, with $\lambda_1=\lambda_2=0$ spanning the kernel $\protect\spn\{\mathbf{I},T\}$ and a gap of ten orders of magnitude to $\lambda_3$, the prolate property in the form $\lambda_2=0$.}
\label{fig:slepian}
\end{figure}

The contrast is the point of the example. The classical symmetry that inspection or a permutation search reveals, centrosymmetry, yields only a two-block reduction. The transform itself, the full DPSS resolution, comes from a generator that is not a symmetry of $\R$ in any geometric sense and that historically had to be guessed. Searching a structured generator space rather than a permutation group, the matched-generator pencil recovers it as the second null vector of $\ad_\R^2$, with no input beyond $\R$ and the choice of candidate structure.

\subsection{The Hahn transform}\label{sec:hahn}

The sinc recovery is not isolated. The same exact, finite-dimensional phenomenon holds for the discrete orthogonal-polynomial ensembles, the discrete bispectral families whose three-term recurrence supplies an exact commuting tridiagonal operator~\cite{grunbaumvinetzhedanov2018}. We give the Hahn ensemble as a worked case.

Fix $\alpha=\beta=2$ and let $w(x)=\binom{\alpha+x}{x}\binom{\beta+N-1-x}{N-1-x}$ be the symmetric Hahn weight on $x=0,\ldots,N-1$. The orthonormal Hahn polynomials are the eigenvectors of the Hahn difference operator $L$, a symmetric tridiagonal matrix with eigenvalues $\theta_n=n(n+\alpha+\beta+1)$~\cite{koekoek2010}; in self-adjoint form $L$ generates a reversible birth-and-death chain whose stationary law is $w$. We take as covariance the heat operator of that generator,
\begin{equation}\label{eq:hahnR}
\R=\exp(-tL),\qquad t>0,
\end{equation}
a full-rank positive-definite matrix whose modes are the Hahn polynomials and whose spectrum $\exp(-t\theta_n)$ is smooth and simple. The smooth, simple spectrum is what makes the recovery exact. The eigenvalues $\exp(-t\theta_n)$ are distinct, so $\R$ has simple spectrum and Theorem~\ref{thm:tridiag-unique} applies with the irreducible Jacobi matrix $L$, pinning the tridiagonal commutant to $\spn\{\mathbf{I},L\}$; a rank-deficient concentration operator, by contrast, would have a degenerate spectrum and admit many commuting tridiagonals.

The classical symmetry of $\R$ is again coarse. The symmetric weight makes $\R$ centrosymmetric, $\mathbf{J}\R\mathbf{J}=\R$, so $\delta(\mathbf{J},\R)=0$; but $\mathbf{J}^2=\mathbf{I}$ resolves $\R$ only into even and odd subspaces. Solving the pencil $\ad_\R^2(A)=\lambda A$ over the symmetric tridiagonal matrices recovers the full transform. With $N=40$ and $t=0.03$, the zero-residual subspace is two-dimensional, $\lambda_1=\lambda_2=0$ (both below $10^{-16}$) with a gap of twelve orders of magnitude to $\lambda_3\approx5\times10^{-6}$; the kernel is $\spn\{\mathbf{I},L\}$ as Theorem~\ref{thm:tridiag-unique} guarantees, the nontrivial generator $L$ is recovered to relative residual $3\times10^{-11}$, and $\delta(L,\R)\approx6\times10^{-17}$. The generator has simple spectrum, so by Proposition~\ref{prop:prolate}(iii) its eigenvectors are the Hahn-polynomial transform, and they diagonalize $\R$ to machine precision. Figure~\ref{fig:hahn} shows the covariance, the recovered tridiagonal generator, the leading Hahn modes, and the pencil spectrum.

\begin{figure}[H]
\centering
\includegraphics[width=0.96\textwidth]{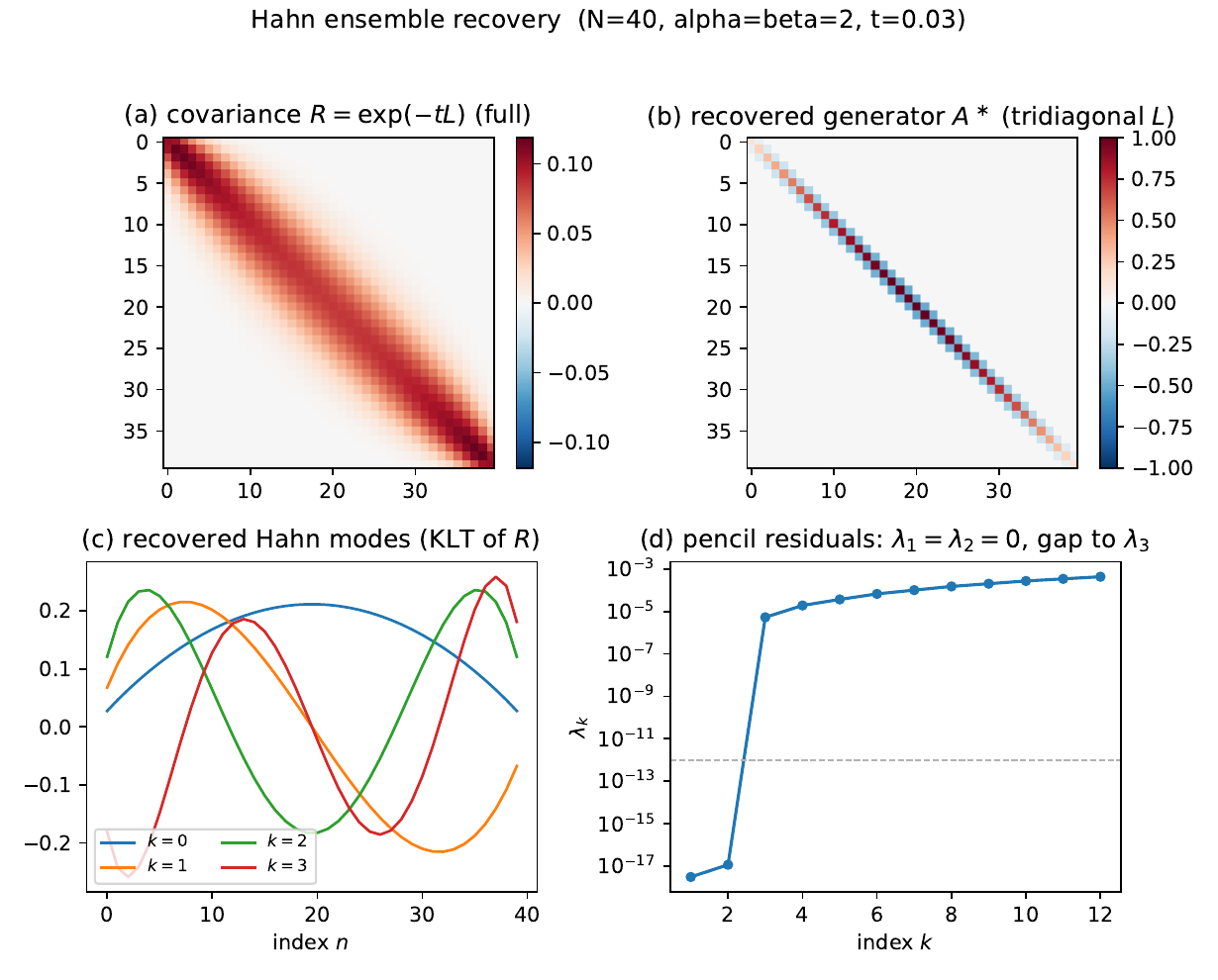}
\caption{Exact recovery of the Hahn transform by the double-commutator pencil ($N=40$, $\alpha=\beta=2$, $\R=\exp(-tL)$ with $t=0.03$). (a) The covariance $\R$, a full matrix. (b) The nontrivial matched generator returned by the pencil over symmetric tridiagonal candidates: the Hahn difference operator $L$. (c) The leading Hahn modes, the eigenvectors of the recovered generator and the Karhunen--Lo\`eve transform of $\R$. (d) The pencil residual spectrum, with $\lambda_1=\lambda_2=0$ spanning the kernel $\protect\spn\{\mathbf{I},L\}$ and a gap of twelve orders of magnitude to $\lambda_3$.}
\label{fig:hahn}
\end{figure}

The construction is not particular to Hahn. The dual Hahn, Krawtchouk, and Racah ensembles carry the same exact commuting tridiagonal by the same mechanism~\cite{grunbaumvinetzhedanov2018}, with Hahn shown here. The contrast with the preceding section is the point: where the Airy and Bessel kernels recover only in the continuum, the discrete ensembles recover exactly in finite dimensions, of the same character as the sinc case. The prolate property is recovered by the pencil across a family of transforms, not at a single point.

\subsection{The discrete cosine (DCT) basis}\label{sec:dct}

The same mechanism recovers the discrete cosine transform, the example that motivates Section~\ref{sec:ckl}. There the mixed kernel was persymmetric, so the reflection was its exact matched generator, but the reflection squares to the identity and yields only the even/odd parity split, not a sinusoidal transform; the cosine basis itself requires the stronger even-stationary structure. That structure is a Jacobi commutant. The DCT-II vectors $u_k[n]=c_k\cos\!\bigl(\pi(n+\tfrac12)k/N\bigr)$ are exactly the eigenvectors of the \emph{cosine generator} $C$, the symmetric tridiagonal Neumann second-difference matrix with diagonal $(1,2,\ldots,2,1)$ and off-diagonals $-1$, whose eigenvalues are $2-2\cos(k\pi/N)$; this is the sinusoidal-family characterization of Jain~\cite{jain1979}, and $C$ is an irreducible Jacobi matrix. Taking the even-stationary covariance
\begin{equation}\label{eq:dctR}
\R=\exp(-tC),\qquad t>0,
\end{equation}
the heat operator of $C$, equivalently the covariance of a reflecting first-order Markov process whose precision matrix is $C$, gives a full-rank positive-definite matrix whose spectrum $\exp(-t(2-2\cos(k\pi/N)))$ is simple, so Theorem~\ref{thm:jacobi-recovery} applies with the Jacobi matrix $C$.

The classical symmetry is again coarse. The covariance is centrosymmetric, $\mathbf{J}\R\mathbf{J}=\R$, so $\delta(\mathbf{J},\R)=0$; but $\mathbf{J}^2=\mathbf{I}$ resolves $\R$ only into its even and odd subspaces, the parity split, not the cosine basis. Solving the pencil $\ad_\R^2(A)=\lambda A$ over the symmetric tridiagonal matrices recovers the full transform. With $N=40$ and $t=0.3$, the zero-residual subspace is two-dimensional, $\lambda_1=\lambda_2=0$ (both below $10^{-16}$) with a gap of thirteen orders of magnitude to $\lambda_3\approx1.4\times10^{-4}$; the kernel is $\spn\{\mathbf{I},C\}$ as Theorem~\ref{thm:jacobi-recovery} guarantees, the cosine generator $C$ is recovered to relative residual $4\times10^{-13}$, and $\delta(C,\R)\approx4\times10^{-16}$. Its spectrum is simple, so its eigenvectors are the DCT-II basis: they agree with the analytic cosines to alignment $1.000000$ and diagonalize $\R$ to machine precision. Figure~\ref{fig:dct} shows the covariance, the recovered cosine generator, the leading DCT modes, and the pencil residual spectrum.

\begin{figure}[H]
\centering
\includegraphics[width=0.96\textwidth]{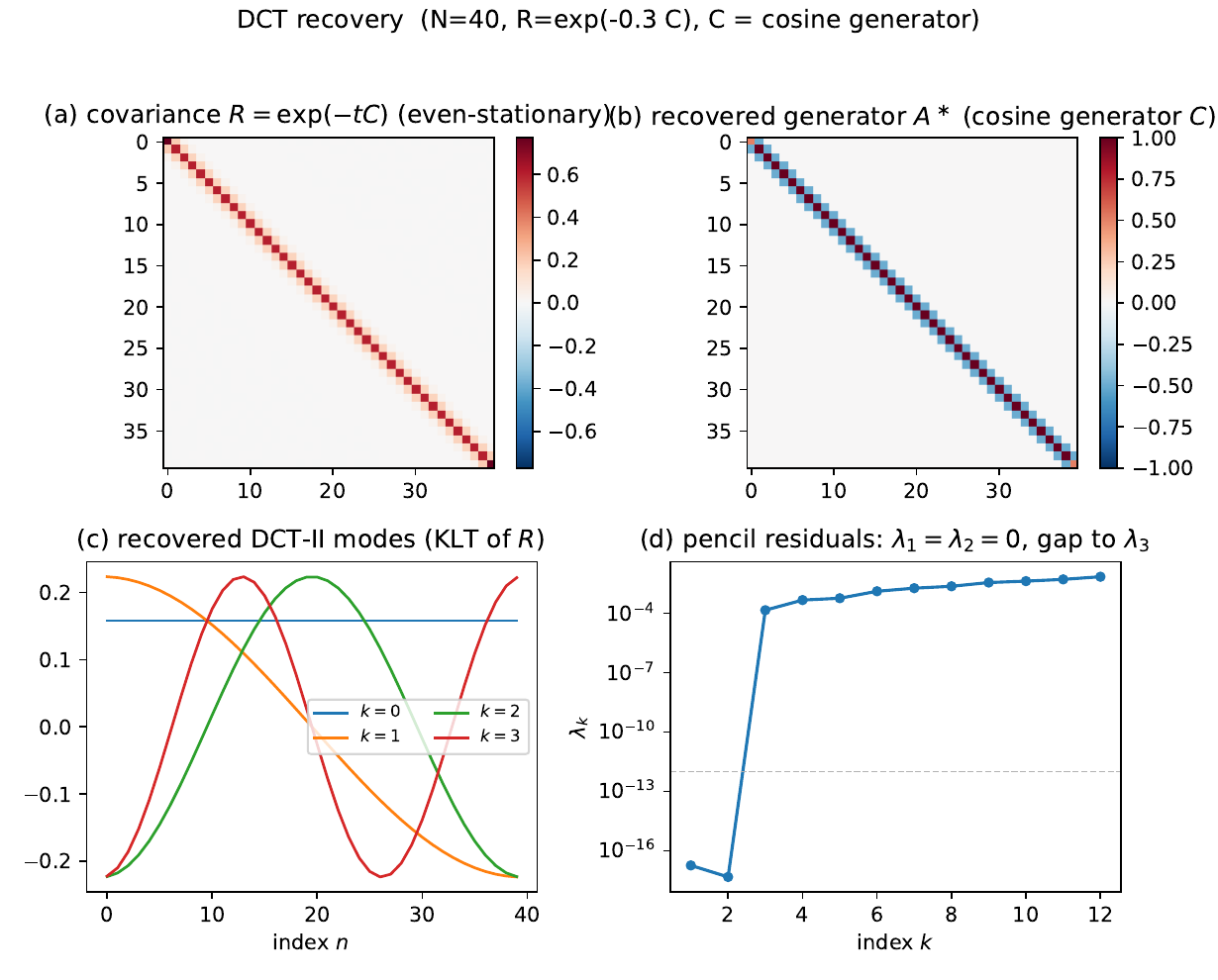}
\caption{Exact recovery of the discrete cosine transform (DCT-II) by the double-commutator pencil ($N=40$, $\R=\exp(-tC)$ with $t=0.3$). (a) The even-stationary covariance $\R$. (b) The nontrivial matched generator returned by the pencil over symmetric tridiagonal candidates: the cosine generator $C$, the Neumann second-difference matrix. (c) The leading DCT-II modes, the eigenvectors of the recovered generator and the Karhunen--Lo\`eve transform of $\R$. (d) The pencil residual spectrum, with $\lambda_1=\lambda_2=0$ spanning the kernel $\protect\spn\{\mathbf{I},C\}$ and a gap of thirteen orders of magnitude to $\lambda_3$.}
\label{fig:dct}
\end{figure}

This closes the contrast that opened the paper. For an even-stationary process the reflection is an exact but coarse symmetry, giving only the parity split, while the hidden tridiagonal cosine generator gives the full DCT; recovering the cosine basis genuinely requires the even-stationary, Jacobi-commutant structure that the persymmetric-but-nonstationary kernel of Section~\ref{sec:ckl} lacks. The discrete cosine transform joins the prolate spheroidal and Hahn bases as an exact instance of Theorem~\ref{thm:jacobi-recovery}.

\subsection{The Haar wavelet basis}\label{sec:wavelet}

The Slepian and Hahn recoveries both live in the bispectral world: the hidden generator is tridiagonal, the matrix form of a three-term recurrence, and the recovered transform is a Sturm--Liouville eigenbasis. The matched-generator principle is not confined to that world. We give a third example, in which the hidden generator is \emph{hierarchical} rather than tridiagonal and the recovered transform is the Haar wavelet basis, the prototype of multiresolution analysis~\cite{haar1910,mallat1989,daubechies1992}. Unlike the two preceding cases this one is a numerical demonstration rather than an exact-recovery theorem: the ultrametric-to-Haar diagonalization it rests on is classical, but we do not establish a commutant-uniqueness result for the hierarchical generator class analogous to Theorem~\ref{thm:tridiag-unique}, and the two-dimensionality of the pencil kernel over this class is observed numerically rather than proved.

Index $N=2^L$ points by the leaves of a balanced binary tree. For $1\le j\le L$ let $S_j$ be the scale-$j$ coupling operator, the symmetric matrix with $(S_j)_{ik}=1$ when the paths from the root to leaves $i$ and $k$ first diverge at level $j$, and zero otherwise. Every combination $\sum_j w_j S_j$ is an ultrametric matrix, and ultrametric matrices are diagonalized by the Haar wavelets of the tree; this is the matrix form of the identification of wavelet analysis with $p$-adic spectral analysis~\cite{kozyrev2002}. A pure ultrametric therefore fixes the dyadic scale eigenspaces but, because each scale is highly degenerate, not a unique basis within them. Adding a small position gradient $\delta\mathbf{X}$, with $\mathbf{X}=\operatorname{diag}(n-\tfrac{N-1}{2})$, breaks the within-scale degeneracy and resolves a unique multiscale basis, exactly the mechanism of Proposition~\ref{prop:prolate}(iii). Take
\begin{equation}\label{eq:waveletA}
A=\sum_{j=1}^{L} w_j S_j+\delta\mathbf{X},\qquad w_j=\rho^{\,j},\qquad \R=\exp(-tA),
\end{equation}
a hierarchical covariance whose modes are the Haar wavelets.

Solving the pencil $\ad_\R^2(A)=\lambda A$ over the hierarchical class $\{\mathbf{I},\mathbf{X},S_1,\ldots,S_L\}$ recovers the generator numerically. With $N=64$, $\rho=0.7$, $\delta=2\times10^{-3}$, and $t=0.15$, the zero-residual subspace is two-dimensional, $\lambda_1=\lambda_2=0$ (below $10^{-18}$) with a large gap to $\lambda_3\approx1.4\times10^{-7}$; the kernel is $\spn\{\mathbf{I},A\}$, the nontrivial generator is recovered with $\delta(A,\R)\approx3\times10^{-16}$, and it has simple spectrum. By Proposition~\ref{prop:prolate}(iii) its eigenvectors are the recovered transform, and they coincide with the Haar wavelet basis to alignment $0.99997$; the residual is the infinitesimal within-scale ordering imposed by $\delta\mathbf{X}$, which can be made arbitrarily small. The recovered basis carries the multiscale signature, sparse and localized across a range of scales, and it diagonalizes $\R$ to machine precision. Figure~\ref{fig:wavelet} shows the hierarchical covariance, the recovered ultrametric generator, the leading Haar modes, and the pencil spectrum.

\begin{figure}[H]
\centering
\includegraphics[width=0.96\textwidth]{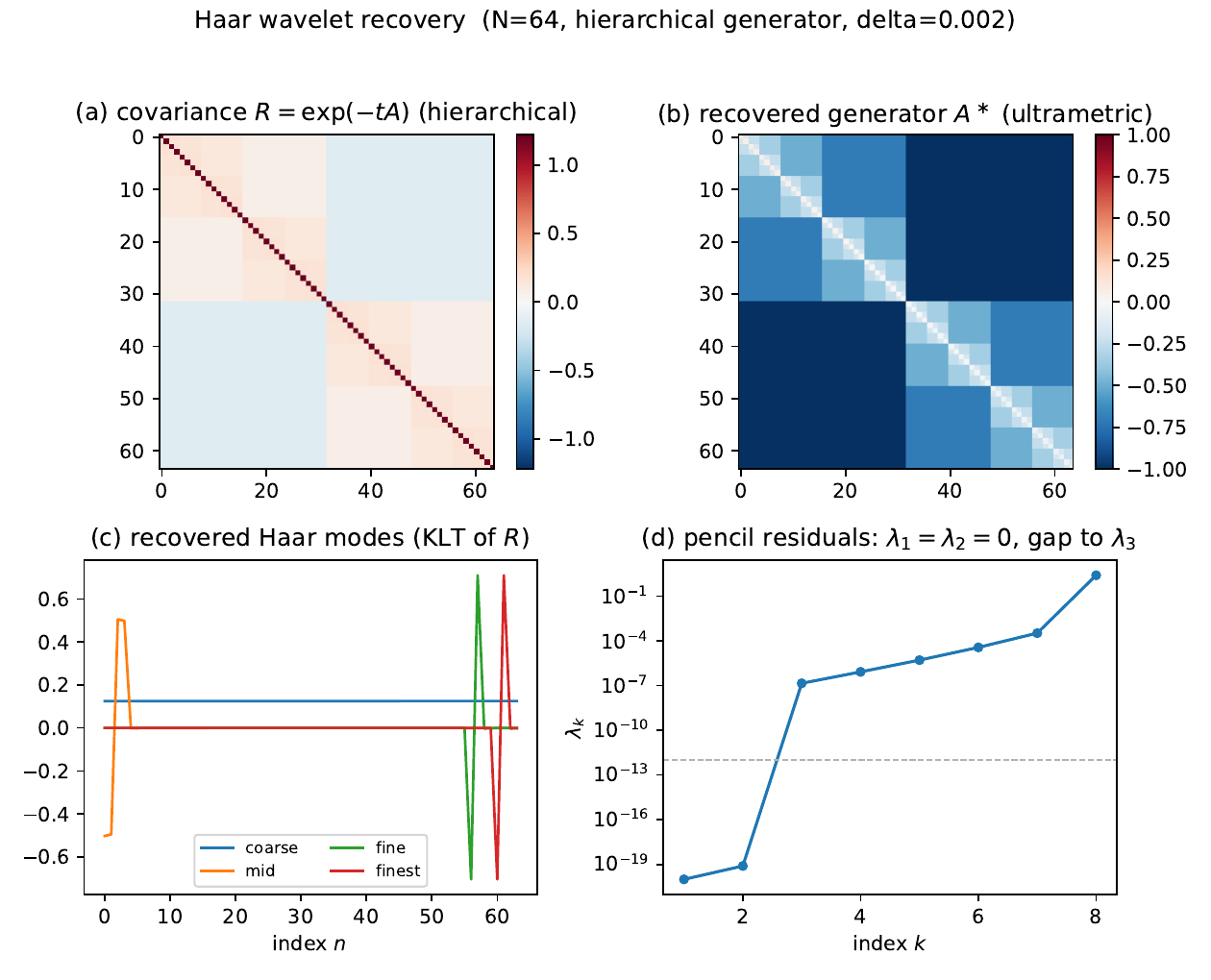}
\caption{Numerical recovery of the Haar wavelet basis by the double-commutator pencil ($N=64$, hierarchical generator $A=\sum_j 0.7^{\,j}S_j+\delta\mathbf{X}$ with $\delta=2\times10^{-3}$, $\R=\exp(-tA)$, $t=0.15$). (a) The hierarchical covariance $\R$, with nested dyadic block structure. (b) The nontrivial matched generator returned by the pencil over the hierarchical class: an ultrametric matrix on the dyadic tree. (c) The leading recovered modes, the eigenvectors of the generator and the Karhunen--Lo\`eve transform of $\R$, Haar wavelets localized across a range of scales. (d) The pencil residual spectrum, with $\lambda_1=\lambda_2=0$ spanning the kernel $\protect\spn\{\mathbf{I},A\}$ and a large gap to $\lambda_3$.}
\label{fig:wavelet}
\end{figure}

Unlike the preceding two examples, the generator here is not tridiagonal and the transform is not a Sturm--Liouville eigenbasis. The demonstration suggests that the matched-generator principle reaches beyond the bispectral, three-term-recurrence family: alongside the Fourier (cyclic), cosine (reflection), prolate (Slepian tridiagonal), and Hahn (orthogonal-polynomial) transforms, the Haar wavelet basis is recovered by the same pencil and the same degeneracy-resolution mechanism, from a hierarchical rather than a local generator. A rigorous characterization of the hierarchical generator class, identifying which covariances make its commutant exactly two-dimensional, would promote this demonstration to a recovery theorem; we leave it open. The multiresolution structure that makes wavelets computationally fast is exactly the structure the matched generator supplies.

\subsection{Positioning}

The commuting tridiagonal operator is classical~\cite{slepian1978}, and the prolate spheroidal property is studied in depth in the bispectral framework, where commuting integral and differential pairs are organized by the adelic Grassmannian and were, before that theory, built case by case~\cite{duistermaat1986,casper2023}. We claim none of that. The contribution here is a lens: the matched-generator eigenvalue problem recovers the commuting operator as the minimal solution of a variational problem over a structured generator space, and Proposition~\ref{prop:prolate} turns the existence of a hidden symmetry into the spectral condition $\lambda_2=0$ on the pencil. This makes locating the commuting operator a computation rather than a construction, and it suggests a systematic search for prolate-type reductions in other operators by varying the candidate structure, a direction the Bessel and Airy kernels of random matrix theory make natural.

\paragraph{The bispectral triple.} The sinc kernel is one of three basic instances of the same phenomenon. Rescaling the Gaussian unitary ensemble at the soft edge of its spectrum produces the Airy kernel and at a hard edge the Bessel kernel, and both admit a commuting second-order operator just as the band-limiting operator does~\cite{tracywidom1994airy,tracywidom1994bessel}; the sinc, Airy, and Bessel functions are the basic rank-one and rank-two bispectral functions, and the commuting operator in each case is a consequence of bispectrality~\cite{grunbaumyakimov2004,casper2023}. The exact finite recovery above is particular to the sinc case, where Slepian's tridiagonal commutes with the band-limiting matrix to machine precision. The Airy and Bessel commuting operators are objects of the continuum: under discretization their commutation is only approximate, the matched residual decreasing with the grid rather than reaching the floor, and the rapid spectral decay of these trace-class kernels makes a direct generator search numerically delicate. The exact finite analogues of the prolate property for these families are instead the discrete orthogonal-polynomial ensembles, whose three-term recurrence supplies an exact commuting tridiagonal operator~\cite{grunbaumyakimov2004,grunbaumvinetzhedanov2018}, as the Hahn recovery of Section~\ref{sec:hahn} demonstrates.

\subsection{The matrix-valued setting}\label{sec:matrix}

The recoveries above act on scalar signals: the covariance is a scalar kernel and the recovered transform is a basis of scalar functions. The matched-generator principle is not restricted to that case. When the signal is vector valued, with a value in $\mathbb{C}^d$ at each point, the covariance is a $d\times d$ block kernel and the natural orthogonal systems are matrix valued. Deciding which weight matrices have orthogonal polynomials that are eigenfunctions of a second-order differential operator is the matrix Bochner problem, and unlike its scalar predecessor it is only partially resolved~\cite{duran2004,casperyakimov2022}. This is where the framework reaches genuinely new objects, and where its limits are also sharpest. We record one verified worked example and a deformation analysis of the reducible point.

\paragraph{A matrix-Hermite matched generator.} Let $a\in\mathbb{R}\setminus\{0\}$ and take the Hermite-type weight matrix
\begin{equation}\label{eq:matrixhermite}
W(x)=e^{-x^2}\begin{pmatrix} 1+a^2x^2 & ax \\ ax & 1\end{pmatrix}
      =e^{-x^2}\,e^{Ax}e^{A^{*}x},\qquad A=\begin{pmatrix}0&a\\0&0\end{pmatrix}.
\end{equation}
No constant invertible matrix block-diagonalizes $W$: the coupling $ax$ rotates the channel frame with $x$, so $W$ is irreducible. Its monic matrix orthogonal polynomials $P_n$, defined by $\langle P_n,P_m\rangle=\int P_n\,W\,P_m^{*}\,dx=\delta_{nm}H_n$, satisfy two eigenvalue relations at once. In the degree index they obey a matrix three-term recurrence $xP_n=P_{n+1}+B_nP_n+C_nP_{n-1}$ with matrix coefficients $B_n,C_n$, a block-tridiagonal difference operator. In the variable $x$ they are eigenfunctions of the second-order matrix differential operator
\begin{equation}\label{eq:matrixD}
\mathcal{D}=\partial_x^{2}\,\mathbf{I}
  +\partial_x\begin{pmatrix}-2x & 2a\\ 0 & -2x\end{pmatrix}
  +\begin{pmatrix}-2 & 0\\ 0 & 0\end{pmatrix},
\qquad \mathcal{D}\,P_n=\Lambda_n P_n,\quad \Lambda_n=\operatorname{diag}(-2n-2,\,-2n),
\end{equation}
acting with coefficients on the right of $P_n$~\cite{duran2004,casperyakimov2022}. The polynomials are therefore bispectral, eigenfunctions of a difference operator in $n$ and a differential operator in $x$, and $\mathcal{D}$ is the matrix analogue of the Slepian and Hermite operators of the scalar case: the matched generator whose eigenbasis is the matrix-Hermite transform of $\mathbb{C}^2$-valued signals. The operator is genuinely matrix valued. Its first-order coefficient is a non-normal Jordan block sharing no constant eigenbasis with the zeroth-order term, the two eigenvalue sequences $-2n-2$ and $-2n$ differ, and the corresponding orthonormal functions, the rows of $H_n^{-1/2}P_nW^{1/2}$, couple the two channels; the transform is not a pair of scalar transforms in a fixed frame. We reconstructed this example numerically inside the present framework as a check: building $P_n$ by the stable matrix recurrence (orthogonality to $5\times10^{-14}$), confirming \eqref{eq:matrixD} to $2\times10^{-12}$ across degrees, and confirming the matrix-Hermite functions form an orthonormal basis of $L^2(\mathbb{R},\mathbb{C}^2)$ to $6\times10^{-15}$. The weight and operator are not new: \eqref{eq:matrixhermite} is the Dur\'an--Gr\"unbaum Hermite-type weight and \eqref{eq:matrixD} the explicit operator recorded in the matrix Bochner analysis~\cite{duran2004,casperyakimov2022}. What the matched-generator lens adds is the reading of $\mathcal{D}$ as the structured generator selected by commutation with the covariance, exactly as in the scalar recoveries.

\paragraph{The reducible point: a deformation analysis.} A natural hope is to reach a matrix example by coupling two scalar ones. Fix two scalar band-limiting systems of distinct bandwidths $W_1\neq W_2$ on $\mathbb{C}^N$, with covariances $\R_1,\R_2$ and commuting Slepian tridiagonals $T_1,T_2$, so that $[\R_i,T_i]=0$. On $\mathbb{C}^2\otimes\mathbb{C}^N$ the direct sum
\begin{equation}\label{eq:matrix-reducible}
\R_0=\operatorname{diag}(\R_1,\R_2),\qquad T_0=\operatorname{diag}(T_1,T_2),\qquad [\R_0,T_0]=0
\end{equation}
is bispectral by construction. We analyze its concentration-preserving deformations, holding the scalar band-limiting blocks $\R_{11}\equiv\R_1$ and $\R_{22}\equiv\R_2$ fixed while an off-diagonal channel coupling turns on, and ask whether a genuine matrix prolate is reached by deforming this reducible point.

\begin{lemma}[First-order coupling tangent]\label{lem:matrix-tangent}
Assume the spectral disjointness
\begin{equation}\label{eq:spec-disjoint}
\operatorname{spec}(T_1)\cap\operatorname{spec}(T_2)=\varnothing,
\end{equation}
which is generic for distinct bandwidths $W_1\neq W_2$ and holds in the numerical cases considered here. Write a channel coupling as $\dot\R=E_{12}\otimes S+E_{21}\otimes S^{\top}$ and $\dot T=E_{12}\otimes\tau+E_{21}\otimes\tau^{\top}$, with $E_{ij}$ the $2\times2$ matrix units and $\tau$ tridiagonal. The commutation condition $[\dot\R,T_0]+[\R_0,\dot T]=0$ is equivalent to the Sylvester relation
\begin{equation}\label{eq:matrix-sylvester}
T_1 S - S T_2 = \R_1\tau-\tau \R_2 .
\end{equation}
Under \eqref{eq:spec-disjoint} the Sylvester map $S\mapsto T_1S-ST_2$ is invertible, so \eqref{eq:matrix-sylvester} determines $S$ uniquely from $\tau$ and the coupling tangent is parametrized by the tridiagonal $\tau$.
\end{lemma}

\begin{proof}
The $E_{12}$ block of $[\dot\R,T_0]+[\R_0,\dot T]$ is $ST_2-T_1S+\R_1\tau-\tau\R_2$ and the $E_{21}$ block is its transpose, giving~\eqref{eq:matrix-sylvester}. The eigenvalues of the Sylvester operator $S\mapsto T_1S-ST_2$ are the differences $\mu-\nu$ with $\mu\in\operatorname{spec}(T_1)$ and $\nu\in\operatorname{spec}(T_2)$; the disjointness~\eqref{eq:spec-disjoint} makes all of them nonzero, so the operator is invertible and $S$ is determined by $\tau$. For Slepian tridiagonals~\cite{slepian1978} of distinct bandwidth the spectra are generically disjoint, and disjointness is verified directly in the numerical instances.
\end{proof}

\begin{lemma}[Second-order obstruction in the tridiagonal commutant cokernel]\label{lem:matrix-coker}
Holding $\R_1,\R_2$ fixed, the within-channel part of the second-order commutation equation is
\begin{equation}\label{eq:matrix-secondorder}
[\R_1,\ddot T_{11}]=-2\,(S\tau^{\top}-\tau S^{\top}),\qquad \ddot T_{11}\ \text{symmetric tridiagonal},
\end{equation}
and likewise for channel $2$ with $S^{\top}\tau-\tau^{\top}S$ and $\R_2$. The map $\ddot T_{11}\mapsto[\R_1,\ddot T_{11}]$ on the symmetric tridiagonal matrices has kernel exactly $\operatorname{span}\{\mathbf{I},T_1\}$ by Theorem~\ref{thm:tridiag-unique}, so its image has dimension $2N-3$ and its cokernel, inside the antisymmetric matrices, has dimension $\binom{N}{2}-(2N-3)$. A first-order coupling extends to second order if and only if $S\tau^{\top}-\tau S^{\top}$ and $S^{\top}\tau-\tau^{\top}S$ project to zero on these two cokernels.
\end{lemma}

\begin{proof}
With $\R_{11}\equiv\R_1$ the term $[\ddot\R_{11},T_1]$ is absent and the quadratic term $[\dot\R,\dot T]$ lands within channel since $E_{12}E_{12}=0$, giving~\eqref{eq:matrix-secondorder}; both sides are antisymmetric. A symmetric tridiagonal commutes with $\R_1=f(T_1)$ if and only if it lies in the tridiagonal commutant of $T_1$, which is $\operatorname{span}\{\mathbf{I},T_1\}$ by Theorem~\ref{thm:tridiag-unique}; the dimension count follows from the symmetric tridiagonals having dimension $2N-1$ and the antisymmetric matrices dimension $\binom{N}{2}$.
\end{proof}

\begin{lemma}[Channel rotations are excluded]\label{lem:matrix-rotation}
The position-dependent channel rotations, generated by $\Xi=E_{12}\otimes\xi-E_{21}\otimes\xi^{\top}$ with $\xi$ diagonal, lie in the tangent of Lemma~\ref{lem:matrix-tangent} but are generically not concentration-preserving. With the convention $[A,B]=AB-BA$, the rotation flow $e^{\theta\Xi}\R_0\,e^{-\theta\Xi}$ shifts the first diagonal block at second order by $\tfrac12[\Xi,[\Xi,\R_0]]_{11}$, where
\begin{equation}\label{eq:matrix-rotshift}
[\Xi,[\Xi,\R_0]]_{11}=2\,\xi\R_2\xi^{\top}-\xi\xi^{\top}\R_1-\R_1\xi\xi^{\top},
\end{equation}
This term is generically nonzero, and in particular is nonzero for a scalar rotation $\xi=c\mathbf I$ with $c\neq0$ whenever $\R_1\neq\R_2$, where it reduces to $2c^2(\R_2-\R_1)$. A nonzero shift is incompatible with $\R_{11}\equiv\R_1$, so such a rotation is not concentration-preserving. The concentration-preserving class therefore carries no scalar channel-rotation gauge, and generically no channel-rotation gauge at all. This excludes one natural family of deformations; it is not on its own a rigidity statement, and the status of the remaining directions is taken up after the lemmas.
\end{lemma}

\begin{proof}
For $\xi$ diagonal the pair $(\xi\R_2-\R_1\xi,\ \xi T_2-T_1\xi)$ satisfies~\eqref{eq:matrix-sylvester} by the Jacobi identity together with $[\R_i,T_i]=0$, so $\Xi$ is in the tangent. Expanding $[\Xi,[\Xi,\R_0]]$ with $\R_0=\operatorname{diag}(\R_1,\R_2)$ and $\xi^{\top}=\xi$ gives~\eqref{eq:matrix-rotshift} for the $(1,1)$ block. For a scalar $\xi=c\mathbf I$ this reduces to $2c^2(\R_2-\R_1)$, nonzero whenever $\R_1\neq\R_2$; for general diagonal $\xi$ it is generically nonzero.
\end{proof}

These lemmas reduce the local question to the tridiagonal commutant cokernel of Theorem~\ref{thm:tridiag-unique} and exclude the channel rotations exactly, but they do not by themselves establish rigidity. The second-order obstruction is not bounded below on the rotation complement: minimizing it, rather than sampling, which in these dimensions misses the relevant directions, drives it to a small and size-dependent floor for $N\ge 7$ and to exactly zero for $N=6$, where the reducible point carries genuine second-order-flat coupling directions whose continuation is unobstructed through fourth order. Whether the reducible point is rigid for every bandwidth pair, and the classification of irreducible matrix transforms beyond the matrix-Hermite example, are not settled by this local analysis. Both are global: in the resolved cases of the matrix Bochner problem the irreducible solutions are noncommutative Darboux transforms of scalar direct sums by a non-constant intertwiner, the factor $e^{Ax}$ in~\eqref{eq:matrixhermite}, and the rigidity that would distinguish a genuine matrix prolate from a deformed reducible point is the obstruction to writing that intertwiner perturbatively~\cite{casperyakimov2022}. We leave both as future work.

\paragraph{Placement.} The two facts together locate the matrix-valued frontier. The matched-generator catalog, scalar throughout this paper, the Fourier, cosine, prolate, Hahn, and Haar bases, extends to genuinely irreducible matrix-valued transforms, of which the matrix Hermite is the cleanest worked example and $\mathcal{D}$ its matched generator. But these objects are isolated and intrinsically irreducible, reached by a non-constant intertwining of scalar systems, a global Darboux construction, and their classification, like the rigidity of the reducible point, is open. The framework carries over unchanged in form: the same variational reading of the generator and the same commutation residual now range over matrix-valued generator spaces. A genuinely new matrix bispectral family would require a weight outside the known Hermite, Laguerre, Jacobi, and Gegenbauer matrix classes, and the constructive check used here, a stable matrix recurrence paired with a verified operator solver, is a tool for that search.

\section{Interpolation Between Classical Transforms}\label{sec:interp}

Because the generator space $\mathcal{B}$ is continuous, one may form convex combinations of known generators and obtain transforms that lie between the classical ones. A parametric family of sinusoidal unitary transforms containing the Fourier, cosine, sine, and first-order Markov KLT as asymptotically equivalent members was given by Jain~\cite{jain1979}. The construction here differs in that the interpolation is driven by the covariance through a structured self-adjoint generator rather than by a fixed sinusoidal parameterization, so it is not restricted to stationary or sinusoidal structure.

\begin{definition}[Parameterized generator family]
For self-adjoint structured generators $A_1,A_2$ on $\mathbb{C}^M$, the interpolating generator is $A(\alpha)=(1-\alpha)A_1+\alpha A_2$, $\alpha\in[0,1]$. Each $A(\alpha)$ is again self-adjoint, and the transform at parameter $\alpha$ is its orthonormal eigenbasis. Here the generators act as spectral transforms through their eigenbases, not as infinitesimal generators of unitary flows $\exp(\theta A(\alpha))$; the examples below are the Neumann second-difference, the ultrametric generator, and their convex combinations, all self-adjoint rather than skew-Hermitian.
\end{definition}

At $\alpha=0$ and $\alpha=1$ the transform is the $A_1$- and $A_2$-transform respectively. A concrete instance takes $A_1$ the periodic second-difference Laplacian, whose orthonormal eigenbasis is the real Fourier (sine-cosine) system, and $A_2$ the Neumann second-difference Laplacian, whose eigenbasis is the DCT-II~\cite{rao1974}. Both are self-adjoint, so each $A(\alpha)$ has an orthonormal eigenbasis lying between the two. The periodic endpoint carries the two-fold frequency degeneracy of the spectrum $2-2\cos(2\pi k/M)$, so its eigenbasis is fixed only up to rotation within each frequency pair; any Neumann admixture ($\alpha>0$) lifts this degeneracy and selects a unique interpolating basis, the multiplicity-resolution mechanism of Proposition~\ref{prop:prolate}(iii). For a process whose covariance shares the symmetry of an interior generator $A(\alpha_0)$, the residual-minimizing $\alpha$ recovers $\alpha_0$, and the matched transform improves on both the real-Fourier and the DCT endpoint in rank-$K$ truncation energy (Figure~\ref{fig:interp}). For a single covariance the optimal generator is, by Theorem~\ref{thm:main}, the exact eigenvector of $\ad_\R^2$ restricted to the span $\{A_1,A_2\}$, so the interpolation parameter is not a free knob but the solution of the $2\times2$ pencil.

\begin{figure}[H]
\centering
\includegraphics[width=0.96\textwidth]{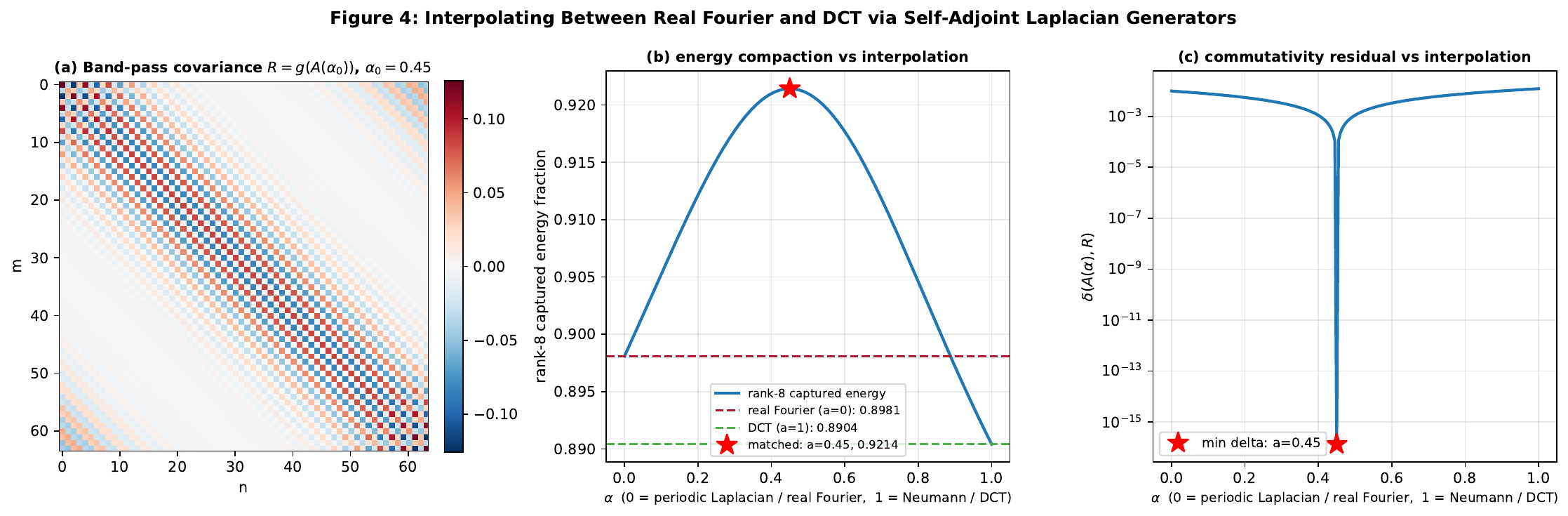}
\caption{Interpolation between the real Fourier and DCT bases via the self-adjoint Laplacian generator family of Section~\ref{sec:interp}. (a) A band-pass covariance whose symmetry generator is the interior mixed Laplacian $A(\alpha_0)$, $\alpha_0=0.45$. (b) The rank-$K$ truncation energy is maximized at the matched interior $\alpha$, improving on both the real-Fourier ($\alpha=0$) and DCT ($\alpha=1$) endpoints. (c) The commutativity residual $\delta(A(\alpha),\R)$ vanishes at $\alpha_0$. The periodic endpoint is spectrally degenerate; the Neumann admixture resolves a unique basis.}
\label{fig:interp}
\end{figure}

\subsection{An interpolant outside the sinusoidal family}\label{sec:interpjain}

The Fourier-to-cosine interpolant above stays within the sinusoidal world, and so does not yet exhibit the freedom asserted at the start of this section, since Jain's family already interpolates the sinusoidal transforms. To reach outside it, one interpolates a sinusoidal generator with a non-tridiagonal one. Let $A_1$ be the cosine generator, the Neumann second-difference, a symmetric tridiagonal operator whose eigenbasis is the DCT and hence a member of Jain's family, and let $A_2$ be the hierarchical ultrametric generator of Section~\ref{sec:wavelet}, whose eigenbasis is the Haar system and which is not tridiagonal. For the covariance $\R=\exp(-sG)$ with $G=\tfrac12 A_1+\tfrac12 A_2$ and $s=0.15$ on $M=64$ points, the two-generator pencil over $\spn\{A_1,A_2\}$ recovers the data-optimal interpolant exactly: $\alpha^\ast=0.500$ with $\delta(A^\ast,\R)=1.2\times10^{-16}$. The matched transform is a hybrid of the two, aligned $0.80$ with the DCT basis and $0.37$ with the Haar basis, and its generator carries $67\%$ of its Frobenius mass outside the tridiagonal band.

That the interpolant lies outside Jain's family admits an exact certificate through Theorem~\ref{thm:tridiag-unique}. Every member of the sinusoidal family is the eigenbasis of a symmetric tridiagonal (Jacobi) operator, so if no symmetric tridiagonal matrix other than the identity commutes with $\R$, the KLT of $\R$ is the eigenbasis of no tridiagonal operator and is therefore not a sinusoidal transform. For the mixed $\R$ the symmetric-tridiagonal commutant is one-dimensional, the identity alone; for a pure-cosine covariance it is two-dimensional, $\spn\{\mathbf{I},A_1\}$, exhibiting the Jain generator. The data-matched continuum thus passes through transforms that no fixed sinusoidal parameterization contains. This is a structural statement about the reach of the construction, not a claim of improved compaction. Figure~\ref{fig:interpjain} shows a matched mode of the interpolant against its nearest DCT and Haar modes, the residual $\delta(A(\alpha),\R)$ over the interpolation, and the tridiagonal-commutant certificate.

\begin{figure}[H]
\centering
\includegraphics[width=0.96\textwidth]{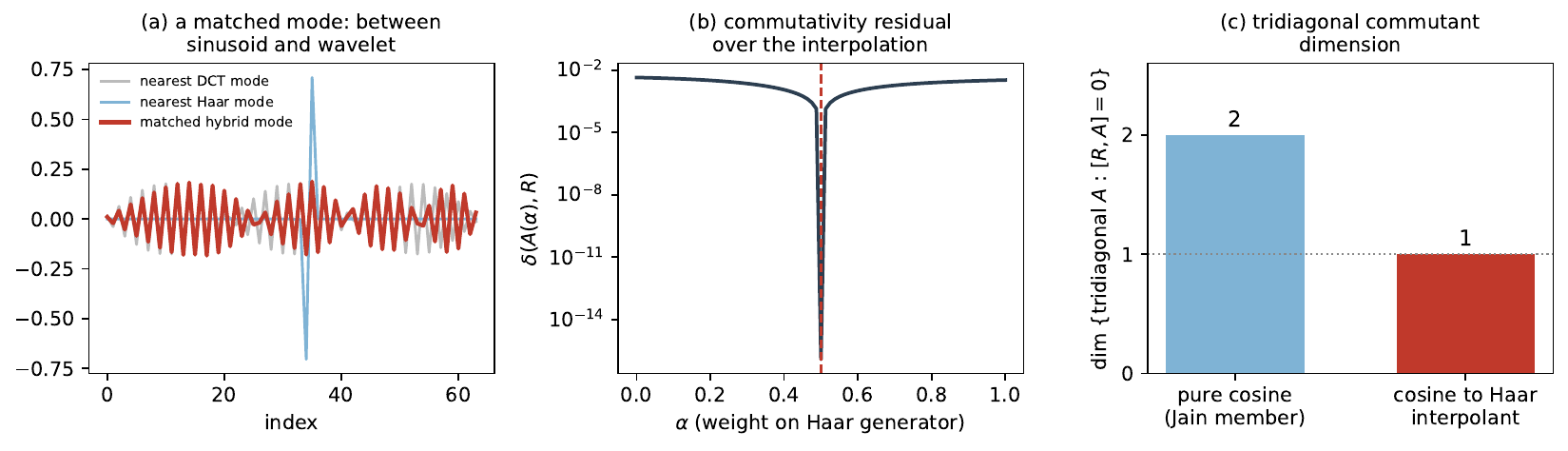}
\caption{An interpolated transform outside Jain's sinusoidal family, from the cosine generator (a Jain member) to the hierarchical Haar generator of Section~\ref{sec:wavelet}. (a) A matched mode of the data-optimal interpolant against its nearest DCT and Haar modes: it is neither. (b) The commutativity residual $\delta(A(\alpha),\R)$ over the interpolation, minimized at the recovered $\alpha^\ast$. (c) The certificate of Theorem~\ref{thm:tridiag-unique}: the symmetric-tridiagonal commutant of the mixed covariance has dimension one, so no tridiagonal operator commutes and the KLT is not sinusoidal, versus dimension two for a pure-cosine covariance, which exhibits its Jain generator.}
\label{fig:interpjain}
\end{figure}

\section{Application: Matched-Generator Karhunen--Lo\`eve Transforms}\label{sec:application}

The preceding theory is stated for an arbitrary covariance and an arbitrary candidate space.
This section specializes it to the case that arises most often in practice, a covariance built
from two known paradigms, and reports how the construction behaves on real signals. The setting
is the data-scarce regime: rather than estimate a full covariance and eigendecompose it, one
commits to known algebraic structure and estimates only the few coefficients that combine it. A
fuller signal-processing treatment, with the complete experimental protocol, is given in a
companion manuscript; the account here is condensed.

\subsection{The two-paradigm synthesis}\label{sec:twoparadigm}

Suppose the covariance is modeled as a convex combination of two paradigms with known matched
generators,
\begin{equation}\label{eq:mixture}
\R(\alpha)=\alpha\,\R_A+(1-\alpha)\,\R_B,\qquad \R_A=\exp(-tG_A),\quad \R_B=\exp(-tG_B),
\end{equation}
where $G_A$ and $G_B$ are the generators whose eigenbases are the two classical transforms, for
example the periodic second-difference $G_F$ (Fourier) and the Neumann second-difference $G_C$
(cosine) of Section~\ref{sec:commutant}. The mixing weight $\alpha$ and the diffusion time $t$ are
unknown. The matched transform is synthesized directly from the two generators, without forming or
eigendecomposing $\R(\alpha)$, by the double-commutator pencil of Theorem~\ref{thm:main} applied to
the two-dimensional candidate space $\mathcal{B}=\spn\{G_A,G_B\}$: the synthesized generator is
$A^\ast=c_1^\ast G_A+c_2^\ast G_B$, where $\mathbf{c}^\ast$ is the smallest-eigenvalue generalized
eigenvector of the $2\times2$ pencil $\mathbf{M}\mathbf{c}=\lambda\,\mathbf{G}\mathbf{c}$ with
$\mathbf{M}$ and $\mathbf{G}$ as in~\eqref{eq:MG}, and the synthesized transform is the eigenbasis of
$A^\ast$.

Three properties of the construction are inherited directly from the general theory. First, exact
recovery: if some generator in $\spn\{G_A,G_B\}$ commutes with $\R(\alpha)$, then by
Theorem~\ref{thm:lie-recovery} the pencil returns it and the synthesized transform is the KLT
exactly. Second, stability: the near-optimality gap is governed by the commutativity residual
$\delta(A^\ast,\R)$. A true symmetry has residual that is second order in the data perturbation
(Theorem~\ref{thm:quad}), and the subspace error of the recovered transform is controlled by the
residual through the Davis--Kahan bound of Theorem~\ref{thm:genstab}, so the high-resolution coding
penalty~\eqref{eq:hr-rate} of the synthesized transform over the KLT scales as $\delta^2$. Third,
sample efficiency: the synthesis estimates only the $d=2$ coefficients $\mathbf{c}$ rather than the
$M(M+1)/2$ free entries of a covariance, so by the sample-complexity corollary~\ref{cor:sample} it
reaches the population KLT from far fewer observations than the sample-covariance eigendecomposition
requires. This is the operating principle: algebraic structure substitutes for the data that a
covariance estimate would need.

\subsection{A worked case: Fourier and cosine}\label{sec:worked}

Instantiate~\eqref{eq:mixture} with the Fourier and cosine generators on $M=64$ points, $t=0.3$, and
sweep the mixing weight $\alpha$ across the unit interval. Figure~\ref{fig:appl-synthetic} reports
two figures of merit against the KLT computed from the full $\R(\alpha)$: the rank-$K$ energy
compaction ratio at $K=8$, and the high-resolution coding penalty~\eqref{eq:hr-rate}. The synthesized
transform holds a compaction ratio of one to six decimal places at every $\alpha$ and a coding
penalty that never exceeds $6\times10^{-8}$ bits per sample, that is, it is numerically the exact KLT
across the entire mixture. Each fixed endpoint, by contrast, is optimal only at its own corner and
degrades away from it. The matched generator commutes with $\R(\alpha)$ to a residual
$\delta(A^\ast,\R)\approx10^{-5}$ at interior $\alpha$, which is the quantitative reason the synthesis
is numerically exact, consistent with the $\delta^2$ scaling of the coding penalty established in
Section~\ref{sec:perturbation}.

\begin{figure}[H]
\centering
\includegraphics[width=0.92\textwidth]{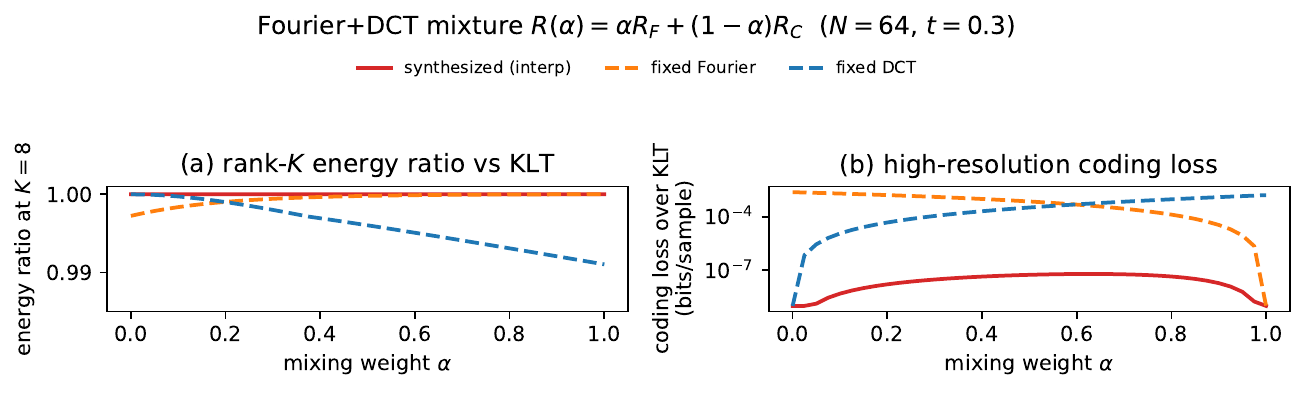}
\caption{Synthesized transform versus the KLT on a Fourier and cosine mixture
$\R(\alpha)=\alpha\R_F+(1-\alpha)\R_C$ at block length $M=64$ and $t=0.3$ (the panel label $N$ is the
block length $M$). (a) Rank-$K$ energy compaction at $K=8$: the synthesized transform holds one across
the mixture while each fixed endpoint degrades away from its corner. (b) High-resolution coding
penalty over the KLT (log scale): the synthesized transform stays below $10^{-7}$ bits per sample at
every weight, orders of magnitude below either fixed transform.}
\label{fig:appl-synthetic}
\end{figure}

\subsection{Sample efficiency on real signals}\label{sec:realsignals}

On real data the proportion of the two paradigms is unknown, the model~\eqref{eq:mixture} holds only
approximately, and the covariance must be estimated from a finite number of windows. Framing a signal
into length-$M$ windows, setting aside a held-out test set, and computing each transform from a
training set of $N_{\mathrm{tr}}$ windows swept from few to many, two questions separate. The first is
how few windows the synthesized transform needs to reach the compaction of the full-data KLT; the
second is whether the synthesized transform improves on the best fixed transform.

On the first question the answer is uniform and is the robust finding. Across four public signal
classes, an electrocardiogram record~\cite{moody2001mitbih}, a natural texture~\cite{vanderwalt2014scikit}, bearing-fault vibration~\cite{smith2015cwru}, and intracranial
EEG~\cite{fedele2021ieeg,fedele2017hfo}, the synthesized transform reaches the full-data KLT ceiling from five to twenty windows, where
the sample-covariance KLT overfits and needs hundreds. This is the data-scarcity benefit of
Section~\ref{sec:twoparadigm} made concrete: estimating two coefficients is stable on little data,
estimating a covariance is not.

On the second question the answer depends on the structure of the signal. The improvement over the best
fixed transform is a property of genuinely two-paradigm signals (Theorem~\ref{thm:criterion} and the
$\delta^2$ stability above): where a signal is close to a single paradigm, the best fixed transform is
already near the KLT and no improvement is available. Table~\ref{tab:realdata} summarizes the four
classes. The synthesized transform improves on the best fixed transform on exactly one of them, the
band-passed bearing faults, where band-passing removes a dominant single-paradigm component and
exposes a Fourier--cosine mixture in the resonance ringing; the gain is small but consistent across
three fault types. On the electrocardiogram and the texture the synthesized transform matches the best
fixed transform, and on the intracranial-EEG cohort it shows a statistically consistent but
operationally negligible margin of about $4\times10^{-4}$. Exposing two-paradigm structure by
band-pass filtering is itself delicate: too narrow a band drives the signal back toward a single
paradigm and removes the improvement.

\begin{table}[H]
\centering
\caption{Synthesized transform on four public signal classes. The data-scarcity benefit (reaching the
full-data KLT from few windows) is uniform; the improvement over the best fixed transform appears only
where the signal is genuinely a two-paradigm mixture.}
\label{tab:realdata}
\begin{tabular}{@{}llll@{}}
\toprule
Signal class & $M$ & $K$ & Synthesized versus best fixed transform \\
\midrule
Electrocardiogram (MIT--BIH)        & 128 & 32 & matches (single-paradigm signal) \\
Natural texture (grass)             & 64  & 32 & matches (single-paradigm signal) \\
Bearing fault, band-passed          & 128 & 32 & improves, small and consistent over three faults \\
Intracranial EEG (HFO)              & 256 & 64 & negligible margin ($\approx 4\times10^{-4}$) \\
\bottomrule
\end{tabular}
\end{table}

The principal use-case for the construction in this setting is the low-sample regime. When only a
small number of noisy observations are available, directly estimating a covariance and
eigendecomposing it approximates the Karhunen--Lo\`eve transform poorly; by constraining the solution
to a two-generator family derived from prior structural knowledge, the synthesis estimates only a few
structural degrees of freedom and often produces a substantially more accurate approximation to the
underlying transform. Where a signal additionally carries genuine two-paradigm structure, as in the
band-passed bearing faults, the same synthesis identifies an interpolated transform that improves on
the best fixed transform.

\subsection{Discovery from a candidate library}\label{sec:appl-discovery}

When the active pair is not known in advance but several candidate generators are, the same pencil of
Theorem~\ref{thm:main} over the larger candidate space discovers the combination whose transform most
nearly commutes with the data covariance, the finite-dimensional instance of the blind group matching
discussed in Section~\ref{sec:discussion}. One caution carries over from the few-window regime. A
candidate that achieves a small individual residual need not be the right one: on the bearing faults a
hierarchical (Haar) decoy attains the smallest individual residual at few windows, so the bare pencil
and a hard residual threshold both latch onto it and the recovered transform collapses. A soft
inverse-residual weighting, which combines the candidates in proportion to the reciprocal of their
residuals rather than selecting one, structurally never assigns zero weight to a true paradigm and
tracks the supplied pair from five windows. The advantage of inverse weighting is a property of the
few-dominant-paradigm regime: a signal mixing many mutually incoherent paradigms is near-white, and no
candidate-selection rule has anything to recover. The deterministic and statistical modes of the
blind-group-matching selector screen the decoy when their criteria are used in conjunction, abstaining
in exactly the few-window regime where a single residual is not a reliable identifier.

\section{Discussion}\label{sec:discussion}

The viewpoint of this paper is that the question ``which transform diagonalizes this covariance'' is, after the classical KLT reduction, a question about the symmetry commutant, and that the symmetry can be found by a single self-adjoint eigenvalue problem on a space of generators. The double-commutator theorem makes the search polynomial in the ambient dimension and exact at the commuting boundary. The numerical illustrations throughout are produced by scripts included with the manuscript and seeded for exact reproduction.

Several lines of prior work meet this one. The algebraic signal processing theory of P\"uschel and Moura~\cite{puschel2008} derives the Fourier and cosine/sine transforms as the Fourier transforms of polynomial-algebra modules, the algebraic counterpart of the commutant table of Section~\ref{sec:commutant}; that theory assumes exact symmetry, whereas the residual $\delta$ quantifies and minimizes departure from it. Joint approximate diagonalization seeks one basis that nearly diagonalizes a \emph{set} of matrices~\cite{cardoso1996} and underlies common principal component analysis~\cite{flury1984}; the present problem instead seeks the generator nearest to commuting with a \emph{single} operator, returning the matched symmetry rather than a joint eigenbasis. The optimality of the KLT for energy compaction (Theorem~\ref{thm:klt-opt}) is itself conditional on the coding setup~\cite{effros2004}. The blind group-discovery problem has also been approached through supervised, regular-representation methods requiring labeled data~\cite{hoyos2026groupalgebraictensorsprovablyoptimalequivariant}; the construction here is unsupervised, operates on the covariance commutant, and uses a single observation. The finite-dimensional instance of the present inverse problem is the blind group matching problem of algebraic-diversity signal processing, where a single covariance estimate is matched to the finite group whose representation best commutes with it; the combinatorial search over groups of order $M$ becomes the polynomial-time pencil of Section~\ref{sec:gevp} with the sequential deflation of Section~\ref{sec:sequential}. Its signal-processing development, including the single-observation variance reduction and the comparison against conventional multi-snapshot estimation, is carried out in companion work~\cite{thornton2026algebraicdiversitygrouptheoreticspectral,thornton2026continuousalgebraicdiversityunifying}. The asymptotic equivalence of the discrete Fourier transform and the KLT for stationary processes~\cite{gray2006} is the precise sense in which the cyclic row of the commutant table holds in finite dimension.

The operator at the centre of this paper, the double commutator $\ad_R^2$, is the same object that drives the double-bracket flow $\dot H=[H,[H,N]]$ of Brockett~\cite{brockett1991}, which evolves a matrix along a gradient flow on the orthogonal group toward diagonal form for a fixed target $N$, and which anchors the isospectral-flow view of numerical linear algebra surveyed by Chu~\cite{chu2008} (see also the projected-gradient flows under spectral constraints of Chu and Driessel~\cite{chudriessel1990} and the systematic treatment of Helmke and Moore~\cite{helmkemoore1994}). The use here is different in kind. Rather than flowing a matrix to diagonal form against a fixed $N$, we pose a static Rayleigh--Ritz problem for $\ad_R^2$ restricted to a fixed subspace of candidate generators, whose minimizer names the matched symmetry of a fixed $R$. The pencil~\eqref{eq:gevp} computes the low end of the $\ad_R^2$ spectrum that the equilibria of such a flow would approach, and repurposes it from diagonalization to group identification.

Open directions include the choice of generator basis for a given observation geometry, the extension of the sequential procedure of Section~\ref{sec:sequential} to continuous non-compact groups, where the deflation step must be formulated on the Lie algebra rather than on a finite permutation set, and a continuous-time realization of the matched-generator search as a gradient flow of the generalized Rayleigh quotient on the candidate algebra, which would connect the static pencil to the double-bracket and isospectral flows discussed above.

\appendix
\section{Sequential Generator Recovery on the Six-Cycle}\label{app:sequential}

We illustrate the sequential procedure of Section~\ref{sec:sequential} on the diffusion covariance $\R=(\mathbf{I}+\mathbf{L})^{-1}$ of the cycle $C_6$, for which $\Aut(\R)=D_6$ of order $12$ by Theorem~\ref{thm:aut}. The candidate basis is $\{\mathbf{P}_\tau-\mathbf{I},\,\mathbf{P}_{\tau^2}-\mathbf{I},\,\mathbf{P}_\rho-\mathbf{I},\,\mathbf{P}_\eta-\mathbf{I}\}$ with $\tau=(1\,2\,3\,4\,5\,6)$, $\tau^2=(1\,3\,5)(2\,4\,6)$, $\rho=(1\,6)(2\,5)(3\,4)$ a reflection, and $\eta=(1\,3)$ a transposition not in $\Aut(\R)$. Direct evaluation gives $\delta(\mathbf{P}_\sigma,\R)<10^{-15}$ for $\sigma\in\{\tau,\tau^2,\rho\}$ and $\delta(\mathbf{P}_\eta,\R)=0.119$.

Running the procedure at threshold $\tau=0$ (numerical tolerance $10^{-10}$) produces the trace

\begin{center}
\small
\begin{tabular}{ccccc}
\toprule
$k$ & $\lambda_{\min}$ & $\sigma^\ast_k$ & $\delta(\mathbf{P}_{\sigma^\ast_k},\R)$ & $|G_k|$ \\
\midrule
$0$ & $-$ & $e$ & $-$ & $1$ \\
$1$ & $5.0\times10^{-18}$ & $(1\,2\,3\,4\,5\,6)$ & $4.6\times10^{-17}$ & $6$ \\
$2$ & $6.9\times10^{-18}$ & $(1\,2\,3\,6\,5\,4)$ & $1.4\times10^{-1}$ & rejected \\
\bottomrule
\end{tabular}
\end{center}

At iteration $1$ the eigenvalue problem returns $\lambda_{\min}=0$ and rounding yields $\tau$, accepted, so $G_1=\langle\tau\rangle$ of order $6$. At iteration $2$ the deflated problem again returns $\lambda_{\min}=0$, at the residual of $\mathbf{P}_\rho-\mathbf{I}$, but Hungarian rounding of that residual lands on a permutation outside $\Aut(\R)$; the acceptance test rejects and the procedure terminates. The four named results of Section~\ref{sec:sequential} all hold on this trace: forward progress ($\tau\notin\{e\}$), strict growth ($|G_1|=6>1$), the iteration bound ($K=1\le\lceil\log_2 6\rceil=3$), and convergence ($\langle\tau\rangle\subseteq D_6$). The recovered subgroup is proper: rounding failed to recover $\rho$ from its deflation residual. The procedure is therefore sound, every accepted permutation is a genuine automorphism, but not complete: this basis is not generating-complete in the sense of Theorem~\ref{thm:full-recovery}, the rounded second residual missing $\rho$, so the procedure halts at the proper subgroup $\langle\tau\rangle\subsetneq D_6$ (Remark~\ref{rem:completeness}). Completeness depends on the interaction of the basis, the deflation, and the rounding step, an algorithmic question pursued in the companion development. The reproducibility script \texttt{sequential\_gevp\_C6.py} is included with the manuscript.

\bibliographystyle{plain}
\bibliography{klt_arxiv_refs}

\end{document}